\documentclass[prx,twocolumn,aps,superscriptaddress,longbibliography, nofootinbib, 10pt]{revtex4-2}
\usepackage{tikz}
\usetikzlibrary{decorations,decorations.pathmorphing,decorations.markings,calc,snakes,positioning,fixedpointarithmetic,shapes,shapes.geometric,patterns}

\usepackage{etoolbox}
\AtBeginEnvironment{tikzpicture}{\catcode`$=3 } 

\tikzset{alignmid/.style={baseline={([yshift=-.5ex]current bounding box.center)}}} 
\tikzset{every picture/.append style=alignmid}

\tikzset{
bottomzigzag/.style={postaction={draw,decorate, decoration={zigzag,amplitude=1pt,segment length=3pt,raise=1pt}}},
zigzag/.style={draw,decorate, decoration={zigzag,amplitude=1pt,segment length=3pt}},
rc/.style=rounded corners,
}

\tikzset{
    -|/.style={to path={-| (\tikztotarget)}},
    |-/.style={to path={|- (\tikztotarget)}},
}

\usepackage{ifthen}
\usepackage{xparse,pgffor}

\tikzset{
netmark/.code={
\tikzset{postaction={/network/mark/.cd,#1,/tikz/.cd,decorate,decoration={name=markings,mark=at position \netmarkpos with{
\begin{scope}[netmarktrafo]
\netmarkcode
\end{scope}
}}}}
\def\netmarkpos{0.5}
},
}

\def\netmarkpos{0.5}
\def\netmarkposoff{0cm}
\def\netmarkcode{}

\tikzset{
netmarktrafo/.style={},
netmarkstyle/.style={solid,semithick,sharp corners},
}

\pgfkeys{
/network/mark/.cd,
sty/.code={
\tikzset{netmarkstyle/.style={#1}}
},
asty/.code={
\tikzset{netmarkstyle/.append style={#1}}
},
f/.style={asty=fill},
p/.code={ 
\def\netmarkpos{#1}
},
poff/.code={ 
\def\netmarkposoff{#1cm}
},
e/.code={ 
\def\netmarkpos{\pgfdecoratedpathlength-0.005cm-\netmarkposoff}

\tikzset{netmarktrafo/.append style={shift={(-\netmarkwidth,0)}}}
},
s/.code={ 
\def\netmarkpos{0.005cm+\netmarkposoff}

\tikzset{netmarktrafo/.append style={shift={(\netmarkwidth,0)},xscale=-1,yscale=-1}}
},
a/.code={ 
\def\netmarkpos{\pgfdecoratedpathlength-0.005cm}

\tikzset{netmarktrafo/.append style={xscale=-1,shift={(-\netmarkwidth,0)}}}
},
b/.code={ 
\def\netmarkpos{0.005cm}

\tikzset{netmarktrafo/.append style={xscale=-1,shift={(\netmarkwidth,0),yscale=-1}}}
},
-/.code={ 
\tikzset{netmarktrafo/.append style={xscale=-1}}
},
r/.code={ 
\tikzset{netmarktrafo/.append style={yscale=-1}}
},
sideoff/.code={ 
\tikzset{netmarktrafo/.append style={shift={(0,#1)}}}
},
slab/.code={ 
\def\netmarkwidth{0}
\def\netmarkcode{
\node[inner sep=0.04cm,netmarkstyle,draw=none] (mylabelwidthtest) at (0,0){\phantom{#1}};
\path let \p1=(mylabelwidthtest.north east), \p2=(mylabelwidthtest.south east), \n1 = {max(abs(\y1),abs(\y2))} in node[inner sep=0.04cm,netmarkstyle] at (0,\n1) {#1};
}
},
lab/.code={ 
\def\netmarkwidth{0}
\def\netmarkcode{
\node[inner sep=0.04cm,anchor=\netmarkanchor] (mylabelwidthtest) at (0,0) {\phantom{#1}};
\draw[white] (mylabelwidthtest.\pgfdecoratedangle)--(mylabelwidthtest.\pgfdecoratedangle+180);
\node[inner sep=0.04cm,anchor=\netmarkanchor,netmarkstyle] at (0,0) {#1};
}
},
rotlab/.code={ 
\def\netmarkwidth{0}
\def\netmarkcode{
\node[inner sep=0.04cm,fill=white,transform shape,rotate=90,anchor=\netmarkrotanchor,netmarkstyle] (mydecorationnodename) at (0,0) {#1};
}
},
mrotlab/.style={ 
rotlab={$\scriptstyle{#1}$}
},
arr/.code={ 
\def\netmarkwidth{0.04}
\def\netmarkcode{\draw[netmarkstyle] (-0.04,0.08)--(0.04,0)--(-0.04,-0.08);}
},
arrbar/.code={ 
\def\netmarkwidth{0.08}
\def\netmarkcode{\draw[netmarkstyle] (-0.08,0.08)--(0,0)--(-0.08,-0.08) (0.04,0.08)--(0.04,-0.08);}
},
rarr/.code={ 
\def\netmarkwidth{0.04}
\def\netmarkcode{\draw[netmarkstyle] (-0.04,-0.08)arc(90-180:90:0.08);}
},
dot/.code={ 
\def\netmarkwidth{0.08}
\def\netmarkcode{\draw[netmarkstyle] (0,0)circle(0.08);}
},
flag/.code={ 
\def\netmarkwidth{0.06}
\def\netmarkcode{\draw[netmarkstyle] (-0.06,0)--(0,0.09)--(0.06,0)--cycle;}
},
darr/.code={ 
\def\netmarkwidth{0.08}
\def\netmarkcode{\draw[netmarkstyle] (-0.04,0)--(0.04,0)--(-0.04,0.08)--cycle;}
},
circ/.code={
\def\netmarkwidth{0.1}
\def\netmarkcode{\draw[netmarkstyle] (0,0) circle (0.1);}
},
hcirc/.code={ 
\def\netmarkwidth{0.1}
\def\netmarkcode{\draw[netmarkstyle] (-0.1,0) arc (180:0:0.1);}
},
diamond/.code={
\def\netmarkwidth{0.1}
\def\netmarkcode{\draw[netmarkstyle] (-0.1,0)--(0,-0.1)--(0.1,0)--(0,0.1)--cycle;}
},
bar/.code={ 
\def\netmarkwidth{0.05}
\def\netmarkcode{
\draw[netmarkstyle] (0,-0.08cm-0.5*\pgflinewidth)--(0,0.08cm+0.5*\pgflinewidth);
}
},
dbar/.code={ 
\def\netmarkwidth{0.13}
\def\netmarkcode{
\draw[netmarkstyle] (-0.04cm,-0.08cm-0.5*\pgflinewidth)--(-0.04cm,0.08cm+0.5*\pgflinewidth) (0.04cm,-0.08cm-0.5*\pgflinewidth)--(0.04cm,0.08cm+0.5*\pgflinewidth);
}
},
hbar/.code={ 
\def\netmarkwidth{0.05}
\def\netmarkcode{
\draw[netmarkstyle] (0, 0.5*\pgflinewidth)--++(0,0.12);
}
},
mill/.code={
\def\netmarkwidth{0.16}
\def\netmarkcode{
\draw[netmarkstyle] (0,-0.5*\pgflinewidth)--++(-0.08,-0.08)--++(0,0.08);
\draw[netmarkstyle] (0,0.5*\pgflinewidth)--++(0.08,0.08)--++(0,-0.08);
}
},
three dots/.code={
\def\netmarkwidth{0.2}
\def\netmarkcode{
\fill (-0.12,0) circle (0.5*0.05) (0,0) circle (0.5*0.05) (0.12,0) circle (0.5*0.05);
}
},
}

\tikzset{wid/.style={minimum width=#1cm}}
\tikzset{hei/.style={minimum height=#1cm}}

\tikzset{sx/.style={xshift=#1cm}}
\tikzset{sy/.style={yshift=#1cm}}

\tikzset{box/.style={draw,rectangle}}
\tikzset{fbox/.style={draw,rectangle, line width=1.1}}
\tikzset{roundbox/.style={draw,rectangle,rounded corners}}
\tikzset{froundbox/.style={draw,rectangle, rounded corners, line width=1.1}}
\tikzset{rounddiamond/.style={draw,diamond,rounded corners}}
\tikzset{dot/.style={draw, shape=circle, fill=black, scale=0.5}}

\tikzset{
netbox/.code={
\node[draw,netbdstyle] (\atomname) at (0,0) {#1};
\coordinate (\atomname-r) at (\atomname.east);
\coordinate (\atomname-l) at (\atomname.west);
\coordinate (\atomname-t) at (\atomname.north);
\coordinate (\atomname-b) at (\atomname.south);
\coordinate (\atomname-tr) at (\atomname.north east);
\coordinate (\atomname-br) at (\atomname.south east);
\coordinate (\atomname-tl) at (\atomname.north west);
\coordinate (\atomname-bl) at (\atomname.south west);
},
}

\tikzset{bdlw/.code={\tikzset{mybdstyle/.style={draw, line width=#1}}}}
\tikzset{bdcol/.code={\tikzset{mybdstyle/.append style={#1}}}}

\newcommand\atoms[2]{
\foreach \name/\keys in {#2}{
\expandafter\atom\expandafter{\keys,#1}{\name}
}
}

\newcommand\atom[2]{
\def\atomname{#2}
\tikzset{
nettrafo/.style={},
netatompos/.style={},
netdeco/.style={},
netpostdeco/.style={},
}

\pgfkeys{/network/atom/.cd,#1}

\begin{scope}[netatompos] 
\begin{scope}[nettrafo] 
\netshapecoords 
\fill[netbackstyle] \netshapepath;
\clip \netshapepath;
\tikzset{netdeco}
\draw[netbdstyle] \netshapepath;
\end{scope}
\tikzset{netpostdeco} 
\end{scope}

}

\pgfkeys{
/network/atom/.cd,
square/.code={

\def\netshapepath{(-\tempsize,-\tempsize)rectangle (\tempsize,\tempsize)}
\def\netshapecoords{
\node[rectangle,wid=2*\tempsize,hei=2*\tempsize,inner sep=0,transform shape](\atomname)at(0,0){};
\coordinate(\atomname-c) at (0,0);
\coordinate(\atomname-r) at (\tempsize,0);
\coordinate(\atomname-l) at (-\tempsize,0);
\coordinate(\atomname-t) at (0,\tempsize);
\coordinate(\atomname-b) at (0,-\tempsize);
\coordinate(\atomname-br) at (\tempsize,-\tempsize);
\coordinate(\atomname-tr) at (\tempsize,\tempsize);
\coordinate(\atomname-bl) at (-\tempsize,-\tempsize);
\coordinate(\atomname-tl) at (-\tempsize,\tempsize);
}},
circ/.code={

\def\netshapepath{(0,0)circle(\tempsize)}
\def\netshapecoords{
\node[circle,wid=2*\tempsize,hei=2*\tempsize,inner sep=0,transform shape](\atomname)at(0,0){};
\coordinate(\atomname-c) at (0,0);
\coordinate(\atomname-r) at (\tempsize,0);
\coordinate(\atomname-l) at (-\tempsize,0);
\coordinate(\atomname-t) at (0,\tempsize);
\coordinate(\atomname-b) at (0,-\tempsize);
}},
triang/.code={

\def\netshapepath{(-30:\tempsize)--(90:\tempsize)--(-150:\tempsize)--cycle}
\def\netshapecoords{
\node[regular polygon,regular polygon sides=3,wid=2*\tempsize,inner sep=0,transform shape](\atomname)at(0,0){};
\coordinate(\atomname-c) at (0,0);
\coordinate(\atomname-cr) at (-30:\tempsize);
\coordinate(\atomname-cl) at (-150:\tempsize);
\coordinate(\atomname-ct) at (90:\tempsize);
\coordinate(\atomname-mb) at (-90:0.5*\tempsize);
\coordinate(\atomname-mr) at (30:0.5*\tempsize);
\coordinate(\atomname-ml) at (150:0.5*\tempsize);
}},
diamond/.code={

\def\netshapepath{(0,-\tempsize)--(\tempsize,0)--(0,\tempsize)--(-\tempsize,0)--cycle}
\def\netshapecoords{
\node[rotate=45,rectangle,wid=sqrt(2)*\tempsize,hei=sqrt(2)*\tempsize,inner sep=0,transform shape](\atomname)at(0,0){};
\coordinate(\atomname-c) at (0,0);
\coordinate(\atomname-r) at (\tempsize,0);
\coordinate(\atomname-l) at (-\tempsize,0);
\coordinate(\atomname-t) at (0,\tempsize);
\coordinate(\atomname-b) at (0,-\tempsize);
}},
pentagon/.code={

\def\netshapepath{(-126:\tempsize)--(-54:\tempsize)--(18:\tempsize)--(90:\tempsize)--(162:\tempsize)--cycle}
\def\netshapecoords{
\node[regular polygon,regular polygon sides=5,wid=2*\tempsize,inner sep=0,transform shape](\atomname)at(0,0){};
\coordinate(\atomname-c) at (0,0);
\coordinate (\atomname-mb)at(-90:{\tempsize*cos(36)});
\coordinate (\atomname-mbr)at(-18:{\tempsize*cos(36)});
\coordinate (\atomname-mtr)at(54:{\tempsize*cos(36)});
\coordinate (\atomname-mtl)at(126:{\tempsize*cos(36)});
\coordinate (\atomname-mbl)at(-162:{\tempsize*cos(36)});
\coordinate (\atomname-cbr)at(-54:\tempsize);
\coordinate (\atomname-cr)at(18:\tempsize);
\coordinate (\atomname-ct)at(90:\tempsize);
\coordinate (\atomname-cl)at(162:\tempsize);
\coordinate (\atomname-cbl)at(-126:\tempsize);
}},
halfcirc/.code={

\def\netshapepath{(\tempsize,0)arc(0:180:\tempsize)--++(0,-0.04)-|cycle}
\def\netshapecoords{
\node[circle,wid=2*\tempsize,hei=2*\tempsize,inner sep=0,transform shape](\atomname)at(0,0){};
\coordinate(\atomname-c) at (0,0);
\coordinate(\atomname-r) at (\tempsize,0);
\coordinate(\atomname-l) at (-\tempsize,0);
\coordinate(\atomname-t) at (0,\tempsize);
\coordinate(\atomname-b) at (0,0);
}},
void/.code={
\def\netshapepath{}
\def\netshapecoords{
\coordinate(\atomname) at (0,0);
\coordinate(\atomname-c) at (0,0);
}},
labbox/.code={
\def\netshapepath{(0,0)}
\def\netshapecoords{}
\tikzset{netpostdeco/.append style={netbox=#1}}
},
}

\pgfkeys{
/network/atom/.cd,
p/.code={\tikzset{netatompos/.append style={shift={(#1)}}}},
xscale/.code={\tikzset{nettrafo/.append style={xscale=#1}}},
yscale/.code={\tikzset{nettrafo/.append style={yscale=#1}}},
scale/.code={\tikzset{nettrafo/.append style={scale=#1}}},
rot/.code={\tikzset{nettrafo/.append style={rotate=#1}}},
hflip/.style={xscale=-1},
vflip/.style={yscale=-1},
big/.style={scale=3/2},
small/.style={scale=2/3},
huge/.style={scale=9/4},
tiny/.style={scale=4/9},
flat/.style={yscale=2/3},
fflat/.style={yscale=4/9},
}

\tikzset{
netbdstyle/.style={line width=0.15em}, 
netdecstyle/.style={},
netpostdecstyle/.style={},
netbackstyle/.style={white},
}
\pgfkeys{
/network/atom/.cd,
bdstyle/.code={\tikzset{netbdstyle/.style={#1}}},
bdastyle/.code={\tikzset{netbdstyle/.append style={#1}}},
decstyle/.code={\tikzset{netdecstyle/.style={#1}}},
decastyle/.code={\tikzset{netdecstyle/.append style={#1}}},
postdecstyle/.code={\tikzset{netpostdecstyle/.style={#1}}},
postdecastyle/.code={\tikzset{netpostdecstyle/.append style={#1}}},
backstyle/.code={\tikzset{netbackstyle/.style={#1}}},
backastyle/.code={\tikzset{netbackstyle/.append style={#1}}},
style/.style={bdstyle=#1, decstyle=#1, postdecstyle=#1},
astyle/.style={bdastyle=#1, decastyle=#1, postdecastyle=#1},
whiteblack/.style={decstyle=white,backstyle=black},
nobd/.style={bdstyle={line width=0}},
}

\tikzset{
netbscope/.code={\begin{scope}[#1]},
netescope/.code={\end{scope}},
}

\pgfkeys{
/network/atom/.cd,
bscope/.code={\tikzset{netdeco/.append style={netbscope={#1}}}},
escope/.code={\tikzset{netdeco/.append style={netescope}}},
dec/.style 2 args={bscope={#1},#2,escope}, 
vbar/.style={dec={rotate=90}{hbar}},
dcross/.style={dec={rotate=45}{hbar},dec={rotate=-45}{hbar}},
cross/.style={hbar,vbar},
sect6/.style={sect=60},
sect8/.style={sect=45},
sect10/.style={sect=36},
}

\def\regdec#1{\pgfkeys{/network/atom/.cd,#1/.code={\tikzset{netdeco/.append style={net#1}}}}}
\regdec{all}\regdec{rhalf}\regdec{rquart}\regdec{brquart}\regdec{dot}\regdec{spiral}\regdec{swirl}\regdec{hstripe}\regdec{hbar}\regdec{rrey}
\pgfkeys{/network/atom/.cd,sect/.code={\tikzset{netdeco/.append style={netsect=#1}}}}

\tikzset{
netall/.code={\fill[netdecstyle] (-0.3,-0.3)rectangle (0.3,0.3);}, 
netrhalf/.code={\fill[netdecstyle] (0,-0.3)rectangle (0.3,0.3);}, 
netrquart/.code={\fill[netdecstyle] (0.075,-0.3)rectangle (0.3,0.3);}, 
netbrquart/.code={\fill[netdecstyle] (0,0)rectangle (0.3,-0.3);}, 
netsect/.code={\fill[netdecstyle] (0,0)--(0,-0.3)arc(-90:-90+#1:0.3)--cycle;}, 
netdot/.code={\fill[netdecstyle] (0,0)circle(0.07);}, 
netspiral/.code={\draw[netdecstyle] plot [variable=\t,domain=0:4] ({0.075*\t*cos(pi*(\t-0.5) r)},{0.075*\t*sin(pi*(\t-0.5) r)});}, 
netswirl/.code={\fill[netdecstyle] plot [variable=\t,domain=0:2] ({0.15*\t*cos(pi*(\t-0.5) r)},{0.15*\t*sin(pi*(\t-0.5) r)}) arc(-90:-450:0.3)--cycle;}, 
nethstripe/.code={\fill[netdecstyle] (-0.3,-0.05)rectangle(0.3,0.05);}, 
nethbar/.code={\draw[netdecstyle] (-0.3,0)--(0.3,0);}, 
netrrey/.code={\draw[netdecstyle] (0,0)--(0.3,0);} 
}

\pgfkeys{
/network/atom/.cd,
postbscope/.code={\tikzset{netpostdeco/.append style={netbscope={#1}}}},
postescope/.code={\tikzset{netpostdeco/.append style={netescope}}},
postdec/.style 2 args={postbscope={#1},#2,postescope}, 
arc/.code={\tikzset{netpostdeco/.append style={netarc=#1}}},
lab/.code={\tikzset{netpostdeco/.append style={netlab={#1}}}},
shadecirc/.code={\tikzset{netpostdeco/.append style=netshadecirc}},
shaderect/.code={\tikzset{netpostdeco/.append style={netshaderect={#1}}}},
debug/.code={\tikzset{netpostdeco/.append style=netdebug}},
markline/.code 2 args={\tikzset{netpostdeco/.append style={netmarkline={#1}{#2}}}},
postcirc/.code={\tikzset{netpostdeco/.append style=netpostcirc}},
}

\tikzset{
netlab/.code={
\pgfkeys{/network/atom/lab/.cd,#1}
\node[netpostdecstyle] at (\ifdefined\netlabpos\netlabpos\else\netlabang:\netlabdist\fi) {\netlabwrap{\netlabtext}};
},
netarc/.code args={#1:#2:#3}{
\draw[netpostdecstyle] (#1:#3) arc (#1:#2:#3);
},
netshadecirc/.code= {
\fill[opacity=0.4,netpostdecstyle] (0,0)circle(0.4);
},
netpostcirc/.code= {
\draw[netpostdecstyle] (0,0)circle(0.15);
},
netshaderect/.code= {
\fill[rc,opacity=0.4,netpostdecstyle] ($-1*(#1)$) rectangle (#1);
},
netdebug/.code= {
\node[red] at (0,0){\atomname};
},
netmarkline/.code 2 args= {
\draw (\atomname)edge[netmark={#2}]++(#1);
},
}

\def\netlabwrap#1{#1}
\pgfkeys{
/network/atom/lab/.cd,
t/.code={\def\netlabtext{#1}}, 
p/.code={\def\netlabpos{#1}}, 
ang/.code={\def\netlabang{#1}},
dist/.code={\def\netlabdist{#1}},
wrap/.code={\def\netlabwrap##1{#1}}, 
}

\usepackage{graphicx}
\usepackage[dvipsnames]{xcolor}
\usepackage{amsmath,amssymb,amsthm}
\usepackage{mathtools}
\usepackage{dsfont}
\usepackage[colorlinks=true,linkcolor=BrickRed,citecolor=Violet,urlcolor=Violet]{hyperref}
\usepackage{bm}      
\usepackage{times}
\usepackage{physics}

\tikzset{
ind/.style={netmark={lab=#1,a}}, 
}

\usepackage{tikz}
\usetikzlibrary{positioning, calc, shapes.geometric, plotmarks}
\usepackage{pgfplots}
\pgfplotsset{width=\linewidth,compat=1.18}

\definecolor{myred}{RGB}{255, 0, 0}
\definecolor{mygreen}{RGB}{0, 200, 30}
\definecolor{myblue}{RGB}{0, 100, 255}

\colorlet{global_color}{myblue}
\colorlet{JIT_x_color}{myred}
\colorlet{JIT_xz_color}{RedViolet}
\colorlet{JIT_xz_heu_color}{myred}
\colorlet{three_zero_D_color}{myblue}

\usepackage[ruled,vlined]{algorithm2e} 

\def\zz{\mathbb Z}

\theoremstyle{remark}
\newtheorem{defn}{Definition}
\newcommand\myparagraph[1]{\smallskip\noindent\textbf{#1. }}

\theoremstyle{remark}
\newtheorem{remark}{Remark}

\theoremstyle{plain}
\newtheorem{thm}{Theorem}

\newcommand{\bZ}{\mathbb{Z}}
\newcommand{\bR}{\mathbb{R}}

\newcommand{\sD}{\mathsf{D}}
\newcommand{\cH}{\mathcal{H}}
\newcommand{\one}{\mathds{1}}
\newcommand{\sDJIT}{\sD_{\mathrm{JIT}}}

\def\ttr{\texttt r}
\def\ttg{\texttt g}
\def\ttb{\texttt b}
\def\twod{\text{2D}}
\def\threed{\text{3D}}
\def\col{\mathcal C}

\DeclareMathOperator*{\argmin}{arg\,min}
\DeclareMathOperator*{\spn}{span}

\newcommand\mycube{
\path (0.8,0)coordinate(x) (0,0.8)coordinate(y) (0.5,0.35)coordinate(z);
\draw[dashed] (0,0)--++(z)--++(x) (z)--++(y);
}
\newcommand\mycubefront{
\draw (0,0)rectangle++($(x)+(y)$) (x)--++(z)--++(y)--++($-1*(x)$)--++($-1*(z)$) ($(x)+(y)$)--++(z);
}

\begin{document}

\title{High-threshold decoding of non-Pauli codes for 2D universality}
\author{Julio C. Magdalena de la Fuente}
\email{jm@juliomagdalena.de}
\author{Noa Feldman}
\affiliation{
Dahlem Center for Complex Quantum Systems, Freie Universit\"at Berlin, 14195 Berlin, Germany}
\author{Jens Eisert}
\affiliation{
Dahlem Center for Complex Quantum Systems, 
Freie Universit\"at Berlin, 14195 Berlin, Germany}
\affiliation{Helmholtz-Zentrum Berlin f{\"u}r Materialien und Energie, 14109 Berlin, Germany}
\author{Andreas Bauer}
\affiliation{
Department of Mechanical Engineering, Massachusetts Institute of Technology, Cambridge, MA 02139, USA}

\begin{abstract}
Topological codes have many desirable properties that allow fault-tolerant quantum computation with relatively low overhead. A core challenge for these codes, however, is to achieve a low-overhead universal gate set with limited connectivity. In this work, we explore a non-Pauli stabilizer code that can be used to complete a universal gate set on topological toric and surface codes in strictly two dimensions.
Fault-tolerant syndrome extraction for the non-Pauli code requires mid-circuit $X$ corrections, a key difference to conventional Pauli codes.
We construct and benchmark a just-in-time (JIT) matching decoder to reliably decide these corrections.
Under a phenomenological error model with equally likely physical and measurement errors, we find a high threshold of $\approx 2.5\,\%$, close to the $\approx 2.9\,\%$ of a decoder with access to the full syndrome history.
We also perform a finite-size scaling analysis to estimate how the logical error rate scales below threshold and verify an exponential suppression in both physical error rate and in the system size.
A second global decoding step for $Z$ errors is required and the non-Clifford gates in the circuit reduce the threshold from $\approx 2.9\,\%$ to $\approx 1.8\,\%$ with a naive decoder.
We show how $Z$ decoding can be improved using knowledge of the $X$ corrections, pushing the threshold to $\approx 2.2\,\%$.
Our results suggest non-Clifford logic in 2D codes could perform comparably to 2D quantum memory.
Our formalism for efficient benchmarking and decoding directly generalizes to a broader family of CSS codes whose $X$ stabilizers are twisted by diagonal Clifford operators, and spacetime versions thereof, defined by CSS-like circuits enriched by $CCZ$, $CS$, and $T$ gates.
\end{abstract}

\maketitle

\section{Introduction}

\myparagraph{Context}
In any realistic scenario, reliable large-scale quantum computing requires \emph{quantum error correction} (QEC)~\cite{shor1995QEC,SteaneCode,Shor1996, Roads, terhal2015review,MindTheGaps}.
Despite the presumably daunting overhead of active QEC, significant theoretical progress 
has been made in this field of research, and 
small-scale experimental demonstrations 
have 
just become available~\cite{Quantinuum2024demonstration, Daguerre2025Demonstration, Google2025Threshold, rosenfeld2025cultivation, Lukin2024QEC, Bluvstein2025architecture}.
To further scale a fault-tolerant computation with reasonable overheads, substantial challenges and bottlenecks remain.

To really scale a fault-tolerant quantum computing scheme, a universal and addressable set of logic gates that have a high error-correcting threshold, using an efficient decoder, is needed.
In topological codes all these aspects are well understood.
\emph{Two-dimensional} (2D) topological 
codes admit planar syndrome extraction circuits and a flexible implementation of logical Clifford gates that can be decoded with a high-threshold, scalable and efficient decoder~\cite{Dennis2002, fowler2012surface, Litinski2019game, Brown2017poking,Boundaries, Kubica2023colorcodedecoders}.
This makes them particularly promising candidates for the first generation of useful quantum computers.
To complete a universal gate set, non-Clifford operations are also needed.
These turn out to be not natively implementable in 2D Pauli codes~\cite{BravyiKoenig, pastawski2015nogologic, OConnor2018disjointness}, giving rise to substantial challenges.
Magic-state distillation is 
the leading approach to circumvent these obstructions~\cite{Bravyi2005MSD, Bravyi2012low-overheadMSD, Litinski2019notascostly, Gidney2019efficientmagicstate, gidney2019autoCCZ, Lee2025MSDColor}.
Despite a large effort in optimization throughout the last decade \cite{Litinski2019notascostly, Gidney2019efficientmagicstate}, the cost of running distillation circuits on logical qubits remains high. Other non-scalable schemes that rely on post-selection techniques perform surprisingly well but must likely be combined with distillation to prevent error floors while maintaining a reasonable overhead~\cite{Psi2024postselection, gidney2024cultivation}. Complementing these efforts along a different line of thought, 
recently, significant theoretical advances have been made to realize a fault-tolerant universal gate set on 2D topological codes in a more native way without distillation.
These schemes rely on encoding information in a non-Pauli stabilizer code during the computation.

\myparagraph{Universal QC in 2D without distillation}
Bombin presented a 2D scheme for universal fault tolerant quantum computation in Ref.~\cite{bombin20182d3d}.
The construction relies on a non-Clifford logic gate in a 3D code that is implemented by turning one of the spatial directions into a time direction~\cite{bombin20182d3d, Brown2020cczsurfacecode}.
Despite being a 2D scheme, both the large qubit overhead and the more complex decoding problem seemed to render these methods impractical~\cite{Scruby2022JITnumerics}.
The main reason for the more complicated decoding problem is that the schemes require many rounds of syndrome extraction, each of which includes non-Clifford operations.
Since corrections cannot be freely propagated through such a circuit in software, like in Clifford circuits, corrections must be applied actively while performing syndrome extraction.
These corrections can be chosen reliably in the presence of faulty measurements with a \emph{just-in-time (JIT) decoder} that only has access to limited syndrome information.
The performance of such a decoder seemed to be significantly reduced compared to global decoding schemes~\cite{Scruby2022JITnumerics}.

Recently, progress has been made to reduce the local physical qubit and circuit overhead by observing that within the non-Clifford circuit the quantum system evolves through a code defined by non-Pauli stabilizers and realizes non-Abelian topological order~\cite{davydova2025universal}.
Logical information can be transferred between these non-Abelian non-Pauli codes and conventional Pauli codes via simple space or time-like interfaces.
This perspective has led to a wealth of different types of non-Clifford logic gates for 2D surface and color codes, including unitary gates and protocols that prepare certain magic states, as well as low-overhead circuits that implement syndrome extraction within the non-Pauli code~\cite{davydova2025universal, bauer2025planar, huang2026hybridsurgery, warman2026clifford-hierarchy, kobayashi2025cliffordcodes, huang2025magicstates, manjunath2026GSC}.
While it was shown that a JIT decoder achieves a finite threshold against local stochastic noise in principle, the practical value of these new low-overhead schemes highly depend on the numerical value of the threshold and the below-threshold logical error-rate scaling.
Previous decoders showed a poor performance which lead to high resource estimates for previous schemes for universal computation in 2D~\cite{scruby2025nodistillation}.

\myparagraph{Results}
In this work, we show that the threshold can be significantly improved by using a better decoding strategy and indicate that also below-threshold scaling of a JIT decoder is comparable to a more standard decoder.
We devise an efficient JIT decoder that is matching based and optimized to the decoding graph at hand.
We observe that the threshold is on par with a global (standard) decoder that is given access to the full syndrome information at once.
We perform explicit numerical benchmarking of the memory error in a non-Pauli code that appears in the middle of several non-Clifford gates first presented in Ref.~\cite{davydova2025universal}.
These logical errors within the non-Pauli code are the dominant source of errors for the logic gate.
The noise model we consider is a phenomenological one with i.i.d.\ measurement errors and Pauli $X$ and $Z$ errors.
The threshold we find is an order of magnitude larger than previously found thresholds~\cite{Scruby2022JITnumerics}, comparable with the memory threshold of a toric code under the same noise model.
We believe that the comparability with the toric-code threshold will persist also for more realistic noise models.

The performance of the full protocol also depends on other types of errors that can be decoded globally.
In the circuits that we consider the JIT decoder is necessary to decode the $X$ errors in the protocol and it will leave (sparse) residual errors.
The non-Clifford gates in the circuit propagate these residual errors causing additional \emph{twisted} $Z$ errors.
This is analogous to the error propagation in the 3+0D dimensional jump protocols for non-Clifford gates~\cite{bombin2016dimensionaljump, beverland2021costuniversal, bombin2018error-propagation, Scruby2022nonPauli} and reduces the threshold of the second, global, decoding step.
We also numerically estimate by how much the error propagation due to the non-Clifford gate reduces the $Z$ threshold, and devise a new decoder that mitigates the effect of the propagation to a certain extent.
We again find that the $Z$ threshold remains comparable to the toric-code threshold under the same noise model.

We also lay the groundwork for analyzing error correction for a broad class of non-Clifford circuits in general.
While decoding and simulating error-correcting Clifford circuits and Pauli stabilizer codes is well understood \cite{Aaronson2004simulation, delfosse2023spacetime, Gidney2021stim}, not much is known for non-Clifford circuits.
For these circuits existing techniques to efficiently sample the correct syndrome statistics and assessing the success probability of a proposed correction~\cite{surti2025logicalmagic} are not applicable, and new methods are required.
We develop a method to compute the interaction between $X$ and $Z$ errors using a path-integral approach.
Our formalism applies to CSS-type QEC circuits consisting of controlled-$X$ gates and $X$ and $Z$ measurements, along with diagonal gates in the third level of the Clifford hierarchy.
This includes syndrome-extraction circuits of non-commuting ``twisted'' CSS codes whose $X$ stabilizers are decorated with diagonal Clifford operators~\cite{Ni2015XS, kobayashi2025cliffordcodes, williamson2026higher-form-gauging, zhu2026twisted-ldpc, christos2026twisted-ldpc}.

\myparagraph{Decoding and simulating non-Clifford QEC circuits}
Before describing our formalism in technical detail, we outline the conceptual differences between the commonly studied error-correcting Clifford circuits and the non-Clifford circuits considered here.
These circuits repeatedly measure a set of non-commuting Clifford stabilizers with eigenvalues of $\pm 1$.
Measuring such stabilizers projects the system onto their joint eigenspaces, and the goal of the QEC circuit is to reliably preserve logical information encoded in the $+1$ eigenspace.
The main obstacle to fault-tolerantly recover the information after many measurement rounds is the non-commuting nature of the stabilizers.
In particular, it follows that subspaces identified with a set of $-1$ outcomes can have reduced error-correcting capabilities.
Hence, to ensure that only high weight errors can alter the encoded information, corrections must be applied within the syndrome extraction circuit.
Reliable corrections require the use of a \emph{just-in-time} (JIT) decoder that decides on a correction based on the partial syndrome information that is available at the time that the correction must be applied.

For topological protocols, including the code we benchmark here, it has been shown that a JIT decoder can achieve a finite threshold under local stochastic noise~\cite{bombin20182d3d, Brown2020cczsurfacecode, davydova2025universal}.
The structural feature that enables JIT decoding in the circuits studied here is the existence of linear constraints among the measurement outcomes.
These constraints are used to detect the presence of harmful errors and yield a well-defined JIT decoding problem.

A fundamental complication in simulating protocols that involve many non-Clifford gates accurately is that Pauli errors can cause random measurement outcomes.
This makes it challenging to efficiently simulate the correct measurement statistics for a given error configuration and to decide on the logical action of the erroneous circuit given the output of a decoder.
One of the main technical contributions of this work is an efficient parametrization of the modified distribution over measurement outcomes through random additional twisted errors.
Moreover, we provide sufficient conditions to decide on the success of a given correction.
Together, this allows for a sampling-based, efficient, simulation method to estimate the logical failure probability of a large class of error-correcting circuits that include non-Clifford operations.
 
\section{TQD circuits with errors}\label{sec:TQDcircuits}
In this section, we define the \emph{twisted quantum double (TQD)} error-correcting circuit and the error model we consider, state the crucial properties of errors, measurement results and corrections, and discuss the implications on the way we decode and benchmark the logical performance.
We will give a rigorous and detailed derivation of the properties of errors, outcomes and corrections in Appendix~\ref{sec:tqd_derivation} by relating the circuit to a topological path integral.

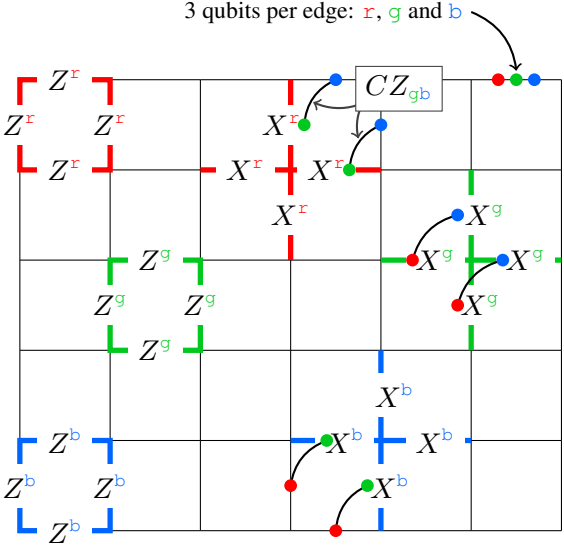
\begin{figure}
\begin{tikzpicture}[scale=1.2]
\foreach \x in {-2,...,3}{
\foreach \y in {-1,...,4}{
\draw (\x,\y) -- (\x+1,\y);
}
}
\foreach \x in {-2,...,4}{
\foreach \y in {-1,...,3}{
\draw (\x,\y) -- (\x,\y+1);
}
}
%
\coordinate (A) at (3,4.75);
\node[anchor=east] at (A) {3 qubits per edge: ${\color{myred}\texttt{r}}$, ${\color{mygreen}\texttt{g}}$ and ${\color{myblue}\texttt{b}}$};
\fill[myred] (3.3,4) circle (2pt) coordinate (p1);
\fill[mygreen] (3.5,4) circle (2pt) coordinate (p2);
\fill[myblue] (3.7,4) circle (2pt) coordinate (p3);
\draw[->, thick] (A) to[bend left=30] (3.5,4.1);
%
\begin{scope}[scale=1, shift={(-2,3)}, transform shape]
\draw[myred, line width=2pt] (0,0) rectangle +(1,1);
\node[fill=white] at (0,0.5) {$Z^{\color{myred}\texttt{r}}$};
\node[fill=white] at (0.5,1) {$Z^{\color{myred}\texttt{r}}$};
\node[fill=white] at (1,0.5) {$Z^{\color{myred}\texttt{r}}$};
\node[fill=white] at (0.5,0) {$Z^{\color{myred}\texttt{r}}$};    
\end{scope}
%
%
%
\begin{scope}[scale=1, shift={(-1,1)}, transform shape]
\draw[mygreen, line width=2pt] (0,0) rectangle +(1,1);
\node[fill=white] at (0,0.5) {$Z^{\color{mygreen}\texttt{g}}$};
\node[fill=white] at (0.5,1) {$Z^{\color{mygreen}\texttt{g}}$};
\node[fill=white] at (1,0.5) {$Z^{\color{mygreen}\texttt{g}}$};
\node[fill=white] at (0.5,0) {$Z^{\color{mygreen}\texttt{g}}$};    
\end{scope}
%
\begin{scope}[scale=1, shift={(-2,-1)}, transform shape]
\draw[myblue, line width=2pt] (0,0) rectangle +(1,1);
\node[fill=white] at (0,0.5) {$Z^{\color{myblue}\texttt{b}}$};
\node[fill=white] at (0.5,1) {$Z^{\color{myblue}\texttt{b}}$};
\node[fill=white] at (1,0.5) {$Z^{\color{myblue}\texttt{b}}$};
\node[fill=white] at (0.5,0) {$Z^{\color{myblue}\texttt{b}}$};    
\end{scope}
%
%
\begin{scope}[scale=1, shift={(0,3)}, transform shape]
\draw[myred, line width=2pt] (0,0) -- +(2,0);
\draw[myred, line width=2pt] (1,-1) -- +(0,2);
\node[fill=white] at (0.5,0) {$X^{\color{myred}\texttt{r}}$};
\node[fill=white] at (1.4,0) {$X^{\color{myred}\texttt{r}}$};
\node[fill=white] at (0.9,0.5) {$X^{\color{myred}\texttt{r}}$};
\node[fill=white] at (1,-0.5) {$X^{\color{myred}\texttt{r}}$};
\coordinate (C1) at (1.15,0.5);
\coordinate (Z1) at (1.5,1);
\draw[thick] (C1) to[bend left=30] (Z1);
\fill[mygreen] (C1) circle (2pt);
\fill[myblue] (Z1) circle (2pt);
\coordinate (C2) at (1.65,0);
\coordinate (Z2) at (2,0.5);
\draw[thick] (C2) to[bend left=30] (Z2);
\fill[mygreen] (C2) circle (2pt);
\fill[myblue] (Z2) circle (2pt);
\draw[thick, ->, darkgray] (2.1,1) to[bend right=40] (1.75, 0.35);
\draw[thick, ->, darkgray] (2.1,1) to[bend left=40] (1.25, 0.75);
\node[fill=White, draw=darkgray] (label) at (2.2,0.9) {$CZ_{\color{mygreen}\texttt{g}\color{myblue}\texttt{b}}$};
\end{scope}
%
\begin{scope}[scale=1, shift={(2,2)}, transform shape]
\draw[mygreen, line width=2pt] (0,0) -- +(2,0);
\draw[mygreen, line width=2pt] (1,-1) -- +(0,2);
\node[fill=white] at (0.6,0) {$X^{\color{mygreen}\texttt{g}}$};
\node[fill=white] at (1.6,0) {$X^{\color{mygreen}\texttt{g}}$};
\node[fill=white] at (1.15,0.5) {$X^{\color{mygreen}\texttt{g}}$};
\node[fill=white] at (1.1,-0.5) {$X^{\color{mygreen}\texttt{g}}$};
\coordinate (C1) at (0.35,0);
\coordinate (Z1) at (0.85,0.5);
\draw[thick] (C1) to[bend left=30] (Z1);
\fill[myred] (C1) circle (2pt);
\fill[myblue] (Z1) circle (2pt);
\coordinate (C2) at (0.85,-0.5);
\coordinate (Z2) at (1.35,0);
\draw[thick] (C2) to[bend left=30] (Z2);
\fill[myred] (C2) circle (2pt);
\fill[myblue] (Z2) circle (2pt);
\end{scope}
%
\begin{scope}[scale=1, shift={(1,0)}, transform shape]
\draw[myblue, line width=2pt] (0,0) -- +(2,0);
\draw[myblue, line width=2pt] (1,-1) -- +(0,2);
\node[fill=white] at (0.6,0) {$X^{\color{myblue}\texttt{b}}$};
\node[fill=white] at (1.6,0) {$X^{\color{myblue}\texttt{b}}$};
\node[fill=white] at (1.15,0.5) {$X^{\color{myblue}\texttt{b}}$};
\node[fill=white] at (1.1,-0.5) {$X^{\color{myblue}\texttt{b}}$};
\coordinate (C1) at (0,-0.5);
\coordinate (Z1) at (0.4,0);
\draw[thick] (C1) to[bend left=30] (Z1);
\fill[myred] (C1) circle (2pt);
\fill[mygreen] (Z1) circle (2pt);
\coordinate (C2) at (0.5,-1);
\coordinate (Z2) at (0.85,-0.5);
\draw[thick] (C2) to[bend left=30] (Z2);
\fill[myred] (C2) circle (2pt);
\fill[mygreen] (Z2) circle (2pt);
\end{scope}
\end{tikzpicture}
\caption{The TQD code can be defined on the square lattice, with three qubits assigned to each edge, which we label in three colors, $\{\ttr, \ttg, \ttb\}\eqqcolon \col$. The code is defined by its stabilizers:
For each color $\lambda\in\col$ and face $f\in F$ we define a stabilizer $B_{\lambda f} = \prod_{e\in f} Z_{\lambda e}$ acting with a product of $Z$ operators on the qubits of color $\lambda$ that are adjacent to $f$.
For each color $\lambda\in\col$ and each vertex $v\in V$ and we define a Clifford stabilizer $A_{\lambda v} = CZ_{(\lambda v)}\prod_{e\ni v}X_{\lambda e}$ that acts with a product of Pauli $X$ operators on adjacent qubits of color $\lambda$ and two $CZ$ operators that we together denote as $CZ_{(\lambda v)}$.
They each act on two neighboring qubits of the two color different from $\lambda$.
The figure shows a set of stabilizer generators -- all others are obtained from spatial translation.
Each bent line segment corresponds to one $CZ$ gate and the color on the endpoints indicates which qubit it acts on.}
\label{fig:stabilizers}
\end{figure}

\myparagraph{Error-correcting circuit and error model}
The error-correcting \emph{TQD circuit} consists of repeated measurements of the stabilizers of the \emph{TQD code} shown in Figure~\ref{fig:stabilizers}.
The stabilizers of the TQD code are Clifford operators that square to the identity but do not mutually commute.
The TQD code has first been introduced explicitly in Refs.~\cite{Hsin2025, kobayashi2025cliffordcodes}, but syndrome-extraction circuits for the same code were also constructed directly using spacetime methods in Refs.~\cite{davydova2025universal, Bauer2025allAbelian}.
Concretely, the code is defined on a square lattice with periodic boundary conditions, whose set of vertices, edges, and faces we denote by $V$, $E$, and $F$, respectively.
There are three qubits on each edge $e\in E$, each associated to a red/green/blue color $\lambda\in \col\coloneqq \{\ttr,\ttg,\ttb\}$.
We denote the Pauli-$X$ and $Z$ operators acting on the qubit of color $\lambda\in\col$ on edge $e\in E$ by $X_{\lambda e}$ and $Z_{\lambda e}$, respectively.
The stabilizers come in two types: One \emph{vertex stabilizer} $A_{\lambda v}$ for each color $\lambda\in\col$ and each vertex $v\in V$, and one \emph{plaquette stabilizer} and $B_{\lambda f}$ for each $\lambda\in \col$ and each face $f\in F$.
$B_{\lambda f}$ is equal to the product 
\begin{equation}
\label{eq:plaquette_stabilizer}
B_{\lambda f}=\prod_{e\in f} Z_{\lambda e}\;
\end{equation}
of Pauli-$Z$ on the edges of $f$.
$A_{\lambda v}$ is equal to the product 
\begin{align}
\label{eq:vertex_stabilizer}
A_{\lambda v} = CZ_{\lambda_0e_0,\lambda_0'e_0'}CZ_{\lambda_1e_1,\lambda_1'e_1'}\prod_{e\ni v} X_{\lambda e}\;
\end{align}
of Pauli-$X$ on the edges adjacent to $v$, along with two $CZ$ (controlled-$Z$) operations acting on nearby qubit pairs.
The precise location of the $CZ$ terms is shown in Figure~\ref{fig:stabilizers}.
Both plaquette and vertex stabilizers are Hermitian operators with eigenvalues $\{\pm 1\}$, and correspond to measurements whose outcomes we label by $s\in \{0,1\}\leftrightarrow\{\pm1\}$.

Since the stabilizers do not commute, we need to specify in which order the TQD circuit measures them.
In our numerical simulation we will consider the order $B_\ttb\rightarrow A_\ttb\rightarrow B_\ttg\rightarrow A_\ttg \rightarrow B_\ttr\rightarrow A_\ttr$, repeated periodically.
The order of $A$ and $B$ measurements of the same color among one another does not matter as they commute mutually.
We consider the circuit under a phenomenological error model:
Right before measuring $B_\lambda$, each $\lambda$-colored qubit experiences a Pauli-$X$ error with probability $p_X$, and a Pauli-$Z$ error with probability $p_Z$.
Additionally, each $B_\lambda$ measurement experiences a random bit-flip error with probability $p_X$, and each $A_\lambda$ measurement experiences an error with $p_Z$.
Finally, it will be necessary to apply Pauli-$X$ corrections on the $\lambda$-colored qubits during the execution of the circuit, namely right after measuring $B_\lambda$.

\myparagraph{Error-correction distance in the TQD code}
Let us develop some intuition on why active corrections are necessary within the TQD circuit introduced above using properties of the 2D non-commuting stabilizer group.
We assume to start with an ideal code state vector $\ket{\psi}$ for which $A_{\lambda v} \ket{\psi} = B_{\lambda f}\ket{\psi} = \ket{\psi}$, $\forall v,f,\lambda$.
Then, a Pauli $X$ error $E_X$ is applied to a subset of qubits.
This rotates $\ket{\psi}$ into a different eigenspace of a subset of $B$ stabilizers since each generator either commutes or anti-commutes with the error.
The ideal $B$ measurement can be labeled by a bitstring $b_m\in \bZ_2^{F}$ and the post-measurement state without measurement errors will be in a subspace $\cH^{b_m}$ which we call $b_m$-sector.
Next, we consider the measurement of $A$ generators in some arbitrary order on the state vector $E_X\ket{\psi}\in \cH^{b_m}$.
Since $A$ stabilizers commute with $B$ stabilizers their measurement preserves every $b_m$-sector.
The action within a non-trivial sector, however, explicitly depends on $b_m$.
To describe every possible state after measuring a complete set of $A$ stabilizers we calculate the commutation relation between the $A$ stabilizers and find
\begin{align}\label{eq:A-groupcomm}
    A_{\lambda v}A_{\lambda' v'} A_{\lambda v}^{\dagger} A_{\lambda' v'}^\dagger \in \langle \{B_{\ell f}\}_{\ell, f}\rangle.
\end{align}
For a fixed value of $b_m$, 
the right hand side is an array of $\pm 1$ such that within each sector, there exists a maximal subgroup $H_1^{b_m}$ of $A$ stabilizers whose elements commute among one another and an orthogonal subgroup $H_2^{b_m}$ that includes all elements that anti-commute with $H_1^{b_m}$.
We can hence treat $H_1^{b_m}$ as a non-Pauli stabilizer group of a subsystem code in the sense that the logical information is left invariant by each $A$ stabilizer but only the operators in $H_1^{b_m}$ provide useful syndrome information. 
In particular, we will not be able to detect any error that commutes with $H_1^{b_m}$.
Since $H_2^{b_m}$ is generated by operators that only act non-trivially in the vicinity of violated $B$ operators (see Eq.~\eqref{eq:A-groupcomm}), the set of undetectable $Z$ errors in the $\cH^{b_m}$ sector are $Z$ strings with endpoints in the vicinity of $B$-violations.
This is sufficient to see that the effective $Z$-distance depends on $b_m$.
Roughly, it scales like the minimum length of a non-trivial path that connects endpoints of $X$ strings.
This is only extensive in the system size if the $X$ errors are corrected before they spread on a length scale comparable to the linear system size.

There is also another reason 
why we do not 
want $X$ errors to spread arbitrarily in TQD circuits, even if they do not violate any $B$ stabilizer.
Namely, $A$ stabilizers do not commute with both $Z$ and $X$ errors.
$A_{\lambda v}$ anti-commutes with $Z$ errors on edges $e$ with $e\ni v$, so we use their syndrome information to infer the location of $Z$ errors.
The group commutator with $X$ errors on the other hand, is a product of $Z$ operators and it follows that the measurement of $A$ stabilizers that overlap with an $X$ error effectively create additional $Z$ errors.
This removes the ability to detect $Z$ errors further, not only in the vicinity of violated $B$ stabilizers but along the entire $X$ error.
Note that as a consequence, a product of $X$ operators along a loop in the dual lattice is not a stabilizer of the code, although it commutes with every $B$ operator.

Both properties highlighted above motivate that the instantaneous state throughout the TQD circuit must be close to an ideal code state at any time in the circuit to preserve the error-correcting distance.
Later in this section we will provide a rigorous formulation of the analogous statements to the above in spacetime and define a JIT decoder to reliably choose the $X$ corrections in the presence of measurement errors in Sec.~\ref{sec:JIT-decoder}.

\myparagraph{Basic notions of cellular cohomology}
The erroneous circuit is best described on a 3D cubic spacetime lattice using the language of $\zz_2$ cellular cohomology, which we briefly review in the following.
The spacetime lattice consists of cells of different dimensions.
The vertices correspond to $0$-cells, edges to $1$-cells, faces to $2$-cells and volumes to $3$-cells.
An \emph{$i$-chain} $a$ is a $\zz_2$-valued vector whose entries correspond to the different $i$-cells.
We denote the space of all $i$-chains by $C^i$.
The \emph{$i$th coboundary map} $d_i$ is the adjacency matrix between $i$-cells and $i+1$-cells:
$(d_i)_{xy}=1$ if the $i$-cell $y$ is adjacent to the $i+1$-cell $x$, and $(d_i)_{xy}=0$ otherwise.
We can view the coboundary map as a $\zz_2$-linear operator $d_i\colon C^i\to C^{i+1}$:
If $a$ is an $i$-chain, then $d_ia$ is an $i+1$-chain, and $d_{i-1}^Ta$ is an $i-1$-chain.
The central identity of cohomology theory is
\begin{equation}
\label{eq:cohomology_relation}
d_{i+1}d_i=0 \quad\forall i\;.
\end{equation}
This holds because every $i+2$-cell and $i$-cell share either no adjacent $i$-cells, or they share $2\equiv 0\in \zz_2\simeq\{0,1\}$.
An $i$-chain $a\in C^i$ is called an \emph{$i$-cocycle} if $d_ia=0$, and it is called an \emph{$i$-coboundary} if $a=d_{i-1}b$ for some $i-1$-chain $b$.
Similarly, $a$ is called an \emph{$i$-cycle} if $d_{i-1}^Ta=0$, and it is called an \emph{$i$-boundary} if $a=d_{i}^Tb$ for some $i+1$-chain $b$.
Every (co)boundary is automatically a (co)cycle due to Eq.~\eqref{eq:cohomology_relation}.
The sets of chains, (co)cycles, and (co)boundaries form groups under element-wise $\zz_2$ addition.
The quotient of the group of $i$-(co)cycles by the subgroup of $i$-(co)boundaries is called the \emph{$i$th (co)homology group}.
We will use $\langle a,b\rangle$ to denote the standard inner product between two $i$-chains $a,b\in C^i$.
On an $n$-dimensional lattice, we define the \emph{fundamental class} $\epsilon\in C^n$ to be the $n$-chain with $\epsilon_c=1\in\zz_2$ for all $n$-cells $c$.
$\epsilon^T=\langle\epsilon,\bullet\rangle$ can be viewed as a linear map $\epsilon^T\colon C^n \to \bZ_2$, and we have
\begin{equation}
\label{eq:fundamental_class_property}
\epsilon^T d_{n-1}=0\;,
\end{equation}
similar to Eq.~\eqref{eq:cohomology_relation}.
The subscript $i$ of the coboundary map $d_i$ can usually be inferred from the context, so we drop it as custom in cohomology theory.

Finally, there is a simple zoomed-out picture for lattice cohomology:
0-chains in $C^0$ are collections of points (those with vector entry $1$), 1-chains in $C^1$ are collections of strings (formed by paths of edges with vector entry $1$), 2-chains in $C^2$ are collections of surface segments (formed by the faces with vector entry $1$), and so on.
$d^Ta$ for $a\in C^1$ corresponds to the endpoints of the string collection $a$, and $a$ is a 1-cycle if it has no endpoints and all strings are closed loops.
$d^Ta$ for $a\in C^2$ corresponds to the boundary lines of the surface segments of $a$, and $a$ is a 2-cycle if there is no boundary and the surface segments form closed membranes.
$i$-cocycles are the same as $n-i$-cycles on the Poincar\'e dual lattice, and have the same zoomed-out picture.

Equipped with cellular cohomology, we are ready to state the properties of measurement outcomes and errors in the circuit which are essential for decoding.
There are two types of properties, namely constraints and equivalences.

\myparagraph{Constraints and equivalences in the toric code}
Before we get to the TQD circuit, it is instructive to recap the constraints and equivalences for the ordinary toric-code syndrome-extraction circuit.
The toric code has stabilizers similar to $A$ and $B$ in Eqs.~\eqref{eq:vertex_stabilizer} and \eqref{eq:plaquette_stabilizer}, but with only a single qubit per edge and without any $CZ$ terms.
As was already pointed out in Ref.~\cite{Dennis2002}, Pauli-$Z$ errors in the circuit can be associated with the space-like edges of a cubic spacetime lattice, and $X$-type measurement outcomes and measurement errors can be associated with the time-like edges.
Taken together, these edges form a $\zz_2$-valued 1-chain $c=c_m+c_e$ on the cubic lattice, where $c_m$ corresponds to the measurement outcomes and $c_e$ to the (measurement or Pauli) errors.
We will refer to $c$ as the \emph{charge configuration}.
Analogously, Pauli-$X$ errors and $Z$-type measurements (measurement errors) form a 1-chain on the Poincar\'e dual lattice, or equivalently a 2-chain on the primal cubic lattice, called the \emph{flux configuration} $b=b_m+b_e$.
For a fixed flux and charge configuration the syndrome extraction circuit implements a product of linear operators, which we denote with $C_{b,c}$.
It can be shown that within the circuit the flux configuration $b$ fulfills the following constraint and equivalence,
\begin{subequations}
\begin{align}
\label{eq:tcb_constraint}
db\neq& 0\quad\Rightarrow\quad C_{c,b}=0\;,\\
\label{eq:tcb_equivalence}
C_{b+d\beta, c} \propto & C_{b,c} \quad\forall \beta\in C^1\;.
\end{align}
\end{subequations}
The constraint in Eq.~\eqref{eq:tcb_constraint} implies that the probability for obtaining measurement outcomes such that $db\neq 0$ is zero.
This helps us to estimate the error $b_e$ given the measurement outcomes $b_m$, as it tells us that $b=b_e+b_m\Rightarrow db_e=db_m$.
In this context, $db_m$ is known as the \emph{syndrome} and the volumes (on which $db_m$ is defined) correspond to \emph{detectors}.
The $\zz_2$-matrix $d$ is the adjacency matrix of the \emph{decoding graph} for $X$-like errors.
For a code defined on a square lattice this is a 3D cubic lattice.
The equivalence in Eq.~\eqref{eq:tcb_equivalence} implies that all configurations of measurement outcomes in the same cohomology class are equally likely.
In the context of decoding, $b$ configurations of the form $b=d\beta$ correspond to undetectable but inconsequential errors, such as the same Pauli-$X$ error happening in two consecutive time steps along with two measurement errors at the two affected $Z$-type stabilizers.
This equivalence means that decoding is successful if the cohomology class of $b_e$ is estimated correctly.
We do not need to estimate its precise support.
Both properties together allow us to decode and successfully correct errors at the end of the protocol with a very high probability.

The situation for $c$ is the same as for $b$, apart from exchanging the cubic lattice and its Poincar\'e dual, and exchanging $d$ and $d^T$.
The constraints and equivalences are
\begin{subequations}
\label{eq:tcc_conseq}
\begin{align}
\label{eq:tcc_constraint}
d^Tc\neq& 0\quad\Rightarrow\quad C_{c,b}=0\;,\\
\label{eq:tcc_equivalence}
C_{b, c+d^T\gamma} \propto& C_{b,c}\quad \forall \gamma\in C^2\;.
\end{align}    
\end{subequations}

\myparagraph{Constraints and equivalences in the TQD circuit}
After discussing the toric code, let us see what is different for the constraints and equivalences of errors, measurement results, and corrections in the TQD circuit.
A rigorous derivation using a path-integral formulation of TQD circuits can be found in Appendix~\ref{sec:tqd_derivation}.
Just like for the toric code, errors and measurement outcomes in the TQD circuit can be mapped onto chains on the cubic lattice.
The only difference is that $b=(b_\ttr,b_\ttg,b_\ttb)$ is now a triple of 2-chains, and $c=(c_\ttr,c_\ttg,c_\ttb)$ is a triple of 1-chains as the TQD code is defined with three qubits on every edge, and three stabilizer generators on each vertex and plaquette, respectively.
Note that in addition to measurement outcomes and errors, we also need to apply Pauli-$X$ corrections during the execution of the circuit.
Similar to $b_e$ and $b_m$, the locations of these corrections can be mapped to another 2-chain $b_c$, and we set $b\coloneqq b_e+b_m+b_c$.
The constraints and equivalences for $b$ in the TQD circuit are
\begin{subequations}
\label{eq:tqdb_conseq}
\begin{align}
\label{eq:tqdb_constraint}
db\neq& 0\quad\Rightarrow\quad C_{c,b}=0\;,\\
\label{eq:tqdb_equivalence}
C_{b+d\beta, c} \propto& C_{b,c} \quad\forall \beta\in C^1\colon \beta \text{ isolated from } b\;,
\end{align}
\end{subequations}
The constraint in Eq.~\eqref{eq:tqdb_constraint} is the same as Eq.~\eqref{eq:tcb_constraint}, except that $b$ is a triple and the coboundary $d$ acts on each of its three components.
So we can still use $db_e=d(b_m+b_c)$ to gain some information about the error $b_e$.
In other words, the decoding graph and syndrome for $X$ errors is unchanged.
The equivalence in Eq.~\eqref{eq:tqdb_equivalence} is much weaker than Eq.~\eqref{eq:tcb_equivalence}:
We can only add $d\beta$ to $b$ if $b$ is zero on the support of $d\beta$ (possibly enlarged by a small constant).
The loss of full cohomological invariance is related to the fact that insertions of Pauli-$X$ strings (corresponding to $b$) can no longer be propagated through the non-Pauli measurements.

So if we want to make sure that the configuration of $b$ causes no logical error, it must not contain any non-contractible loops, in which case we will say that $b$ is \emph{topologically trivial}.
This is more strict than in the toric-code case where it sufficed for $b$ to be cohomologically trivial.
For example, two non-contractible loops of the same type would be trivial in $\zz_2$ cohomology, but they are topologically non-trivial.
Vice versa, if $b$ contains any non-contractible loop, we cannot in general infer the logical effect of the TQD circuit.
In particular, it is not necessarily a stochastic Pauli channel on the logical level.
To keep $b$ topologically trivial we need to apply Pauli-$X$ corrections on the fly to continuously shrink the endpoints of $b$ on the current time slice. 
If we do not do this, a large configuration of $b$ will likely build up over time from which we cannot recover the desired logical effect of the full protocol.
In order to choose the Pauli-$X$ corrections as we go, we need \emph{just-in-time decoding}, which we describe in Section~\ref{sec:JIT-decoder}.

Next, we will state the constraints and equivalences for $c$.
We will refer to these as \emph{twisted} constraints and equivalences to distinguish them from the ``untwisted'' constraints and equivalences of the toric code.
They are given by
\begin{subequations}
\label{eq:tqdc_conseq}
\begin{align}
\label{eq:tqdc_constraint}
K^bd^Tc\neq& \kappa^b\quad\Rightarrow\quad C_{c,b}=0\;,\\
\label{eq:tqdc_equivalence}
C_{b, c+d^T\gamma + \widetilde M^b v} \propto& C_{b,c}\quad \forall \gamma\in C^2, v\in C^0\;.
\end{align}
\end{subequations}
Here, $K^b: \zz_2^{\col\times V}\rightarrow \zz_2^X$ is a linear map from triples of 0-chains to vectors over some abstract set $X$ of generating twisted constraints, $\kappa^b\in\zz_2^X$ is a $\zz_2$-valued vector over the generating twisted constraints, and $\widetilde M^b:\zz_2^{\col\times V}\rightarrow \zz_2^{\col\times E}$ is a linear map from 0-chain triples to 1-chain triples.
Most notably, the constraints and equivalences of $c$ depend on the configuration $b$.
This is an important difference from the toric-code case:
The $c$ decoding problem in the TQD circuit depends on the $b$ configuration.
Explicitly, the map $\widetilde M^b$ is shown in Figure~\ref{fig:zeff_illustration}, and $K^b$ is an injective map whose image is the kernel of $d^T\widetilde M^b$.
If $b=0$, then $K^0=\one$, $\kappa^0=0$, and $\widetilde M^b=0$, so the $c$ constraints and equivalences are equal to these for (three copies of) the toric code shown in Eq.~\eqref{eq:tcc_conseq}.
If $b\neq 0$, then some constraints in Eq.~\eqref{eq:tqdc_constraint} are reduced to the image of $(K^b)^T$ relative to these for $b=0$, shown in Eq.~\eqref{eq:tcc_constraint}.
At the same time, new equivalences are added in Eq.~\eqref{eq:tqdc_equivalence} relative to $b=0$, namely these in the image of $\widetilde M^b$.
In fact, for every missing generating constraint, we get precisely one new generating equivalence.
So the overall ``twisted homology group'' of the decoding problem, the space of $c$ configurations satisfying the constraints modulo the equivalences, remains invariant.
In Figure~\ref{fig:zeff_illustration} (b) we present a simple zoomed-out picture for the modified constraints and equivalences.
The twisted constraints allow $d^Tc$ to be non-zero in the vicinity of $b$, that is, the $c$ strings can terminate near $b$.
While the $b=0$ equivalences correspond to adding small contractible loops to $c$, the new twisted equivalences also allow adding short strings with endpoints in the vicinity of $b$.

In decoding language, Eq.~\eqref{eq:tqdc_constraint} means that the decoding graph is reduced from a cubic graph, defined by $d^T$ to a hypergraph with adjacency matrix $K^bd^T$ and $\kappa^b$ to a error-free syndrome.
The vector $\kappa^b$ exactly captures the linking charge phenomenon~\cite{bombin2018error-propagation}.
The new hypergraph has a potentially smaller code distance, that is, a smaller weight of the minimal $c$ with $K^bd^Tc=0$ of non-trivial homology.
Roughly, the larger $b$ the more the distance can be reduced, as illustrated in Figure~\ref{fig:zeff_illustration} (c).
Intuitively, a non-trivial $c$ does not need to form a full non-contractible loop, but it can connect different $b$ segments to form the loop.
In particular, the distance can approach $0$ if $b$ itself contains any non-contractible loop, even if it is homologically trivial.
Note that in these cases the logical action of $b$ itself is also unknown and there is no way to recover the correct logical channel, even without any $Z$ errors.
The fact that $b$ reduces the distance for the $c$ decoding problem is an additional reason why we need to keep $b$ small as we go, for which we need just-in-time decoding.

Finally, we want to highlight a useful perspective on the new twisted equivalences of the form $\widetilde M^b v$.
They can be understood as Pauli-$Z$ and $X$-measurement errors that occur with probability $\frac12$ everywhere along the support of $b$.
In fact, these $Z$ errors can be correlated, with up to $6$ edges of the cubic lattice of 2 different colors experiencing an error at the same time.
The exact structure of the twisted errors depends on the location of errors and corrections with respect to the non-Clifford elements in the circuit.
We call these probability-$\frac12$ errors \emph{twisted errors}.
Note that any decoder can only know part of the support of $b$, namely $b_m$ (observed measurement outcomes) but not $b_e$ (physical errors).
This allows for a partial heralding strategy for the twisted errors, similar to erasure errors.
We describe one explicit, simple, strategy to make use of this additional information for the $Z$ decoding in Section~\ref{sec:charge-decoder}.

\begin{figure*}
\input{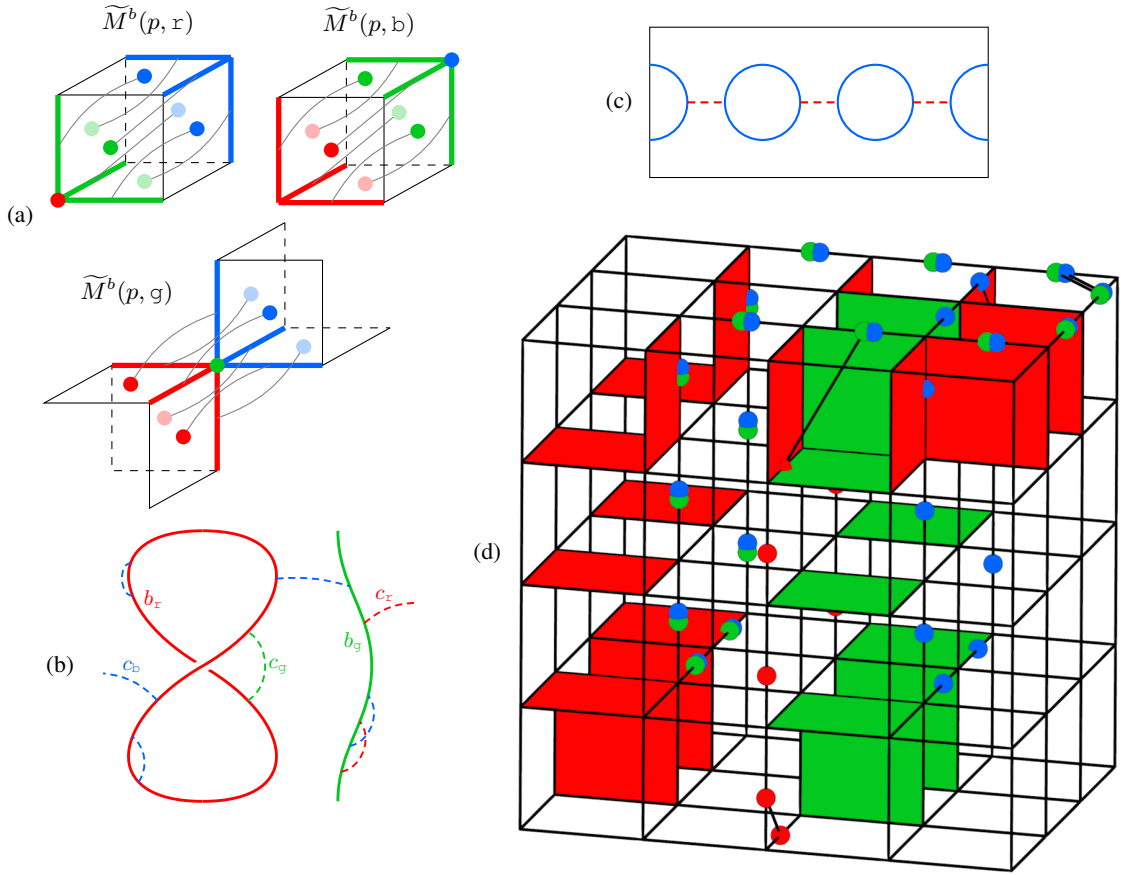}
\caption{
Twisted errors in the TQD circuit.
Figure (a) shows how to obtain the twisted errors from the flux configuration $b$ identified as a set of faces.
There are three types of twisted errors per cubical unit cell, corresponding to $\widetilde M^b (\lambda,p)$ for any color-vertex pair $(\lambda, p)$.
Each twisted error consists of a subset of 6 color-edge pairs, which are marked in the drawing as colored edges.
Each of these 6 color-edge pairs is part of the subset iff $b_{\lambda f}=1$ for a specific color-face pair $(\lambda, f)$ depending on the color-edge pair.
These color-face pairs are marked by a colored dot at the center of $f$ and connected to their color-edge pair with a thin line.
In particular, the twisted errors are trivial (the empty subset) if $b=0$ on all faces nearby.
Figure (b) displays a zoomed-out picture of an example $b$ configuration (solid lines) and an example $c$ configuration (dashed lines) that fulfills the twisted constraints. We see that the twisted constraints allow some charge lines $c$ to ``terminate'' in the vicinity of the flux lines $b$ without being detected by the twisted detectors.
Figure (c) illustrates how a non-trivial flux $b\neq 0$ can reduce the fault distance for the $Z$ decoding.
A non-trivial cycle $c_{\ttr}$ in the $Z$ decoding graph in the presence of flux loops $b_{\ttb}$ can have lower weight than a corresponding non-trivial cycle in the absence of fluxes since the $c_{\ttr}$ is allowed to ``jump between'' the loops.
Figure (d) shows an exemplary flux configuration $b$ consisting of red and green faces ($b_\ttr$ and $b_\ttg$ with $b_\ttb=0$). Each connected component forms a loop on the Poincar\'e dual lattice.
Each color-edge pair that is part of a twisted error is marked by a small ball of the corresponding color on the edges.
For each flux configuration of one color, there are twisted errors of the other two colors on edges in its geometric neighborhood.
For an isolated flux loop of one color, the set of edges that host all the twisted errors are obtained by shifting the dual lattice onto the direct lattice in a specific way.
In the example, there are green and blue twisted errors along the red $b$ loop and red and blue twisted errors along the green $b$ loop.
Most twisted errors correspond to a single color-edge pair.
Some twisted errors are also correlated weight-2 errors. These are indicated by a connecting line between the two corresponding balls.
When multiple $b$ loops coincide at a cube, the situation can get more complicated.
For example, at the top front there is a twisted error consisting of both a red and a green edge, so the non-trivial $b$ leads to a coupling of the red, green, and blue $Z$ decoding problems.
For general $b$ configurations, twisted errors can consist of up to 6 color-edge pairs, even though the example above shows only twisted errors with 1 or 2 edges (which are the most common for a sparse $b$ configuration).
We would like to note that this analysis directly applies to the 3-copy-cup gate on 3D toric codes on the cubic lattice~\cite{breuckmann2025cupsgatesicohomology}.
}
\label{fig:zeff_illustration}
\end{figure*}

\myparagraph{Benchmarking of the TQD circuit}
Given the constraints and equivalences of the measurement outcomes, errors, and corrections, we can efficiently compute an upper bound for the logical error rate of the TQD circuit by Monte-Carlo sampling.
Notably, we do not have to perform a state-vector simulation of the circuit containing non-Clifford gates.
Instead, it suffices 
to sample different configurations of errors and measurement outcomes.
This is similar to how errors and corrections on Pauli codes or in Clifford circuits can directly be sampled on the decoding graph~\cite{Gidney2021stim, delfosse2023spacetime}.

To begin with, we draw error configurations $b_e$ and $c_e$ i.i.d. randomly in accordance with the phenomenological noise model.
Given the error configuration, we determine the $X$ syndrome $db_m=db_e$ via Eq.~\eqref{eq:tqdb_constraint}, and feed it to the just-in-time decoder.
The just-in-time decoder determines an $X$ correction $b_c$ in a step-by-step fashion.
Due to the weakened equivalence in Eq.~\eqref{eq:tqdb_equivalence}, the just-in-time decoder is successful if $b_c+b_m+b_e$ does not have any non-contractible loops.
If there are non-contractible loops, we abort with failure, as the logical effect that $b$ has on the circuit is irreversible, and the following $Z$ decoding problem becomes ill-defined.

Given $b=b_e+b_m+b_c$ and successful just-in-time decoding, we proceed with the global $Z$ decoding.
Even without any actual $Z$-type errors, the twisted errors give rise to a random syndrome, which is a solution to Eq.~\eqref{eq:tqdc_constraint}.
Each solution is equally likely due to Eq.~\eqref{eq:tqdc_equivalence}.
We can sample a random solution by choosing one fixed solution $c_0$ to Eq.~\eqref{eq:tqdc_constraint} and adding a random element of the kernel of $d^TK^b$, which is equal to the image of $\widetilde M^b$.
For the latter, we choose a random 0-chain $v\coloneqq (v_\ttr,v_\ttg,v_\ttb)\in \zz_2^{\col\times V}$, whose value on each color-vertex pair is $0$ or $1$ with probability $\frac12$.
After adding the actual $Z$ error configuration $c_e$, we have $d^Tc_m = d^T (c_e + c_0+\widetilde M^b v)$.
After sampling the $Z$ syndrome $d^T c_m$, we pass it to the global $Z$ decoder and obtain as output a 1-chain $\widetilde c_e$ with $d^T\tilde c_e=d^Tc_m$.
If $\widetilde c_e + c_e+c_0+\widetilde M^b v$ has trivial homology, then the circuit with the actual $c=c_m+c_e$ is equal to the circuit with the estimated $\widetilde c=c_m+\widetilde c_e$, due to the equivalence in Eq.~\eqref{eq:tqdc_equivalence}.
So the decoder is successful in this case -- conversely, the decoder fails if $\widetilde c_e + c_e + c_0+\widetilde M^b v$ has non-trivial homology.
The benchmarking procedure is summarized in Algorithm~\ref{alg:benchmarking}.
The simulation results for the full TQD circuit can be found in Sec.~\ref{sec:charge-decoder}.

\begin{remark}
    In our actual numerical simulation in Sec.~\ref{sec:charge-decoder}, we have slightly simplified Algorithm~\ref{alg:benchmarking} by setting $c_0=0$ in step 5, in other words, we ignore the linking-charge phenomenon.
    Since linked loops are unlikely to occur in a sparse residual flux configuration, the real $c_0$ can be chosen close to zero in most cases, and this is a valid approximation for the noise model and type of JIT decoder we consider.
\end{remark}

\begin{algorithm}
\caption{Sampling TQD circuits with errors}
\label{alg:benchmarking}
\DontPrintSemicolon
Generate $m$ samples of logical success or failure as follows\;
\nl Sample i.i.d.\ $X$ error configurations $b_e$ by including each face of the cubic spacetime lattice with probability $p_X$\;
\nl Apply the just-in-time decoder over $T$ time steps. The output is a 2-chain $b_c$ such that $d(b_e+b_m+b_c) = 0$\;
\nl Determine if $b\coloneqq b_e+b_m+b_c$ contains any non-contractible loops. If yes, declare decoder failure and abort. If no, continue with $Z$ decoding\;
\nl Sample i.i.d.\ $Z$ error configuration $c_e$ by including each edge of the cubic spacetime lattice with probability $p_Z$\;
\nl Compute an arbitrary solution $c_0$ to Eq.~\eqref{eq:tqdc_constraint}, that is, $K^bd^T c_0 = \kappa^b$\;
\nl Sample i.i.d.\ twisted error configuration $v$ by choosing $v_{\lambda v}\in \{0,1\}$ for each vertex $v\in V$ and each color $\lambda\in \col$ with probabilities $\{\frac12,\frac12\}$\;
\nl Run the global decoder for $Z$ errors. The input is given by the syndrome $d^T(c_e + c_0 +\widetilde M^b v)$. The decoder may also be granted access to $b_m+b_c$. The output of the decoder is a 1-chain $\widetilde c_e$ such that $d^T\widetilde c_e = d^T(c_e+c_0+\widetilde M^b v)$\;
\nl If $\widetilde c_e+c_e+c_0\widetilde M^bv$ has trivial homology, declare decoder success, otherwise declare decoder failure\;
\end{algorithm}

\section{The JIT matching decoder}\label{sec:JIT-decoder}
In this section we provide a procedure of how to define a \emph{just-in-time} (JIT) decoder from a global matching decoder.
A similar procedure has also been described in Ref.~\cite{bombin20182d3d} as a ``naive'' JIT decoder and in Ref.~\cite{bauer2025planar}.
In Appendix~\ref{app:generalJIT} we define a more general JIT decoder that might be relevant in a broader context when a correction must be applied with only partial syndrome information.

\myparagraph{Intuitive explanation}
There are two reasons why corrections must be applied actively within TQD circuits which we discussed in detail in Section~\ref{sec:TQDcircuits}.
On a high level, one reason is that the measurements of Clifford stabilizers only provide reliable information about the presence of errors in the subspace stabilized by the plaquette stabilizers.
Another reason is that the logical action of the entire protocol depends on the topology of the chain formed by errors and corrections, not only on its $\bZ_2$-homology.
That is, the logical effect of two parallel chains of $X$ errors and corrections along homologous cycles in the spacetime lattice do not cancel each other.
A viable strategy to ensure that the spacetime region in which the errors and corrections are supported is topologically trivial is by keeping each cluster of errors small.
Note that this is not the case in Clifford circuits, where the Pauli frame can be tracked efficiently and any Pauli correction can be postponed.

A \emph{just-in-time} (JIT) decoder is a decoder that, whenever needed, decides on a correction at that timestep while given access to the syndrome information acquired until that timestep.
In TQD circuits, a JIT decoder must decide on a correction supported within the narrow time window between the last round of syndrome extraction and the current one.
We call that window the \emph{commit region}.

We define a JIT decoder based on a minimum-weight decoder for the global decoding graph with variable weights.
For the phenomenological noise model we consider in our simulations the decoding graph is $\bZ_2$-matchable and we use a \emph{minimum-weight-perfect-matching} (MWPM) decoder.
The same construction also works for any other global decoder whose output is a specific correction.
At time $t$, a matching JIT decoder consists of two steps:
\begin{enumerate}
    \item Make a \emph{current estimate} of the lowest-weight global error that fits the observed syndrome.
    \item Find the lowest-weight correction at the current timestep that \emph{merges previous corrections} with the current estimate.
\end{enumerate}
It is crucial for the performance of the decoder to use the maximal amount of available syndrome information in the first step.
This information is used to decide on a set of endpoints on the boundary of the commit region.
In the second step, the endpoints of the previous timestep and the current one are matched within the commit region.
We illustrate the individual steps in Figure~\ref{fig:jit}.

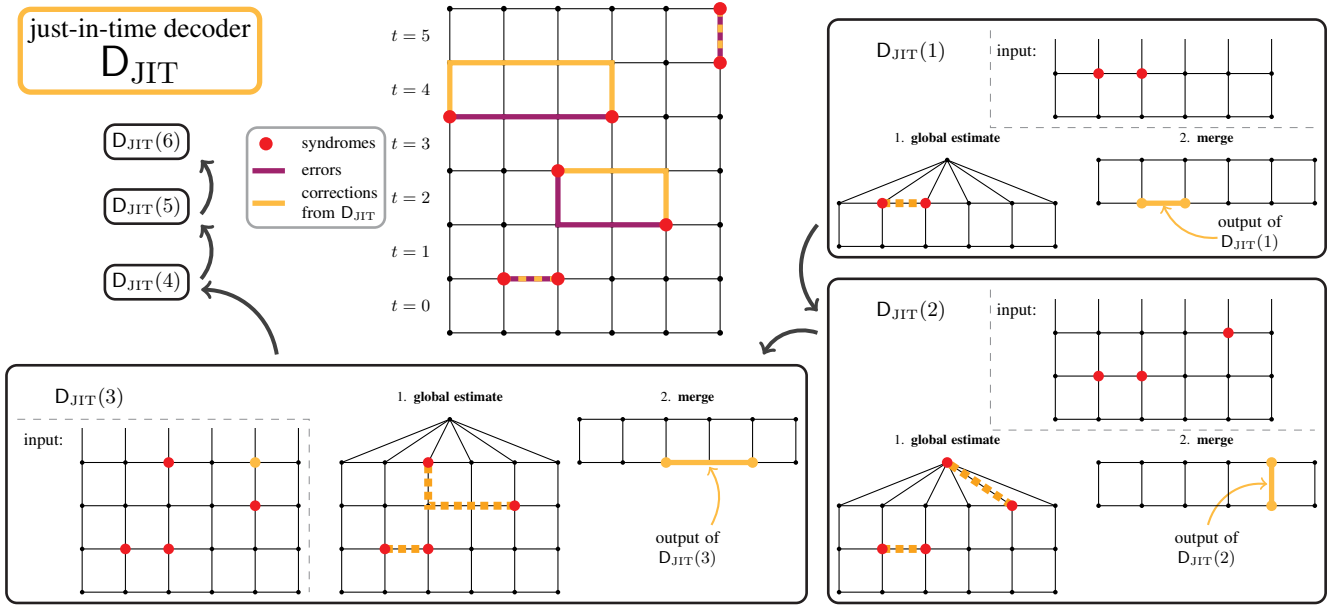
\begin{figure*}
    \centering
    \begin{tikzpicture}[scale=0.575, transform shape,
    vertex/.style={circle,fill=black,inner sep=0pt,minimum size=3pt},
    edge/.style={thin},
    synd/.style={circle,fill=Red,inner sep=0pt,minimum size=7pt},
    futuresynd/.style={circle,fill=Red,inner sep=0pt,minimum size=7pt, opacity=0.6},
    guessedsynd/.style={circle,fill=Dandelion,inner sep=0pt,minimum size=7pt},
    error/.style={line width=2pt, RedViolet},
    cor1/.style={line width=3pt, YellowOrange, dotted},
    cor2/.style={line width=2pt, Dandelion},
    comm/.style={line width=2pt, Dandelion}
]

\newcommand{\drawgrid}[2]{%
\def\rows{#1} 
\def\cols{#2} 
\def\spacing{1} 
\foreach \i in {0,...,\numexpr\rows-1\relax}{
    \foreach \j in {0,...,\numexpr\cols-1\relax}{
        \coordinate (v\i\j) at (\j*\spacing, \i*\spacing);
    }
}
\foreach \i in {0,...,\numexpr\rows-1\relax}{
    \foreach \j in {0,...,\numexpr\cols-2\relax}{ 
        \pgfmathtruncatemacro{\next}{\j+1}
        \draw[edge] (v\i\j) -- (v\i\next);
    }
}
\foreach \i in {0,...,\numexpr\rows-2\relax}{ 
    \foreach \j in {0,...,\numexpr\cols-1\relax}{
        \pgfmathtruncatemacro{\next}{\i+1}
        \draw[edge] (v\i\j) -- (v\next\j);
    }
}
\foreach \i in {0,...,\numexpr\rows-1\relax}{
    \foreach \j in {0,...,\numexpr\cols-1\relax}{
        \node[vertex] at (v\i\j) {};
    }
}
}
\begin{scope}[shift={(-9,1)}, scale=1.25, transform shape]
\drawgrid{7}{6}
\foreach\i in {0,...,5}{
    \node[anchor=east, shift={(-0.25,0.5)}] at (v\i0) {$t=\i$}; 
}

\draw[error] (v12) -- (v11);
\draw[error] (v40) -- (v43);
\draw[error] (v65) -- (v55);
\draw[error] (v24) -- (v22) -- (v32);
\draw[comm, dashed] (v12) -- (v11);
\draw[comm] (v40) -- (v50) -- (v53) -- (v43);
\draw[comm, dashed] (v65) -- (v55);
\draw[comm] (v32) -- (v34) -- (v24);

\node[synd] at (v12) {};
\node[synd] at (v11) {};
\node[synd] at (v40) {};
\node[synd] at (v65) {};
\node[synd] at (v43) {};
\node[synd] at (v55) {};
\node[synd] at (v24) {};
\node[synd] at (v32) {};

\begin{scope}[shift={(-2.75,3)}]
\draw[Gray, rectangle, fill=white, opacity=0.85, rounded corners, line width=1pt] (-1,-1.05) rectangle (1.55,0.8);
\node[align=left, anchor=west] at (-0.1,0) {errors};
\draw[error] (-0.9,0) -- (-0.3,0);
\node[align=left, anchor=west] at (-0.1,0.5) {syndromes};
\node[synd, align=center] at (-0.65, 0.5) {};
\node[align=left, anchor=west, text width=2.5cm] at (-0.1,-0.6) {corrections from $\sDJIT$};
\draw[comm] (-0.9,-0.6) -- (-0.3,-0.6);
\end{scope}

\end{scope}





\begin{scope}[shift={(0,3)}]
\draw[Black, rectangle, fill=white, rounded corners, line width=1pt] (-0.25,-0.25) rectangle (11.25,5.25);
\node[anchor=west, scale=1.5, text width=2cm, align=center] at (0,4.5){$\sDJIT(1)$};
\begin{scope}[shift={(5,3)}]
\drawgrid{3}{6}

\node[synd] at (v12) {};
\node[synd] at (v11) {};
\fill[white] (-0.2, 1.8) rectangle (5.2,2.2);
%
\node[anchor=east, scale=1.25, text width=2cm, align=center] at (0.5,1.5){input:};
\draw[dashed, Gray] (-1.5,2) -- (-1.5,-0.25) -- (6,-0.25);

\end{scope}

\begin{scope}[shift={(0,0)}]

\node[align=center, text width=3cm] at (2.5, 2.5) {1. \textbf{global estimate}};
\drawgrid{2}{6}
\coordinate (w) at (2.5,2);
\foreach \i in {0,...,5}{%
    \draw[edge] (w) -- (v1\i);
}
\node[vertex] at (w) {};

\draw[cor1] (v11) -- (v12);
\node[synd] at (v12) {};
\node[synd] at (v11) {};
\end{scope}

\begin{scope}[shift={(6,1)}]
\node[align=center, text width=2cm] at (2.5, 1.5) {2. \textbf{merge}};
\node[scale=1.25, text width=2cm, align=center] (out) at (3.5,-0.7){output of $\sDJIT(1)$};
\draw[->, thick, Dandelion] (2.75,-0.8) to[bend left=35] (1.5,-0.05);

\drawgrid{2}{6}

\draw[cor2] (v01) -- (v02);
\node[guessedsynd] at (v01) {};
\node[guessedsynd] at (v02) {};
\end{scope}    
\end{scope}

\begin{scope}[shift={(0,-5)}]
\draw[Black, rectangle, fill=white, rounded corners, line width=1pt] (-0.25,-0.25) rectangle (11.25,7.25);
\node[anchor=west, scale=1.5, text width=2cm, align=center] at (0,6.5){$\sDJIT(2)$};

\begin{scope}[shift={(5,4)}]
\drawgrid{4}{6}

\node[synd] at (v12) {};
\node[synd] at (v11) {};
\node[synd] at (v24) {};
\node[synd] at (v32) {};
\fill[white] (-0.2, 2.8) rectangle (5.2,3.2);
\node[anchor=east, scale=1.25, text width=2cm, align=center] at (0.5,2.5){input:};
\draw[dashed, Gray] (-1.5,3) -- (-1.5,-0.25) -- (6,-0.25);
\end{scope}

\begin{scope}[shift={(0,0)}]

\node[align=center, text width=3cm] at (2.5, 3.5) {1. \textbf{global estimate}};
\drawgrid{3}{6}
\coordinate (w) at (2.5,3);
\foreach \i in {0,...,5}{%
    \draw[edge] (w) -- (v2\i);
}
\node[vertex] at (w) {};

\draw[cor1] (v11) -- (v12);
\draw[cor1] (v24) -- (w);
\node[synd] at (v12) {};
\node[synd] at (v11) {};
\node[synd] at (v24) {};
\node[synd] at (w) {};
\end{scope}

\begin{scope}[shift={(6,2)}]
\node[align=center, text width=2cm] at (2.5, 1.5) {2. \textbf{merge}};
\node[scale=1.25, text width=2cm, align=center] (out) at (2.5,-1){output of $\sDJIT(2)$};
\draw[->, thick, Dandelion] (out) to[bend left=40] (3.90,0.5);

\drawgrid{2}{6}

\draw[cor2] (v04) -- (v14);
\node[guessedsynd] at (v04) {};
\node[guessedsynd] at (v14) {};

\end{scope}    
\end{scope}

\begin{scope}[shift={(-19,-5)}]
\draw[Black, rectangle, fill=white, rounded corners, line width=1pt] (-0.25,-0.25) rectangle (18.25,5.25);
\node[anchor=west, scale=1.5, text width=2cm, align=center] at (0,4.5){$\sDJIT(3)$};

\begin{scope}[shift={(1.5,0)}]
\drawgrid{5}{6}

\node[guessedsynd] at (v34) {};
\node[synd] at (v12) {};
\node[synd] at (v11) {};
\node[synd] at (v24) {};
\node[synd] at (v32) {};
\fill[white] (-0.2, 3.8) rectangle (5.2,4.2);
\node[anchor=east, scale=1.25, text width=2cm, align=center] at (0.5,3.5){input:};
\draw[dashed, Gray] (-1.5,4) -- (5.25,4) -- (5.25,0);
\end{scope}

\begin{scope}[shift={(7.5,0)}]

\node[align=center, text width=3cm] at (2.5, 4.5) {1. \textbf{global estimate}};
\drawgrid{4}{6}
\coordinate (w) at (2.5,4);
\foreach \i in {0,...,5}{%
    \draw[edge] (w) -- (v3\i);
}
\node[vertex] at (w) {};

\draw[cor1] (v11) -- (v12);
\draw[cor1] (v24) -- (v22) -- (v32);
\node[synd] at (v12) {};
\node[synd] at (v11) {};
\node[synd] at (v24) {};
\node[synd] at (v32) {};
\end{scope}

\begin{scope}[shift={(13,3)}]
\node[align=center, text width=2cm] at (2.5, 1.5) {2. \textbf{merge}};
\node[scale=1.25, text width=2cm, align=center] (out) at (2.5,-2){output of $\sDJIT(3)$};
\draw[->, thick, Dandelion] (out) to[bend right=30] (3,-0.1);

\drawgrid{2}{6}

\draw[cor2] (v04) -- (v02);
\node[guessedsynd] at (v04) {};
\node[guessedsynd] at (v02) {};
\end{scope}    
\end{scope}

\coordinate (A) at (-17,1.75);
\coordinate (B) at (-17,3.5);
\coordinate (C) at (-17,5);
\coordinate (D) at (-19,6.5);

\node[anchor=south west, scale=1.5, align=left, draw=Black, rectangle, fill=white, rounded corners, line width=1pt] at (A) {$\sDJIT(4)$};
\node[anchor=south west, scale=1.5, align=left, draw=Black, rectangle, fill=white, rounded corners, line width=1pt] at (B) {$\sDJIT(5)$};
\node[anchor=south west, scale=1.5, align=left, draw=Black, rectangle, fill=white, rounded corners, line width=1pt] at (C) {$\sDJIT(6)$};
\node[anchor=south west, scale=2, align=center, fill=white, draw=Dandelion, rounded corners, line width=2pt] at (D) {just-in-time decoder\\ \Large $\sDJIT$};

\draw[->, line width=1.5pt, darkgray] (-0.5,3.5) to[bend right=40] (-0.5,1.5);
\draw[->, line width=1.5pt, darkgray] (-0.5,1) to[bend right=40] (-1.75,0.5);
\draw[->, line width=1.5pt, darkgray] (-13,0.5) to[bend right=40] (-14.75,2);
\draw[->, line width=1.5pt, darkgray] (-14.75,2.25) to[bend right=40] (-14.75,3.5);
\draw[->, line width=1.5pt, darkgray] (-14.75,3.75) to[bend right=40] (-14.75,5);
\end{tikzpicture}
    \caption{A JIT decoder is a decoder that decides on a correction at time $t$ given access to syndrome information acquired until $t$ only.
    We devise a JIT decoder from a global decoder $\sD$ defined on the full decoding graph.
    The resulting JIT decoder calls $\sD$ two times on the decoding graph with modified weights.
    In the first step a \emph{current estimate} is obtained from $\sD$ on the currently available syndrome information. In the second step previous corrections and the new best estimate are \emph{merged} by solving a matching problem on a single time-slice of the decoding graph.
    The figure illustrates the workings of this JIT matching decoder step-by-step on a representative error configuration in a 2D toy picture of the actual decoding graph.
    At each timestep the JIT decoder first performs an estimate of the error based on the currently available syndrome information and then merges this with the endpoints of the corrections done in the previous step.
    This guarantees that all syndromes are paired up after the execution of the circuit.
    We see that while the JIT decoder fails to identify the minimum-weight correction, it matches up the syndromes within each connected component of the error chain.}
    \label{fig:jit}
\end{figure*}

\myparagraph{Formal definition}
We define a \emph{just-in-time} (JIT) version $\sDJIT$ of a global matching decoder $\sD$.
For the purpose of this section, $\sD$ can be thought of as a \emph{minimum-weight perfect-matching} (MWPM) decoder that can take any (weighted) graph $G = (V,E)$ and vertex set $\sigma\subseteq V$ as an input and returns the minimum-weight edge set $\varepsilon\subseteq E$ whose boundary equals $\sigma$.
Throughout this section we identify $\varepsilon$ and $\sigma$ as vectors in $\bZ_2^E$ and $\bZ_2^V$, respectively.

The decoding graph $G$ is predefined by the circuit and the error model.
We decompose the vertices in $G$ into \emph{equal-time slices}
\begin{align}
    V = \bigsqcup_t V_t,
\end{align}
where each $V_t$ denotes the vertices of a connected subgraph of $G$ that is only connected to vertices in $V_{t-1}$ and $V_{t+1}$ via edges in $G$.
We call edges within $V_t$ \emph{spacelike at time $t$}, the edges between $V_{t-1}$ and $V_t$ \emph{past-oriented at time $t$} and the edges between $V_{t}$ and $V_{t+1}$ \emph{future-oriented at time $t$}.

At time $t$ a JIT decoder only has access to the syndrome in $V_{\leq t}$ and must decide on a correction on spacelike and future-oriented edges.
\begin{defn}
    At time $t$, we define the \emph{accessible subgraph} $\tilde{G}_{\leq t}$ as the graph with vertex set
    \begin{align}
        \tilde{V}_{\leq t} = V_{\leq t} \sqcup \{w\},
    \end{align}
    where $V_{\leq t} = \bigsqcup_{\tau\leq t} V_\tau$ and $w$ a newly introduced vertex.
    The edge-set of $G^{(t)}$ is given by 
    \begin{align}
        \tilde{E}_{\leq t} = E_{\leq t} \sqcup \{(v,w)\}_{v\in V_t},
    \end{align}
    where $E_{\leq t}$ is the edge set in $G$ within $V_{\leq t}$.
    The weights of the edges $\{(v,w)\}_{v\in V_t}$ are set to 1.
\end{defn}

\begin{remark}
    The newly introduced vertex $w$ stands in for all the syndrome vertices that lie the future. Instead of adding a single vertex one could equivalently define $\tilde{G}_{\leq t}$ as a graph with the same connectivity as $G$ but with all weights of edges that are not connected to $V_{\leq t}$ set to 0.
\end{remark}

This defines the accessible information of $\sDJIT$ at $t$.
The output of $\sDJIT$ is a set of edges in a commit region.

We define the \emph{commit graph} $\tilde{G}_{t}$ at $t$ as the graph whose edges are the union of spacelike and future-oriented edges at $t$,
\begin{align}
    C_{t} = E_t \sqcup E_{(t,t+1)},
\end{align}
where $E_{(t,t+1)}$ is the set of edges in $G$ between $V_t$ and $V_{t+1}$.

We are now in position to define $\sDJIT$ from a global decoder $\sD$, defined for any weighted graph $G=(V,E)$ with weight-vector $w\in \bR^E$.
$\sD$ is assumed to return $\argmin_{\varepsilon\in \bZ_2^E}(\norm{\varepsilon}\;|\; \partial\varepsilon = s)$, where $\norm{\varepsilon} = \sum_{i}\varepsilon_i w_i$.

\begin{defn}
Let $\sD$ be a global $\bZ_2$ matching decoder as defined above.
Its \emph{JIT-version} $\sDJIT$ consists of a sequence of classical subroutines $\{\sDJIT(t)\}_t$.
Each $\sDJIT(t)$ returns a correction $e_{\mathrm{c}}(t) \in\bZ_2^{C_{t}}$ from the two inputs:
\begin{itemize}
    \item the portion of the syndrome $s$ that is accessible at time $t$, $s_{\leq t}\in \bZ_2^{V_{\leq t}}$,
    \item the output of $\sDJIT(t-1)$, $e_{\mathrm{c}}(t-1)\in \bZ_2^{C_{t-1}}$.
\end{itemize}
The subroutine $\sDJIT(t)$ is composed of two steps and defined in Algorithm~\ref{alg:DJIT-main}.
Over all time steps $\sDJIT$ returns the correction $e_{\mathrm{c}}^{\mathrm{JIT}} = \sum_t e_{\mathrm{c}}(t) \in E$. By construction, its syndrome will match the observed syndrome which we prove explicitly in Appendix~\ref{app:generalJIT} for a general JIT decoder.
\end{defn}

\begin{algorithm}
\caption{Pair-matching $\sDJIT(t)$}
\label{alg:DJIT-main}
\DontPrintSemicolon
\KwIn{syndrome $s_{\leq t}\in \bZ_2^{V_{\leq t}}$, output of $\sDJIT(t-1)$ $e_{\mathrm{c}}(t-1)\in \bZ_2^{C_{t-1}}$}
\KwOut{correction $e_{\mathrm{c}}(t)\in \bZ_2^{C_t}$}
\hrulefill\;
Let $s^{\mathrm{tot}} = \sum_{v\in V_{\leq t}} s_{\leq t}|_v \in \bZ_2$ \; 
Let $s_w = (0\dots0,s_t^{\mathrm{tot}})\in \bZ_2^{V_{\leq t}}\oplus \bZ_2^{\{w\}} = \bZ_2^{\tilde{V}_{\leq t}}$\;\; 
\tcp{Step 1: global estimate}
\nl Run $\sD$ on $\tilde{G}_{\leq t}$ with syndrome $s_t + s_w$, get $e_1\in \bZ_2^{\tilde{E}_{\leq t}}$\;
Let $\tilde{s}_t = \eval{\partial (e_1)}_{V_t \sqcup V_{t+1}}\in \bZ_2^{V_t}\oplus \bZ_2^{V_{t+1}}$\;
Let $\tilde{s}_{t-1} = \eval{\partial (e_{\mathrm{c}}(t-1))}_{V_t \sqcup V_{t+1}}\in \bZ_2^{V_t}\oplus \bZ_2^{V_{t+1}}$\;\;
\tcp{Step 2: merge}
\nl Run $\sD$ on $\tilde{G}_{t}$ with syndrome $\tilde{s}_{t} + \tilde{s}_{t-1}$, get $e_{\mathrm{c}}(t)\in \bZ_2^{C_t}$\;
\Return{$e_{\mathrm{c}}(t)$}
\end{algorithm}

\myparagraph{Threshold simulation}
Here, we examine the threshold that the JIT decoder achieves for correcting Pauli-$X$ errors and $Z$-measurement errors in the TQD circuit, and compare it to the threshold of the ordinary toric code under the same error model.
Concretely, this means that we compare the $\sDJIT$ with $\sD$ on a three-dimensional cubic decoding graph.
In both cases the cubic lattice decoding graph emerges when considering the phenomenological error models with
\begin{itemize}
    \item a Pauli $X$ error with probability $p_{\mathrm{phys}}$ prior to every round of stabilizer measurement and
    \item a bit-flip error on every plaquette measurement with probability $p_{\mathrm{phys}}$.
\end{itemize}
Note that the chosen $X$ correction affects the resilience of the TQD circuit to Pauli-$Z$ errors, as explained in Sec.~\ref{sec:TQDcircuits}.
This leads to an additional difference in performance between the TQD circuit and the ordinary toric code.
We perform numerical simulations to estimate the threshold of different global decoding strategies against Pauli-$Z$ errors and $A$-measurement errors in Section~\ref{sec:charge-decoder}.

To simulate the $X$-error decoding, we perform Monte-Carlo simulations as explained in Sec.~\ref{sec:TQDcircuits}.
We consider an $L\times L$ spatial square lattice with periodic boundary conditions and $L$ time steps of noisy measurements.
To determine the logical class of the output state in our simulation, we continue applying the JIT decoder for $L/2$ additional error-free rounds of measurements, until all the syndromes in the $L^{\times 3}$-sized erroneous region are matched.
We believe that this method has smaller finite-size effects than applying a global matching decoder to the syndrome of the output state.
We declare failure if the cycle obtained from adding the error and the correction from the decoder is non-contractible.

Each error-sample is first decoded with a JIT-decoder $\sDJIT$ and then with a global decoder $\sD$ for comparison.
In Figure~\ref{fig:x_errors} we plot the two logical error-rate curves obtained from the two decoders for various system sizes up to $L=25$.
We can clearly identify a threshold behavior for both decoders.
The threshold and its statistical error are determined through finite-size scaling analysis by collapsing the data around the crossing points of the curves using the autoScale protocol~\cite{melchert_autoscale_2009} as implemented in the \verb|fssa| package~\cite{sorge_pyfssa_2015}.
We find
\begin{align}
    p_{\mathrm{th}}^{\mathrm{JIT}} = 2.51 \pm 0.31 \,\%\qq{and} p_{\mathrm{th}}^{\mathrm{glob}} = 2.90 \pm 0.12 \,\%.
\end{align}
The threshold of the global decoder is consistent with the literature~\cite{Wang2003statmech}.
Remarkably, the JIT-decoder only has a slightly lower threshold while operating with a fundamental limitation.

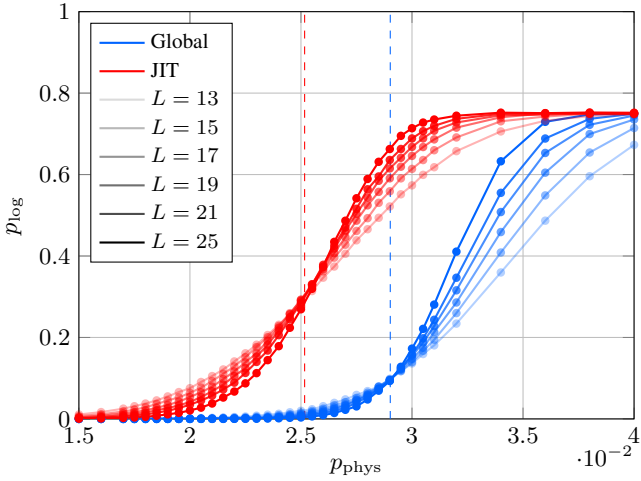
\begin{figure}
    \centering
  \begin{minipage}{0.5\textwidth}
    \centering
    \pgfplotsset{compat=1.18}

\begin{tikzpicture}
\begin{axis}[
    height=7cm,
    ymin=0, ymax=1.0,
    xmin=0.015, xmax=0.04,
    xlabel={$p_{\mathrm{phys}}$},
    ylabel={$p_{\mathrm{log}}$},
    grid=major,
    scaled x ticks=true,
    x tick scale label style={
        at={(1,0)}, 
        yshift=-10pt, 
        anchor=north east
    },
    legend style={draw=none, fill=none},
]
\addplot[global_color, opacity=0.3, thick,mark=*,mark size=1.2pt] coordinates {
  (0.015,5.31401e-05) (0.016,6.76329e-05) (0.017,0.0002657) (0.0175,0.000289855) (0.018,0.000555556) (0.0185,0.000594203) (0.019,0.000908213) (0.0195,0.00124155) (0.02,0.00137681) (0.0205,0.00198068) (0.021,0.00267633) (0.0215,0.003343) (0.022,0.0045314) (0.0225,0.00588406) (0.023,0.00789372) (0.0235,0.0103961) (0.024,0.0125894) (0.0245,0.0162621) (0.025,0.0199517) (0.0255,0.0237391) (0.026,0.0312029) (0.0265,0.0393913) (0.027,0.047913) (0.0275,0.0588502) (0.028,0.0682899) (0.0285,0.0802512) (0.029,0.0965894) (0.0295,0.116981) (0.03,0.136174) (0.0305,0.159353) (0.031,0.180647) (0.032,0.234082) (0.034,0.359597) (0.036,0.486738) (0.038,0.595971) (0.04,0.672942)
};
\addplot[global_color, opacity=0.44, thick,mark=*,mark size=1.2pt] coordinates {
  (0.015,1.95122e-05) (0.016,3.8835e-05) (0.017,8.25243e-05) (0.0175,0.00015534) (0.018,0.000194175) (0.0185,0.000315534) (0.019,0.000262136) (0.0195,0.000660194) (0.02,0.000902913) (0.0205,0.00113592) (0.021,0.00126214) (0.0215,0.00248544) (0.022,0.00294175) (0.0225,0.00359709) (0.023,0.00502913) (0.0235,0.00658738) (0.024,0.00917476) (0.0245,0.0116359) (0.025,0.015165) (0.0255,0.0195388) (0.026,0.0251854) (0.0265,0.0332864) (0.027,0.0412087) (0.0275,0.0514293) (0.028,0.0641951) (0.0285,0.0787561) (0.029,0.0982) (0.0295,0.117746) (0.03,0.139459) (0.0305,0.167107) (0.031,0.194537) (0.032,0.259805) (0.034,0.408595) (0.036,0.548785) (0.038,0.654141) (0.04,0.714268)
};
\addplot[global_color, opacity=0.58, thick,mark=*,mark size=1.2pt] coordinates {
  (0.015,1.95122e-05) (0.016,9.7561e-06) (0.017,4.87805e-05) (0.0175,2.43902e-05) (0.018,0.000126829) (0.0185,0.000107317) (0.019,0.000170732) (0.0195,0.000263415) (0.02,0.000395122) (0.0205,0.000434146) (0.021,0.000673171) (0.0215,0.00125366) (0.022,0.00154146) (0.0225,0.00236098) (0.023,0.00319512) (0.0235,0.00429756) (0.024,0.0062) (0.0245,0.00870732) (0.025,0.0117805) (0.0255,0.0157366) (0.026,0.0208683) (0.0265,0.0280245) (0.027,0.0356912) (0.0275,0.0470049) (0.028,0.0599755) (0.0285,0.07875) (0.029,0.0955441) (0.0295,0.120574) (0.03,0.147402) (0.0305,0.174549) (0.031,0.213172) (0.032,0.285873) (0.034,0.458812) (0.036,0.604629) (0.038,0.699453) (0.04,0.73524)
};
\addplot[global_color, opacity=0.72, thick,mark=*,mark size=1.2pt] coordinates {
  (0.015,0) (0.016,0) (0.017,0) (0.0175,1.96078e-05) (0.018,3.92157e-05) (0.0185,3.92157e-05) (0.019,6.86275e-05) (0.0195,0.000107843) (0.02,0.000176471) (0.0205,0.00027451) (0.021,0.000401961) (0.0215,0.000637255) (0.022,0.000921569) (0.0225,0.00151961) (0.023,0.00213861) (0.0235,0.0030198) (0.024,0.00405941) (0.0245,0.00662376) (0.025,0.009) (0.0255,0.0124356) (0.026,0.0173465) (0.0265,0.0241287) (0.027,0.0325644) (0.0275,0.0450891) (0.028,0.0575644) (0.0285,0.07451) (0.029,0.09495) (0.0295,0.12244) (0.03,0.15479) (0.0305,0.1876) (0.031,0.22802) (0.032,0.31574) (0.034,0.50793) (0.036,0.65305) (0.038,0.72214) (0.04,0.74498)
};
\addplot[global_color, opacity=0.86, thick,mark=*,mark size=1.2pt] coordinates {
  (0.015,0) (0.016,0) (0.017,0) (0.0175,0) (0.018,0) (0.0185,5e-05) (0.019,9e-05) (0.0195,4e-05) (0.02,3e-05) (0.0205,0.00023) (0.021,0.00023) (0.0215,0.00032) (0.022,0.00045) (0.0225,0.00089) (0.023,0.00127) (0.0235,0.00214) (0.024,0.00339) (0.0245,0.00451) (0.025,0.00664) (0.0255,0.00976) (0.026,0.01384) (0.0265,0.01954) (0.027,0.02745) (0.0275,0.03958) (0.028,0.05389) (0.0285,0.07357) (0.029,0.09466) (0.0295,0.12461) (0.03,0.15824) (0.0305,0.19792) (0.031,0.2434) (0.032,0.34666) (0.034,0.55508) (0.036,0.68848) (0.038,0.7368) (0.04,0.74938)
};
\addplot[global_color, opacity=1.0, thick,mark=*,mark size=1.2pt] coordinates {
  (0.015,0) (0.016,0) (0.017,0) (0.0175,0) (0.018,0) (0.0185,0) (0.019,1e-05) (0.0195,0) (0.02,1e-05) (0.0205,5e-05) (0.021,8e-05) (0.0215,7e-05) (0.022,0.00021) (0.0225,0.00024) (0.023,0.00048) (0.0235,0.00093) (0.024,0.00151) (0.0245,0.00267) (0.025,0.00352) (0.0255,0.00612) (0.026,0.00916) (0.0265,0.0146) (0.027,0.02238) (0.0275,0.0307) (0.028,0.04774) (0.0285,0.06933) (0.029,0.09311) (0.0295,0.12792) (0.03,0.17259) (0.0305,0.22113) (0.031,0.28071) (0.032,0.41072) (0.034,0.63257) (0.036,0.72946) (0.038,0.7476) (0.04,0.74994)
};
\draw[global_color, dashed] (axis cs:0.0290213,0) -- (axis cs:0.0290213,1);
\addplot[JIT_x_color, opacity=0.3, thick,mark=*,mark size=1.2pt] coordinates {
  (0.015,0.0108696) (0.016,0.016744) (0.017,0.0252464) (0.0175,0.0311546) (0.018,0.037715) (0.0185,0.0454348) (0.019,0.0552367) (0.0195,0.0658841) (0.02,0.0754686) (0.0205,0.0906425) (0.021,0.105193) (0.0215,0.121469) (0.022,0.141092) (0.0225,0.161097) (0.023,0.182155) (0.0235,0.20599) (0.024,0.231164) (0.0245,0.257199) (0.025,0.285043) (0.0255,0.313961) (0.026,0.346821) (0.0265,0.374145) (0.027,0.405425) (0.0275,0.438517) (0.028,0.465802) (0.0285,0.493841) (0.029,0.521638) (0.0295,0.551546) (0.03,0.573459) (0.0305,0.599251) (0.031,0.618043) (0.032,0.657493) (0.034,0.705913) (0.036,0.732243) (0.038,0.743456) (0.04,0.747806)
};
\addplot[JIT_x_color, opacity=0.44, thick,mark=*,mark size=1.2pt] coordinates {
  (0.015,0.00688293) (0.016,0.0113107) (0.017,0.0182573) (0.0175,0.0226893) (0.018,0.0284951) (0.0185,0.0365631) (0.019,0.0451505) (0.0195,0.0529417) (0.02,0.063801) (0.0205,0.076835) (0.021,0.0928301) (0.0215,0.110646) (0.022,0.128806) (0.0225,0.151626) (0.023,0.176413) (0.0235,0.200592) (0.024,0.228175) (0.0245,0.259641) (0.025,0.291743) (0.0255,0.32501) (0.026,0.362415) (0.0265,0.395976) (0.027,0.427709) (0.0275,0.461527) (0.028,0.492927) (0.0285,0.530629) (0.029,0.558839) (0.0295,0.590673) (0.03,0.613859) (0.0305,0.635629) (0.031,0.659766) (0.032,0.691156) (0.034,0.730795) (0.036,0.742498) (0.038,0.749166) (0.04,0.749683)
};
\addplot[JIT_x_color, opacity=0.58, thick,mark=*,mark size=1.2pt] coordinates {
  (0.015,0.00433659) (0.016,0.00754634) (0.017,0.0136146) (0.0175,0.0165707) (0.018,0.0210098) (0.0185,0.0271805) (0.019,0.0342049) (0.0195,0.042278) (0.02,0.0536) (0.0205,0.0654293) (0.021,0.0791854) (0.0215,0.0976488) (0.022,0.116502) (0.0225,0.13878) (0.023,0.161376) (0.0235,0.191727) (0.024,0.223751) (0.0245,0.256912) (0.025,0.292922) (0.0255,0.330537) (0.026,0.367546) (0.0265,0.407162) (0.027,0.446505) (0.0275,0.486789) (0.028,0.524441) (0.0285,0.557461) (0.029,0.590745) (0.0295,0.619294) (0.03,0.645799) (0.0305,0.667632) (0.031,0.686363) (0.032,0.71524) (0.034,0.740688) (0.036,0.750109) (0.038,0.751129) (0.04,0.75197)
};
\addplot[JIT_x_color, opacity=0.72, thick,mark=*,mark size=1.2pt] coordinates {
  (0.015,0.00256863) (0.016,0.0047767) (0.017,0.00906863) (0.0175,0.0121569) (0.018,0.0155098) (0.0185,0.0212451) (0.019,0.0265196) (0.0195,0.033098) (0.02,0.0416569) (0.0205,0.0535882) (0.021,0.0679706) (0.0215,0.0841667) (0.022,0.105147) (0.0225,0.124657) (0.023,0.152772) (0.0235,0.182366) (0.024,0.216455) (0.0245,0.249525) (0.025,0.289624) (0.0255,0.330743) (0.026,0.374129) (0.0265,0.418277) (0.027,0.462396) (0.0275,0.503812) (0.028,0.543505) (0.0285,0.58468) (0.029,0.61506) (0.0295,0.64122) (0.03,0.66797) (0.0305,0.68889) (0.031,0.70772) (0.032,0.72783) (0.034,0.74561) (0.036,0.74712) (0.038,0.75041) (0.04,0.7487)
};
\addplot[JIT_x_color, opacity=0.86, thick,mark=*,mark size=1.2pt] coordinates {
  (0.015,0.00164) (0.016,0.00329) (0.017,0.00585) (0.0175,0.00824) (0.018,0.01062) (0.0185,0.01524) (0.019,0.01984) (0.0195,0.02652) (0.02,0.03491) (0.0205,0.045) (0.021,0.05706) (0.0215,0.07228) (0.022,0.08904) (0.0225,0.11292) (0.023,0.13997) (0.0235,0.16789) (0.024,0.20562) (0.0245,0.24482) (0.025,0.28462) (0.0255,0.329) (0.026,0.37832) (0.0265,0.42529) (0.027,0.47008) (0.0275,0.51588) (0.028,0.56386) (0.0285,0.59536) (0.029,0.63583) (0.0295,0.66279) (0.03,0.6887) (0.0305,0.70651) (0.031,0.71965) (0.032,0.73686) (0.034,0.74919) (0.036,0.75025) (0.038,0.7496) (0.04,0.74997)
};
\addplot[JIT_x_color, opacity=1.0, thick,mark=*,mark size=1.2pt] coordinates {
  (0.015,0.00055) (0.016,0.00113) (0.017,0.00238) (0.0175,0.00363) (0.018,0.00529) (0.0185,0.00778) (0.019,0.01122) (0.0195,0.01486) (0.02,0.02077) (0.0205,0.02766) (0.021,0.03782) (0.0215,0.05246) (0.022,0.06719) (0.0225,0.08736) (0.023,0.11241) (0.0235,0.1443) (0.024,0.17858) (0.0245,0.22382) (0.025,0.26921) (0.0255,0.31953) (0.026,0.37304) (0.0265,0.43472) (0.027,0.48661) (0.0275,0.54178) (0.028,0.58923) (0.0285,0.6313) (0.029,0.66271) (0.0295,0.69505) (0.03,0.71374) (0.0305,0.72767) (0.031,0.73512) (0.032,0.74448) (0.034,0.75198) (0.036,0.7508) (0.038,0.75237) (0.04,0.75127)
};
\draw[JIT_x_color, dashed] (axis cs:0.025165,0) -- (axis cs:0.025165,1);
\pgfplotsextra{
  \node[anchor=north west, inner sep=4pt]
    at (rel axis cs:0,1) {\pgfplotslegendfromname{legendXraw}};
}
\end{axis}
\begin{axis}[
  hide axis, xmin=0, xmax=1, ymin=0, ymax=1,
  legend style={legend to name=legendXraw,
    font=\footnotesize, draw=black, fill=white,
    cells={anchor=west}},
  legend columns=1,
]
\addlegendimage{global_color, thick}
\addlegendentry{Global}
\addlegendimage{JIT_x_color,  thick}
\addlegendentry{JIT}
\addlegendimage{black!14, thick}
\addlegendentry{$L=13$}
\addlegendimage{black!28, thick}
\addlegendentry{$L=15$}
\addlegendimage{black!42, thick}
\addlegendentry{$L=17$}
\addlegendimage{black!56, thick}
\addlegendentry{$L=19$}
\addlegendimage{black!70, thick}
\addlegendentry{$L=21$}
\addlegendimage{black!100, thick}
\addlegendentry{$L=25$}
\end{axis}
\end{tikzpicture}
    \hfill 
\end{minipage}
    \caption{Estimated logical error rates $p_{\mathrm{log}}$ as a function of the physical error rate $p_{\mathrm{phys}}$ for various system sizes $L$.
    The JIT-decoder yields the logical error rates shown in red and indicate a threshold value around $2.5\,\%$.
    We compare that to a global decoder, whose logical error rates are shown in blue.
    We find a threshold for the global decoder of around $2.9\,\%$, in alignment with the literature~\cite{Wang2003statmech}.
    The dashed vertical lines indicate the thresholds $p_{\mathrm{th}}$ as extracted via the FSS ansatz.
    Note that the logical error-rate approaches 0.75 for $p\gg p_{\mathrm{th}}$ since there are two independent non-trivial cycles in the decoding graph.}
    \label{fig:x_errors}
\end{figure}

\myparagraph{Scaling below threshold}
To further investigate the qualitative behaviour of our JIT-matching decoder we extrapolate the logical error-rate scaling in different parameters below threshold.

\paragraph*{Scaling in $p$, effective distances}
To estimate the \emph{effective distance} of the JIT-decoder we analyze the scaling of logical error rate $p_{\mathrm{log}}$ with $p\ll p_{\mathrm{th}}$ for fixed $L$.
We fit the data to an ansatz
\begin{align}\label{eq:ansatz-deff}
    p_{\mathrm{log}}(p) = C\cdot p^{d_{\mathrm{eff}} }.
\end{align}
from which we extract $d_{\mathrm{eff}}$, the lowest weight error chain that the decoder fails to correct successfully.
We consider $L=7$ for three different physical error rates around $0.1p_{\mathrm{th}}$
for the JIT decoder as well as for the global decoder for comparison.

\begin{figure}
    \centering
  \begin{minipage}{0.5\textwidth}
    \centering
    \input{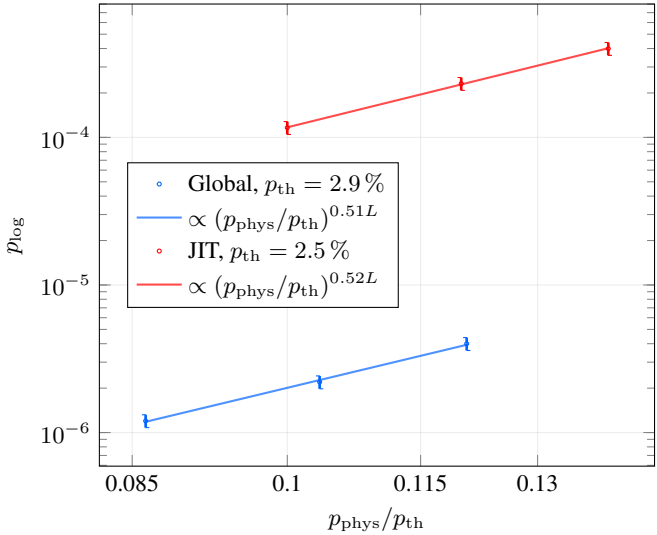}
    \hfill 
\end{minipage}
    \caption{
    We estimate the effective distance of the JIT decoder $\sDJIT$ and compare it to the global minimum-weight decoder $\sD$ in the regime where $p_{\mathrm{phys}}$ is an order of magnitude below $p_{\mathrm{th}}$ by fitting the ansatz in Equation~\eqref{eq:ansatz-deff}. The data is obtained for $L=7$.
    The fitted parameters suggest that both the JIT and the global decoder successfully correct all errors up to weight $L/2$.
    While this is expected behavior for the global decoder, the JIT-decoder could in principle fail on errors of lower weight.
    Our simulations suggest that this is not the case and one should hence expect the JIT-decoder to lead to the same asymptotic error suppression for low $p_{\mathrm{phys}}$ as the global decoder.
    }
    \label{fig:effective_len_L5}
\end{figure}

\begin{figure}
    \centering
  \begin{minipage}{0.5\textwidth}
    \centering
    \input{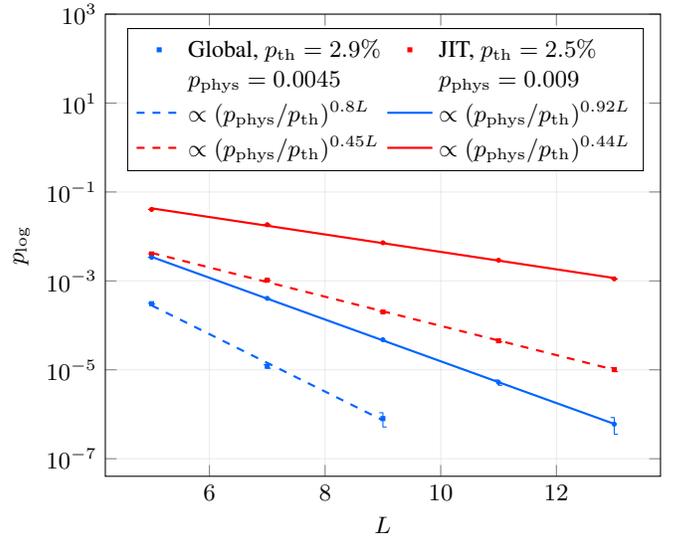}
    \hfill 
\end{minipage}
    \caption{
    Logical error rate $p_{\mathrm{log}}$ as a function of system size $L$ at $p_\mathrm{phys}=0.45\%, 0.9\%$. This is around 18\%, 35\,\% of the JIT decoder's and around 18\%,35\,\% of the JIT decoder's threshold, respectively.
    The scaling parameter in the exponent is extracted for each decoder by fitting to Equation~\eqref{eq:eff_len}.
    The data fits the ansatz well, verifying an exponential suppression of the logical error rate in $L$.
    The JIT decoder shows a factor of $\sim 2$  difference in the scaling in $L$ compared to the global decoder at $p_\mathrm{phys}\approx 0.35\,p_\mathrm{th}$, with a slightly better scaling of $\sim 1.8$ when $p_\mathrm{phys}\approx0.18\,p_\mathrm{th}$.
    This difference between JIT and global decoder would translate to factors of $\sim 8, \sim 5.83$ in the spacetime overhead in $L$ for a target logical error rate, at the simulated physical error rates.
    }
    \label{fig:effective_len_p009}
\end{figure}

The data and fits are shown in Figure~\ref{fig:effective_len_L5}.
Both the JIT and the global decoder show a similar scaling behavior an we extract an effective distance of $\approx L/2$ for both decoding strategies.
This indicates that the JIT-decoder does not fail on lower weight error configurations than the global decoder does and can lead to the same asymptotic error suppression for low $p_{\mathrm{phys}}$.

\paragraph*{Scaling in $L$, overheads at fixed $p_\mathrm{phys}$}
To estimate the spatial overhead needed to reach a target logical error rate for a fixed $p_{\mathrm{phys}}\leq p_{\mathrm{th}}$ we simulate the decoder at fixed $p_\mathrm{phys}=0.45\%,0.9\%$ for various system sizes between 5 and 13.
This corresponds to roughly 18\%, 35\,\% of the JIT threshold, respectively.
We fit the data to the ansatz
\begin{align}\label{eq:eff_len}
p_{\mathrm{log}}(L;p_\mathrm{phys}) = C\cdot p_\mathrm{phys}^{\alpha L}
\end{align}
both for the JIT and the global decoder.
The data and fit is shown in Figure~\ref{fig:effective_len_p009}.
We obtain 
\begin{subequations}
\begin{align}
\begin{split}
\text{at }p_\mathrm{phys} = 0.9\,\%\colon& \\
    \alpha^{\mathrm{JIT}} =& 0.442\pm 0.012 \qq{and}\\
    \alpha^{\mathrm{global}} =& 0.924\pm0.003,
\end{split}\\[6pt]
\begin{split}
\text{at }p_\mathrm{phys} = 0.45\,\%\colon& \\
    \alpha^{\mathrm{JIT}} =&0.446\pm 0.011\qq{and}\\
    \alpha^{\mathrm{global}} =& 0.800\pm0.003.
\end{split}
\end{align}
\end{subequations}
The data for both decoders agrees well with the expected exponential suppression.
We find that the global decoder shows a larger suppression of the logical error rate when increasing $L$ for the simulated error rates.
Specifically, we observe roughly a factor of $2$ decrease in $\alpha$ when using a JIT decoder instead of a global decoder at $p_\mathrm{phys}=0.9\%$.
This points to an increase in spacetime overhead by a factor of $\sim 2^3=8$ to reach a target logical error rate when doing JIT decoding instead of global decoding.
The difference in $\alpha$ becomes better in $p_\mathrm{phys}=0.45\%$, where the factor reduces to 1.8 and the spacetime overhead factor becomes $\sim 1.8^3=5.83$.
It would be interesting to explore if this trend continues for lower physical error rates and if the JIT decoder approaches the same value for $\alpha$ as the global decoder, $p_{\mathrm{phys}}\ll p_{\mathrm{th}}$ in Fig.~\ref{fig:effective_len_L5}.
This could come from the fact that the residual flux configuration after decoding and correction is systematically larger for the JIT decoder compared to the global one since a space-like error chain gets extended into the time direction by the JIT-decoding scheme.
We leave further exploration of the logical error rate scaling of JIT decoders in $L$ to future work.

\section{Mitigating twisted errors}\label{sec:charge-decoder}

In TQD circuits only the $X$-like errors must be JIT decoded and the $Z$-like errors can be decoded globally, after the execution of all the necessary syndrome extraction rounds.
Hence, the logical error rate of the entire protocol implemented by the TQD circuit depends on the success of this second, global, decoding step.

As discussed in Section~\ref{sec:TQDcircuits} any residual $X$-like error chain in the TQD circuit causes additional ``twisted'' $Z$ errors that might significantly decrease the performance of a protocol consisting of TQD circuits.
To mitigate this effect, we devise a simple method to incorporate the knowledge about previously committed $X$ corrections into the $Z$ decoding graph.

In this section, we first introduce a naive strategy that treats twisted errors as errors that are partially heralded and then estimate the threshold of the entire TQD circuit under phenomenological noise and compare it to a related protocol that uses three-dimensional connectivity where the $X$ corrections can be decoded globally as well.

\myparagraph{Partial heralding}
Using the analysis in Sec.~\ref{sec:TQDcircuits}, we can partially predict the twisted errors to mitigate their impact on the logical error rate.
The twisted errors depend on the flux configuration $b=b_e+b_m+b_c$.
However, the $Z$ decoder does not know the precise location of the error $b_e$, only the location of the corrections $b_c$ and measurement outcomes $b_m$.
To deal with this, we introduce a heuristic approach termed \emph{completing-the-loop (CL)}, which uses a global decoder to guess an error configuration $b_e$.

Given the estimated $b$, the twisted error $\widetilde M^b(\lambda,v)$ for each color-vertex pair $(\lambda,v)$ occurs with probability $\frac12$.
This corresponds to adding a weight-0 hyperedge to the decoding graph that is adjacent to the vertices in the syndrome of the twisted error $M^b(\lambda,v)=d^T\widetilde M^b(\lambda,v)$.
As we want to preserve the graph structure of the decoding problem to be able to use a matching decoder we resort to a simple heuristic:
For each edge that occurs in a twisted error, we set the weight of this individual edge to $0$.
In other words, we ignore the correlations between different edges that the twisted errors induce.
We call this part of the heuristic \emph{graph-reduction} (GR).
We summarize this simple heuristic in Algorithm~\ref{alg:completingtheloop}.

\begin{algorithm}
\caption{Completing-the-loop + graph-reduction (CL-GR) heuristic}
\label{alg:completingtheloop}
\DontPrintSemicolon
\KwIn{Observed $X$ syndrome $s_x$, JIT correction $b_{\mathrm{JIT}}$, $Z$ decoding graph}
\KwOut{Reweighted $Z$ decoding graph}
\hrulefill\;
\nl Feed $s_x$ to a global decoder, obtain a correction $b_{\mathrm{glob}}$\;
\nl Define estimated flux configuration $b=b_{\mathrm{JIT}} + b_{\mathrm{glob}} \mod 2$\;
\For{each color-vertex pair $(\lambda,v)$}{
    \nl Consider the corresponding twisted error $\widetilde M^b(\lambda,v)$\;
    \nl Set the weight for each individual color-edge pair of the twisted error to zero\;
    }
\Return{Reweighted $Z$ decoding graph}
\end{algorithm}

To estimate the effectiveness of the heuristic we consider different error rates for $X$-like errors ($p_X$) and $Z$-like errors ($p_Z$).
If the twisted errors could be perfectly mitigated, the failure rate of the $Z$-decoder would be completely independent of $p_X$.
We numerically test the dependence of the $Z$ decoding threshold to $p_X$ with and without the CL-GR heuristic and show the results in Figure~\ref{fig:critical_pz}.
Our results demonstrate that mitigating the twisted errors works in principle as our heuristic strictly improves the threshold for all simulated error rates.
For both the unmodified decoding and the CL-GR heuristic, the $p_Z$ threshold is nearly unchanged for small values of $p_X$ but starts dropping as $p_X$ increases.
This drop is significantly mitigated by the CL-GR strategy and is shifted to higher $p_X$.

\begin{remark}
    Another strategy to turn the hypergraph into a graph in the second step of our heuristic could be pairing up the endpoints of the syndrome vertices of each twisted error by adding new weight-0 edges between the pairs.
    This would also yield a matchable graph and may perform slightly better.
    More generally, many different approaches to the hypergraph decoding problem encountered due to the twisted errors are possible and we leave a careful study to future work.
\end{remark}

\begin{figure}
\pgfplotsset{compat=1.18}

\begin{tikzpicture}
\begin{axis}[
  xlabel={$p_X$},
  ylabel={$p_{\mathrm{th}}^{(Z)}(p_X)$},
  xmin=0, xmax=0.022,
  grid=major,
  legend style={draw=black, fill=white, font=\footnotesize,
                cells={anchor=west}, anchor=north west,
                at={(rel axis cs:0.05,0.3)}, inner sep=4pt},
  legend columns=1,
  width=\linewidth,
]
\addplot[JIT_xz_color, thick, mark=*, mark size=1.5pt] coordinates {(0,0.0290213) (0.005,0.0283287) (0.01,0.0265099) (0.0105,0.0262931) (0.011,0.0259573) (0.012,0.0252735) (0.013,0.0246133) (0.014,0.023906) (0.015,0.0229947) (0.016,0.0218976) (0.018,0.0188532)};
\addlegendentry{naive matching}
\addplot[JIT_xz_heu_color, thick, mark=diamond*, mark size=2pt] coordinates {(0,0.0290213) (0.005,0.0285065) (0.01,0.0278942) (0.0105,0.0278627) (0.011,0.0278458) (0.012,0.0274684) (0.013,0.0272734) (0.014,0.027172) (0.015,0.0267268) (0.016,0.0263813) (0.018,0.0257785) (0.019,0.0244255) (0.02,0.0236732) (0.021,0.0224592) (0.022,0.0204952)};
\addlegendentry{CL+GR heuristic}
\addplot[gray, dashed, thick] coordinates {(0,0.0290213) (0.022,0.0290213)};
\addlegendentry{$p_{\mathrm{th}}^{\mathrm{global}}=0.0290$}
\end{axis}
\end{tikzpicture}
  \caption{We benchmark the effectiveness of the heuristic based on partial heralding of the twisted errors by estimating the threshold of the $Z$ decoder for various error rates for $X$-like errors $p_X$.
  The $X$-like errors are JIT-decoded and the corrections together with the actual errors cause twisted $Z$ errors in addition to the physical i.i.d.\ $Z$-like errors.
  For each data point, $p_X$ is fixed and implicitly determines the statistics of twisted errors from corrections of a JIT decoder.
  Then, the $Z$-thresholds $p_{\mathrm{th}}^{(Z)}(p_X)$ are determined using samples form various $p_Z$ and $L$.
  The purple curve shows the $Z$-threshold for various $p_X$ without any heralding.
  We see a clear dependence of the $Z$-threshold on $p_X$.
  The red curve shows the $Z$-threshold for the same values of $p_X$ using the completing-the-loop + graph-reduction (CL-GR) heuristic to take advantage of the partial heralding of the twisted errors.
  While it still depends on $p_X$ we see a clear improvement in the threshold for all simulated error rates.}
  \label{fig:critical_pz}
\end{figure}
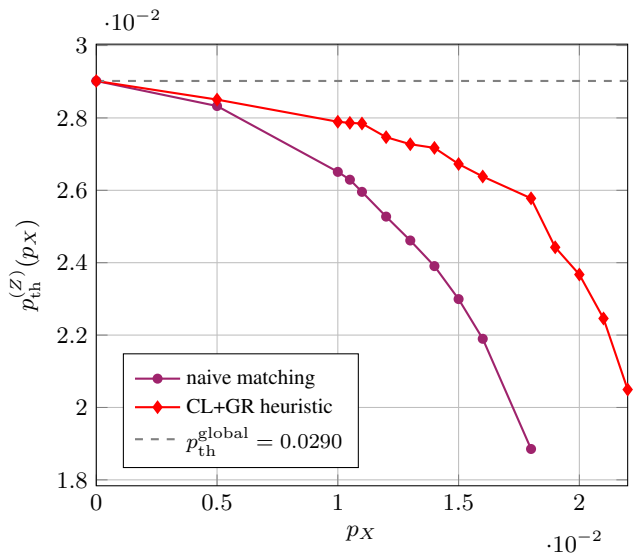

\myparagraph{Threshold results for TQD circuit}
Having benchmarked the JIT decoder on its own and devised a strategy to mitigate the twisted errors we are in position to estimate the threshold of the entire TQD circuit.
We consider the phenomenological noise model discussed in Section~\ref{sec:TQDcircuits} with $p_X=p_Z=p_{\mathrm{phys}}$ and use a JIT-matching decoder from Section~\ref{sec:JIT-decoder} to obtain the $X$ corrections.
The JIT corrections are then passed to a global $Z$-decoder that can use that information to mitigate the twisted errors.
We simulate the entire protocol for system sizes up to $L=13$ and declare logical failure if one of the two decoding steps fails, i.e. the flux configuration contains non-contractible loops or the charge configuration is in a non-trivial homology class.
The simulation is summarized in Algorithm~\ref{alg:benchmarking}.

2+1D TQD circuits can implement the same logical gate on a 2D code as a related 3+0D dimensional jump protocol~\cite{bombin2016dimensionaljump}. Such a protocol relies on three-dimensional connectivity and only a constant number of syndrome extraction rounds.
The main difference with regards to decoding is that in this 3+0D protocol the $X$ corrections can be chosen by a global decoder.
While the statistics of twisted errors is the same for the 3+0D protocol, the global decoding leads to a smaller flux configuration $b$, and thus to a smaller reduction of the $Z$-type threshold.
To complement our results we also perform threshold estimates for a 3+0D protocol by sampling twisted errors from a correction obtained from that global decoder instead of a JIT decoder.
This shows by how much the threshold is reduced due to the spacetime tradeoff that requires a JIT decoder.

In Figure~\ref{fig:threshold_xz} we show our simulation results and estimated threshold for the TQD circuit decoded with the JIT decoder and the CL-GR heuristic.
Using a finite-size-scaling ansatz as in Section~\ref{sec:JIT-decoder} we estimate a threshold of
\begin{align}
    p_{\mathrm{th}} = 2.17 \pm 0.17 \,\%,
\end{align}
which is only slightly lower than the memory threshold of the usual toric code under the same error model.
For comparison we also analyze the threshold for a naive $Z$ decoder that does not use any heralding strategy.
We also estimate the threshold for when the $X$ corrections are chosen with a global decoder as would be the case for a 3+0D protocol.
We do not use heralding to mitigate the $Z$ errors in the 3+0D case.
We believe such mitigation would also lead to a slight improvement of the 3+0D $Z$ threshold, though we note that the ``completing-the-loop'' step no longer applies in this case.
For naive, unheralded matching and 3+0D decoding we find thresholds of 
\begin{subequations}
\begin{align}
    p_{\mathrm{th}}^{\mathrm{naive}} =& 1.81 \pm 0.11 \,\% \qq{and} \\ 
    p_{\mathrm{th}}^{\mathrm{3+0D}} =& 2.30 \pm 0.30 \,\%.
\end{align}    
\end{subequations}
The threshold is highest for when the $X$ corrections are chosen globally and lowest when no heralding is used, according to our expectation.
This indicates that the twisted errors are the main reason why the threshold in these protocols is reduced compared to a toric-code memory experiment.
We also expect that the threshold could potentially be improved further with a more sophisticated heralding strategy.

\begin{figure}
    \centering
  \begin{minipage}{0.5\textwidth}
    \centering
\pgfplotsset{compat=1.18}

\begin{tikzpicture}
\begin{axis}[
    height=7cm,
    ymin=0, ymax=1.0,
    xmin=0.01, xmax=0.029,
    xlabel={$p_{\mathrm{phys}}$},
    ylabel={$p_{\mathrm{log}}$},
    grid=major,
    scaled x ticks=true,
    x tick scale label style={
        at={(1,0)}, 
        yshift=-10pt, 
        anchor=north east
    },
    legend style={draw=none, fill=none},
]

\addplot[JIT_xz_color, opacity=0.65, thick,/tikz/mark=*,/tikz/mark size=1.2pt] coordinates {
  (0.01,0.0075) (0.011,0.011) (0.012,0.0216) (0.013,0.0407) (0.014,0.0643) (0.015,0.1125) (0.016,0.1736) (0.017,0.2643) (0.018,0.3856) (0.019,0.5141) (0.02,0.6491) (0.021,0.7611) (0.022,0.8524) (0.023,0.9154) (0.024,0.9537) (0.025,0.976) (0.026,0.9885) (0.027,0.9931) (0.028,0.9961)
};
\addplot[JIT_xz_color, opacity=1.0, thick,/tikz/mark=*,/tikz/mark size=1.2pt] coordinates {
  (0.01,0.0025) (0.011,0.0058) (0.012,0.0111) (0.013,0.0225) (0.014,0.0441) (0.015,0.0815) (0.016,0.1454) (0.017,0.2474) (0.018,0.376) (0.019,0.5305) (0.02,0.6848) (0.021,0.8143) (0.022,0.9022) (0.023,0.9564) (0.024,0.9727) (0.025,0.99) (0.026,0.9945) (0.027,0.9964) (0.028,0.9973)
};
\draw[JIT_xz_color, dashed] (axis cs:0.0181406,0) -- (axis cs:0.0181406,1);
\addplot[JIT_xz_heu_color, opacity=0.3, thick,/tikz/mark=*,/tikz/mark size=1.2pt] coordinates {
  (0.01,0.0096) (0.011,0.0166) (0.012,0.0245) (0.013,0.039) (0.014,0.0562) (0.015,0.0846) (0.016,0.1277) (0.017,0.1728) (0.018,0.2318) (0.019,0.3138) (0.02,0.4112) (0.021,0.5066) (0.022,0.6056) (0.023,0.701) (0.024,0.7824) (0.025,0.8596) (0.026,0.9083) (0.027,0.9458) (0.028,0.9696)
};
\addplot[JIT_xz_heu_color, opacity=0.65, thick,/tikz/mark=*,/tikz/mark size=1.2pt] coordinates {
  (0.01,0.006) (0.011,0.0077) (0.012,0.0137) (0.013,0.0235) (0.014,0.0368) (0.015,0.0603) (0.016,0.089) (0.017,0.1278) (0.018,0.1927) (0.019,0.2796) (0.02,0.3719) (0.021,0.4937) (0.022,0.6093) (0.023,0.7291) (0.024,0.8195) (0.025,0.8997) (0.026,0.9457) (0.027,0.9761) (0.028,0.987)
};
\addplot[JIT_xz_heu_color, opacity=1.0, thick,/tikz/mark=*,/tikz/mark size=1.2pt] coordinates {
  (0.01,0.0018) (0.011,0.004) (0.012,0.0068) (0.013,0.0123) (0.014,0.0213) (0.015,0.0365) (0.016,0.0591) (0.017,0.0993) (0.018,0.1491) (0.019,0.2283) (0.02,0.3374) (0.021,0.4693) (0.022,0.6131) (0.023,0.7455) (0.024,0.8436) (0.025,0.9182) (0.026,0.9664) (0.027,0.9918) (0.028,0.9927)
};
\draw[JIT_xz_heu_color, dashed] (axis cs:0.021663,0) -- (axis cs:0.021663,1);
\addplot[three_zero_D_color, opacity=0.3, thick,/tikz/mark=*,/tikz/mark size=1.2pt] coordinates {
  (0.013,0.0026) (0.014,0.0049) (0.015,0.0091) (0.016,0.015) (0.017,0.0294) (0.018,0.046) (0.019,0.0689) (0.02,0.1024) (0.021,0.1415) (0.022,0.2023) (0.023,0.2766) (0.024,0.371) (0.025,0.4562) (0.026,0.5514) (0.027,0.6441) (0.028,0.7301) (0.029,0.8043)
};
\addplot[three_zero_D_color, opacity=0.65, thick,/tikz/mark=*,/tikz/mark size=1.2pt] coordinates {
  (0.013,0.0006) (0.014,0.0021) (0.015,0.0042) (0.016,0.0076) (0.017,0.0139) (0.018,0.0262) (0.019,0.046) (0.02,0.0759) (0.021,0.1238) (0.022,0.1857) (0.023,0.2768) (0.024,0.3749) (0.025,0.4922) (0.026,0.6123) (0.027,0.7302) (0.028,0.819) (0.029,0.8941)
};
\addplot[three_zero_D_color, opacity=1.0, thick,/tikz/mark=*,/tikz/mark size=1.2pt] coordinates {
  (0.013,0.0005) (0.014,0.0006) (0.015,0.0018) (0.016,0.0036) (0.017,0.0072) (0.018,0.0157) (0.019,0.032) (0.02,0.0617) (0.021,0.1066) (0.022,0.1715) (0.023,0.2712) (0.024,0.4001) (0.025,0.5366) (0.026,0.6833) (0.027,0.7962) (0.028,0.8772) (0.029,0.9394)
};
\draw[three_zero_D_color, dashed] (axis cs:0.0230473,0) -- (axis cs:0.0230473,1);
\pgfplotsextra{
  \node[anchor=north west, inner sep=4pt]
    at (rel axis cs:0,1) {\pgfplotslegendfromname{legendXZraw}};
}
\end{axis}
\begin{axis}[
  hide axis, xmin=0, xmax=2, ymin=0, ymax=1,
  legend style={legend to name=legendXZraw,
    font=\footnotesize, draw=black, fill=white,
    cells={anchor=west}},
  legend columns=1,
]
\addlegendimage{JIT_xz_color,     thick, mark=*}
\addlegendentry{2+1D no heralding}
\addlegendimage{JIT_xz_heu_color, thick, mark=*}
\addlegendentry{2+1D w/ CL+GR}
\addlegendimage{three_zero_D_color, thick, mark=*}
\addlegendentry{3+0D}
\addlegendimage{black!35, thick}
\addlegendentry{$L=9$}
\addlegendimage{black!70, thick}
\addlegendentry{$L=11$}
\addlegendimage{black!100, thick}
\addlegendentry{$L=13$}
\end{axis}
\end{tikzpicture}
    \hfill 
\end{minipage}
    \caption{
    Simulation results to estimate the threshold of the entire TQD circuit under a phenomenological noise model.
    We find a threshold of around 2.2\,\% when using a JIT-decoder for $X$ errors and the CL-GR strategy to mitigate the twisted errors (red).
    For comparison, we also show the data when no heralding strategy is used (purple) and when the $X$ corrections are determined with a global decoder (blue).
    We extract the threshold values via a finite-size-scaling analysis as in Section~\ref{sec:JIT-decoder} and indicate them with dashed vertical lines.
    Our results suggest that the twisted errors are the main reason for the reduction in threshold compared to a toric-code memory experiment.
    }
    \label{fig:threshold_xz}
\end{figure}
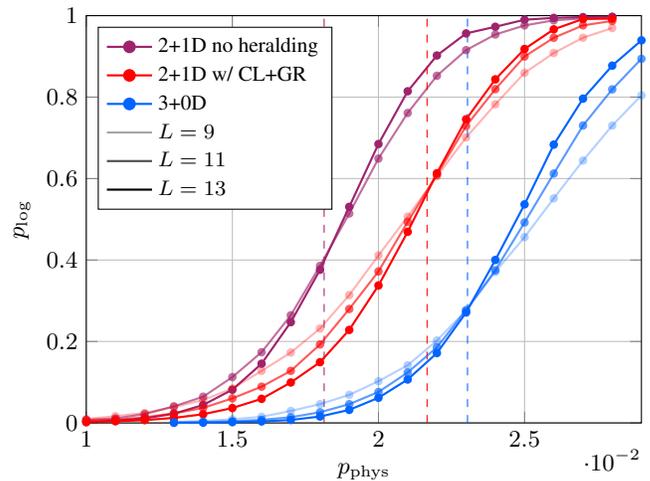

\section{Discussion and outlook}

\myparagraph{Discussion}
In this work, we have devised a just-in-time decoding strategy for a \emph{TQD code}, a 2D non-Pauli stabilizer code that can be used to complete a universal gate set on topological toric or color codes.
For this decoding strategy, we have demonstrated that it achieves high thresholds of above $2\,\%$ under a phenomenological noise model.
Indeed, these are comparable to the memory thresholds of 2D toric codes under the same noise model which might be surprising given previous estimates for a related protocol that uses a different JIT-decoder~\cite{Scruby2022JITnumerics, scruby2025nodistillation}.
This observation motivates the use of such codes for universal fault tolerant quantum computing.
We also extrapolate the logical error rate scaling below threshold and find that a JIT decoder achieves the same effective distance as a global decoder but has a slightly reduced error-suppression rate in the code distance for a fixed physical error rate.

To provide further perspective to this kind of decoding, the nature of the JIT decoder makes it particularly simple to interface with existing schemes that are decoded using a sliding-window decoder~\cite{psidecoding, riverlanedecoding}.
Both types of decoders can use the same global decoder as subroutines.
At every boundary to a commit region one can switch between a sliding-window decoding scheme that relies on propagating the correction through the circuit, and a JIT decoder that commits a correction with less information.

In order to analyze our protocols, we have developed a new framework for efficiently benchmarking error-correcting circuits composed of CSS-like Pauli measurements, controlled-NOT gates, and diagonal gates from the third level of the Clifford hierarchy.
The framework allows for a rigorous characterization of the measurement-outcome statistics for each error pattern, enabling efficient sampling from the resulting distribution and reliable assessment of decoding success.
Furthermore, by formalizing the additional random errors that arise due to residual errors in the circuit, we have devised a simple heuristic to mitigate these errors and demonstrate that it significantly improves the threshold of the TQD circuit.
Taken together, our results indicate that universal computation in 2D could have comparable performance both in terms of threshold and spatial overheads to Clifford logic in 2D.

\myparagraph{Outlook}
We show that protocols that require a JIT decoder can yield much lower logical error rates than previously expected for JIT-decoded schemes.
This directly translates into a significantly lower spacetime overhead when using 2D non-Pauli codes for non-Clifford logic, e.g., for magic-state preparation, and could achieve lower overheads than cultivation~\cite{gidney2024cultivation} combined with distillation~\cite{Litinski2019notascostly}.
The intrinsic scalability to arbitrary distances and topological nature of 2+1D non-Clifford protocols allow to smoothly interpolate between full error correction in a high distance code and (mild) post-selection in a lower distance code.
One end of this spectrum, where no JIT decoder is used, was explored recently in Ref.~\cite{hetenyi2026cultivationgauging} and found that the spacetime volume can indeed be reduced compared to more standard cultivation by measuring the stabilizers of a 2D non-Pauli code, even without any active decoding.
We believe that the flexibility to also incorporate active decoding into such protocols, together with the high thresholds reported in this work, will lead to a rich set of new schemes that can outperform existing protocols to achieve fault-tolerant universal quantum computation, especially for low target logical error rates.

A more precise resource estimate for a specific logic gate, e.g., a magic-state preparation block, will require a more detailed study and comparison among various families of protocols, given the new constructions of space-efficient 2+1D non-Clifford protocols and novelty of the post-selection techniques used in cultivation~\cite{Psi2024postselection, PsiMacromux}.
It would be interesting to benchmark other circuit realizations of the same topological protocol that was studied in this work.
At this point, there is a variety of microscopic TQD-like codes and circuits in the literature~\cite{davydova2025universal, Bauer2025allAbelian, kobayashi2025cliffordcodes, manjunath2026GSC, huang2026hybridsurgery}.
In particular, the measurement of stabilizers with $CZ$ terms requires $CCZ$ gates in the circuit, which may result in additional overhead when compiled into 2-qubit gates.
Therefore, 3D-color-code inspired protocols such as the brickwork code in Ref.~\cite{davydova2025universal} or the twisted color circuits in Ref.~\cite{bauer2025planar} may have better performance as they only require single qubit $T$ gates.

To get a more realistic estimate of the logical error in non-Clifford QEC circuits, it would be important to benchmark the syndrome-extraction circuit of a TQD code under a circuit-level noise model.
Our work makes this possible as we have formulated JIT decoding and the analysis of twisted errors in a way that works for any decoding graph.
Additionally, it is possible to include soft information into the decoding algorithm.

One potential challenge for JIT-decoded schemes is that they require both the decoding algorithm and the classical control to operate within a time that is similar to one round of syndrome extraction.
A slower reaction time would either mean introducing additional idling time, or computing the JIT correction based on less syndrome information.
For the latter, it would be interesting to systematically investigate how the JIT decoding performance is affected by having syndrome information available with a small delay or for a few steps in advance.
Our expectation is that the asymptotic performance and with that the threshold is largely unchanged when scaling both the delay window and the code distance. At the same time we expect the spacetime overhead to be very sensitive to the length of delay.

To reduce delay times, it would be advantageous to use faster decoding algorithms than MWPM, such as union-find decoders~\cite{Delfosse2021unionfind}, as subroutine of our JIT decoder, and analyze their performance.
It would also be highly desirable if one would not have to perform a completely new global decoding in each JIT time step, but reuse some of the past time step to speed up decoding.
Another promising avenue would be to adapt recently improved local decoders for topological Pauli codes~\cite{lake2025localQEC, paletta2026localdecodertoriccode}.

It would also be helpful to develop a more sophisticated mitigation strategy for the twisted errors to improve the global $Z$ decoding after JIT $X$ decoding.
Our CL-GR strategy of approximating the twisted hypergraph matching problem by a reweighted matchable graph can probably be improved further.
Further improvements may be achieved by using hypergraph suitable decoders such as belief-matching decoders~\cite{Higgott2023BPmatching}.

While we have used information from the JIT $X$ decoding to improve global $Z$ decoding, it is also possible to go the other way round and use the $Z$ syndrome to improve $X$ decoding:
Assuming there are no physical $Z$ errors, we know that the $Z$ syndrome is only caused by twisted errors from the residual flux configuration $b$.
This allows us to more accurately predict 
 the $X$-error configuration $b_e$.
When there are physical $Z$ errors, the prediction becomes less reliable, but may still improve performance.
A similar type of heralding was studied in Ref.~\cite{Jing2026heraldingD4} for a 2D decoding problem that arises when all measurements are noiseless, and the authors indeed found 
an improvement.
Note that this heralding is not possible in an associated 3+0D dimension-jump protocols as the $Z$ syndrome only becomes available after the $X$ corrections have been performed.
Thus, this approach can help closing the threshold-gap between the 3+0D and 2+1D protocols.

Recent work has generalized TQD codes to a more general \emph{quantum low density parity check}
(qLDPC) setting~\cite{zhu2026twisted-ldpc, christos2026twisted-ldpc}.
Benchmarking a JIT decoder for these codes would be an interesting direction to generalize our work beyond 2D.
Our framework to efficiently simulate associated syndrome-extraction circuits based on a path-integral formulation straightforwardly generalizes to these settings.

As a final note, while our work is mostly of conceptual and qualitative nature, it is expected to provide actionable advice to practitioners aiming to reduce overheads for non-Clifford logic in various physical architectures.
We provide clear steps and methods to obtain more realistic estimates in specific circuit-implementations and error models.
We hope that the present work inspires further research on assessing realistic resource estimates for 2D-local fault-tolerant quantum computation and beyond~\cite{GidneyRSA, IcebergRSA, CaltechRSA}.

\section*{Code availability}
The code that was used to produce the results presented in this paper is available at \url{https://github.com/NoaFeldman/JustInTimeDecoding.git}.

\section*{Author contributions}
NF performed all numerical simulations and threshold analyses.
JM and AB worked out the theoretical framework and designed the numerical experiments.
JM wrote most of this manuscript, with significant contributions from AB.
All authors contributed to the shaping of this project and helped in finalizing the draft.

\begin{acknowledgments}
We thank T. Scruby and B. Brown for helpful discussions and 
B.\ Brown and T.\ Ellison for valuable feedback on an earlier version of this manuscript.
Part of the research for this work were done while JM and AB were visiting the Kavli Institute for Theoretical Physical within the program ``Noise-robust Phases of Quantum Matter''.
JM is supported by the DFG (CRC 183).
AB is supported by the U.S. Army Research Laboratory and the U.S. Army Research Office under contract/grant number W911NF2310255, and by the U.S. Department of Energy, Office of Science, National Quantum Information Science Research Centers, and the Co-design Center for Quantum Advantage (C2QA) under contract number DE-SC0012704.
This research was supported in part by grant NSF PHY-2309135 to the Kavli Institute for Theoretical Physics (KITP).
NF is supported by the Israeli CHE fellowship for quantum science and technology.
JE is supported by the BMFTR (QSolid, PasQuops, Hybrid++, QuSol), the Quantum Flagship (Millenion, PasQuans2), Berlin Quantum, the Munich Quantum Valley, the DFG (SPP 2514, CRC 183), and the European Research Council (DebuQC). 

\end{acknowledgments}

\begin{appendix}

\section{Derivation of the TQD constraints and equivalences}
\label{sec:tqd_derivation}
In this appendix, we derive the constraints and equivalences in Eqs.~\eqref{eq:tqdb_conseq} and \eqref{eq:tqdc_conseq} by relating the TQD circuit to a TQD path integral.
For general literature on relating error-correcting circuits to path integrals, we refer 
the reader to Refs.~\cite{path_integral_qec, xy_floquet, Bauer2025allAbelian, davydova2025universal, bauer2025planar}.

\myparagraph{The cup product}
First, we need to introduce one additional operation on chains in lattice cohomology, the \emph{cup product}~\cite{Steenrod1947}, which allows us to concisely express the TQD path integral.
The cup product of a collection of chains $\{a_k\}_{0\leq k<n}$ of degrees $i_k$ is a chain of degree $i_0+\ldots+i_{n-1}$ that we denote by $a_0\cup \ldots\cup a_{n-1}$.
The value of $a_0\cup\ldots\cup a_{n-1}$ on an $i_0+\ldots+i_{n-1}$-cell $x$ depends only on the values of $a_0$, $\ldots$, $a_{n-1}$ on the subcells of $x$.
The cup product can be defined on any cellulation~\cite{Bauer2025allAbelian}.
Here, we only need its definition on a square or cubic lattice~\cite{Chen2021}:
For each edge, square, cube or hypercube $y$ let $y^{(0)}$ denote the vertex of $y$ with the coordinate $(0,0,\ldots)$, and let $y^{(1)}$ denote the vertex with coordinate $(1,1,\ldots)$.
If $y$ is a vertex, we set $y^{(0)}=y^{(1)}=y$.
Consider collections $z=(z_0, \ldots, z_{n-1})$, where $z_k$ is an $i_k$-subcell of an $i_0+\ldots+i_{n-1}$-dimensional hypercube $x$.
We call $z$ \emph{sequential} if $z_0^{(0)}=(0,0,\ldots)$, $z_k^{(1)}=z_{k+1}^{(0)}$ for all $0\leq k<n-1$, and $z_{n-1}^{(1)}=(1,1,\ldots)$, and we denote the set of all sequential collections for the hypercube $x$ by $\mathbf Z(x)$.
Then the value of the cup product on an  $i_0+\ldots+i_{n-1}$-cell $x$ is given by
\begin{equation}
(a_0\cup \ldots\cup a_{n-1})(x)
= \sum_{z\in \mathbf Z(x)} a_0(z_0) \cdots a_{n-1}(z_{n-1})\;.
\end{equation}
In other words, we take the products of the values of the $a_k$ on sequential collections of subcells, and sum over all of these collections.
In Figure~\ref{fig:cup_product} we illustrate explicit expressions for the cup product on the cubic lattice for a few relevant examples.

The cup product has a few important properties.
First of all, the cup product is associative, $(a\cup b)\cup c = a\cup (b\cup c)$ so there is no need to write brackets explicitly.
This follows straight-forwardly from the definition and the fact that sequential collections of sequential collections are again sequential collections.
Second, the cup product obeys the Leibnitz rule,
\begin{equation}
\label{eq:leibniz_rule}
d(a\cup b)=da\cup b+a\cup db\;,
\end{equation}
and it follows that the cup product between cocycles is again a cocycle.

\begin{figure}
\def\xcol{\textcolor{myred}{x}}
\def\ycol{\textcolor{mygreen}{y}}
\def\zcol{\textcolor{myblue}{z}}
\begin{tikzpicture}
\node (labs) at (2,3){
\begin{tikzpicture}
\node[inner sep=0.05cm] (100) at (2,0){$\scriptstyle{100}$};
\node[inner sep=0.05cm] (010) at (0,2){$\scriptstyle{010}$};
\node[inner sep=0.05cm] (001) at (1.4,0.8){$\scriptstyle{001}$};
\node[inner sep=0.05cm] (110) at ($(100)+(010)$){$\scriptstyle{110}$};
\node[inner sep=0.05cm] (011) at ($(001)+(010)$){$\scriptstyle{011}$};
\node[inner sep=0.05cm] (101) at ($(100)+(001)$){$\scriptstyle{101}$};
\node[inner sep=0.05cm] (111) at ($(100)+(010)+(001)$){$\scriptstyle{111}$};
\node[inner sep=0.05cm] (000) at (0,0){$\scriptstyle{000}$};
\tikzset{mlab/.style = {netmark={lab=$\scriptstyle{#1}$,p=0.5,asty={fill=white,inner sep=0}}}}
\draw (000)edge[mlab=\bullet00](100) (100)edge[mlab=1\bullet0](110) (110)edge[mlab=\bullet10](010) (010)edge[mlab=0\bullet0](000) (100)edge[mlab=10\bullet](101) (101)edge[mlab=1\bullet1](111) (111)edge[mlab=\bullet11](011) (011)edge[mlab=01\bullet](010) (110)edge[mlab=11\bullet](111);
\draw[dashed] (000)edge[mlab=00\bullet](001) (001)edge[mlab=\bullet01](101) (001)edge[mlab=0\bullet1](011);
\node[inner sep=0.05cm] (0xx) at ($0.5*(010)+0.5*(001)$){$\scriptstyle{0\bullet\bullet}$};
\node[inner sep=0.05cm] (1xx) at ($(100)+0.5*(010)+0.5*(001)$){$\scriptstyle{1\bullet\bullet}$};
\node[inner sep=0.05cm] (x0x) at ($0.5*(100)+0.5*(001)$){$\scriptstyle{\bullet0\bullet}$};
\node[inner sep=0.05cm] (x1x) at ($(010)+0.5*(100)+0.5*(001)$){$\scriptstyle{\bullet1\bullet}$};
\node[inner sep=0.05cm] (xx0) at ($0.5*(100)+0.5*(010)$){$\scriptstyle{\bullet\bullet0}$};
\node[inner sep=0.05cm] (xx1) at ($(001)+0.5*(100)+0.5*(010)$){$\scriptstyle{\bullet\bullet1}$};
\end{tikzpicture}
};
\node [anchor=east] at (labs.west){(a)};
\node at (0,0){
\begin{tikzpicture}
\draw (0,0)rectangle++(0.8,0.8);
\draw[line width=2,myred] (0,0)--++(0:0.8);
\draw[line width=2,myblue] (0:0.8)--++(90:0.8);
\fill[mygreen] (0:0.8)circle(0.1);
\node at (0.4,1.1){$\xcol_{\bullet0} \ycol_{10} \zcol_{1\bullet}$};
\end{tikzpicture}
};
\node at (1,0){+};
\node at (2,0){
\begin{tikzpicture}
\draw (0,0)rectangle++(0.8,0.8);
\draw[line width=2,myred] (0,0)--++(90:0.8);
\draw[line width=2,myblue] (90:0.8)--++(0:0.8);
\fill[mygreen] (90:0.8)circle(0.1);
\node at (0.4,1.1){$\xcol_{0\bullet} \ycol_{01} \zcol_{\bullet1}$};
\end{tikzpicture}
};
\node (c0) at (0,-2){
\begin{tikzpicture}
\mycube
\mycubefront
\draw[line width=2,mygreen] ($(x)+(y)$)--++(z);
\fill[myred] (0,0)--++(x)--++(y)--++($-1*(x)$)--cycle;
\fill[myblue] ($(x)+(y)+(z)$)circle(0.1);
\node at (0.7,1.4){$\xcol_{\bullet\bullet0} \ycol_{11\bullet} \zcol_{111}$};
\end{tikzpicture}
};
\node at (1,-2){+};
\node (c0) at (2,-2){
\begin{tikzpicture}
\mycube
\fill[myred] (0,0)--++(x)--++(z)--++($-1*(x)$)--cycle;
\mycubefront
\draw[line width=2,mygreen] ($(x)+(z)$)--++(y);
\fill[myblue] ($(x)+(y)+(z)$)circle(0.1);
\node at (0.7,1.4){$\xcol_{\bullet0\bullet} \ycol_{1\bullet1} \zcol_{111}$};
\end{tikzpicture}
};
\node at (3,-2){+};
\node (c0) at (4,-2){
\begin{tikzpicture}
\mycube
\fill[myred] (0,0)--++(y)--++(z)--++($-1*(y)$)--cycle;
\mycubefront
\draw[line width=2,mygreen] ($(y)+(z)$)--++(x);
\fill[myblue] ($(x)+(y)+(z)$)circle(0.1);
\node at (0.7,1.4){$\xcol_{0\bullet\bullet} \ycol_{\bullet11} \zcol_{111}$};
\end{tikzpicture}
};
\node (c0) at (0,-4){
\begin{tikzpicture}
\mycube
\mycubefront
\draw[line width=2,myred] (0,0)--++(x);
\fill[myblue] (x)--++(y)--++(z)--++($-1*(y)$)--cycle;
\fill[mygreen] (x)circle(0.1);
\node at (0.7,1.4){$\xcol_{\bullet00} \ycol_{100} \zcol_{1\bullet\bullet}$};
\end{tikzpicture}
};
\node at (1,-4){+};
\node (c0) at (2,-4){
\begin{tikzpicture}
\mycube
\mycubefront
\draw[line width=2,myred] (0,0)--++(y);
\fill[myblue] (y)--++(x)--++(z)--++($-1*(x)$)--cycle;
\fill[mygreen] (y)circle(0.1);
\node at (0.7,1.4){$\xcol_{0\bullet0} \ycol_{010} \zcol_{\bullet1\bullet}$};
\end{tikzpicture}
};
\node at (3,-4){+};
\node (c0) at (4,-4){
\begin{tikzpicture}
\mycube
\draw[line width=2,myred] (0,0)--++(z);
\fill[myblue] (z)--++(y)--++(x)--++($-1*(y)$)--cycle;
\fill[mygreen] (z)circle(0.1);
\mycubefront
\node at (0.7,1.4){$\xcol_{00\bullet} \ycol_{001} \zcol_{\bullet\bullet1}$};
\end{tikzpicture}
};
\node (c0) at (0,-6){
\begin{tikzpicture}
\mycube
\mycubefront
\draw[line width=2,myred] (0,0)--++(x);
\draw[line width=2,mygreen] (x)--++(y);
\draw[line width=2,myblue] ($(x)+(y)$)--++(z);
\node at (0.7,1.4){$\xcol_{\bullet00} \ycol_{1\bullet0} \zcol_{11\bullet}$};
\end{tikzpicture}
};
\node at (1,-6){+};
\node (c0) at (2,-6){
\begin{tikzpicture}
\mycube
\mycubefront
\draw[line width=2,myred] (0,0)--++(x);
\draw[line width=2,mygreen] (x)--++(z);
\draw[line width=2,myblue] ($(x)+(z)$)--++(y);
\node at (0.7,1.4){$\xcol_{\bullet00} \ycol_{10\bullet} \zcol_{1\bullet1}$};
\end{tikzpicture}
};
\node at (3,-6){+};
\node (c0) at (4,-6){
\begin{tikzpicture}
\mycube
\mycubefront
\draw[line width=2,myred] (0,0)--++(y);
\draw[line width=2,mygreen] (y)--++(x);
\draw[line width=2,myblue] ($(x)+(y)$)--++(z);
\node at (0.7,1.4){$\xcol_{0\bullet0} \ycol_{\bullet10} \zcol_{11\bullet}$};
\end{tikzpicture}
};
\node at (-1,-8){+};
\node (c0) at (0,-8){
\begin{tikzpicture}
\mycube
\mycubefront
\draw[line width=2,myred] (0,0)--++(y);
\draw[line width=2,mygreen] (y)--++(z);
\draw[line width=2,myblue] ($(z)+(y)$)--++(x);
\node at (0.7,1.4){$\xcol_{0\bullet0} \ycol_{01\bullet} \zcol_{\bullet11}$};
\end{tikzpicture}
};
\node at (1,-8){+};
\node (c0) at (2,-8){
\begin{tikzpicture}
\mycube
\draw[line width=2,myred] (0,0)--++(z);
\draw[line width=2,mygreen] (z)--++(x);
\mycubefront
\draw[line width=2,myblue] ($(x)+(z)$)--++(y);
\node at (0.7,1.4){$\xcol_{00\bullet} \ycol_{\bullet01} \zcol_{1\bullet1}$};
\end{tikzpicture}
};
\node at (3,-8){+};
\node (c0) at (4,-8){
\begin{tikzpicture}
\mycube
\draw[line width=2,myred] (0,0)--++(z);
\draw[line width=2,mygreen] (z)--++(y);
\mycubefront
\draw[line width=2,myblue] ($(z)+(y)$)--++(x);
\node at (0.7,1.4){$\xcol_{00\bullet} \ycol_{0\bullet1} \zcol_{\bullet11}$};
\end{tikzpicture}
};
\node at (-1.5,0){(b)};
\node at (-1.5,-2){(c)};
\node at (-1.5,-4){(d)};
\node at (-1.5,-7){(e)};
\draw (-2,-1)--++(7,0);
\draw (-2,-3)--++(7,0);
\draw (-2,-5)--++(7,0);
\end{tikzpicture}
\caption{
(a) Naming scheme for the vertices, edges, and faces of a standard cube.
(b)-(e) Formula for the cup product $\xcol\cup \ycol\cup \zcol$ of three chains $\xcol$, $\ycol$, and $\zcol$ of various degrees.
The value of the cup product on a square or cube is the sum over products of $\xcol$, $\ycol$, and $\zcol$ on different triples of subcells of the square/cube.
Concretely, the triples of cells are ``sequential collections'' that connect the $00\ldots$ vertex to the $11\ldots$ vertex in the correct order.
For (b) $\xcol\in C^1, \ycol\in C^0, \zcol\in C^1$, the cup product determines the $CZ$ terms in the vertex stabilizer $A^\ttg$.
For (c) $\xcol\in C^2, \ycol\in C^1, \zcol\in C^0$ and (d) $\xcol\in C^1, \ycol\in C^0, \zcol\in C^2$, the cup products represent two of the six terms of the twisted errors $\widetilde M^b$ in Eq.~\eqref{eq:twisted_error_formula}.
For (e) $\xcol\in C^1, \ycol\in C^1, \zcol\in C^1$, the cup product corresponds to the action of the TQD path integral in Eq.~\eqref{eq:tqd_action}.
}
\label{fig:cup_product}
\end{figure}
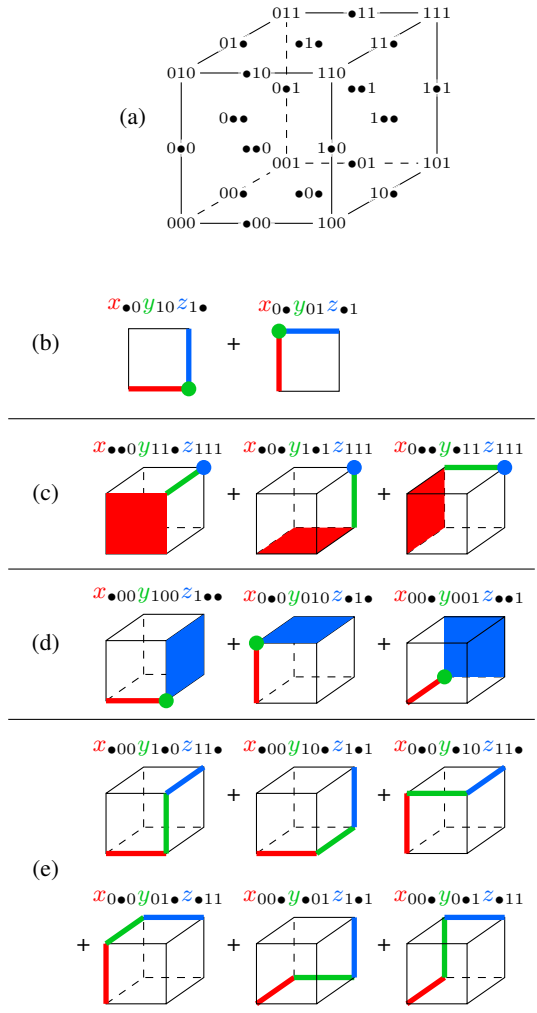

\myparagraph{The TQD path integral}
After introducing the cup product, we are now ready to define the path integral corresponding to our error-correcting circuit.
This is the Dijkgraaf-Witten path integral for the type-III 3-cocycle of $\zz_2^3$~\cite{Dijkgraaf1990, Propitius1995}, which we simply refer to as \emph{TQD path integral}.
The path integral can be defined on an arbitrary 3D cellulation, though for our purposes this lattice will be a regular cubic spacetime lattice.
The path integral is a discrete partition sum of the form
\begin{align}
\begin{split}
    Z =& \sum_{a_\ttr, a_\ttg, a_\ttb: \text{ $1$-cocycles}} (-1)^{\langle\epsilon, a_\ttr\cup a_\ttg\cup a_\ttb\rangle} \\
=& \sum_{a\in \zz_2^{\col\times E}: da=0} e^{2\pi i S[a]},
\end{split}
\end{align}
where
\begin{equation}
\label{eq:tqd_action}
S[a]=\frac12 \langle\epsilon, a_\ttr\cup a_\ttg\cup a_\ttb\rangle
\end{equation}
is called the \emph{action} of the path integral.
To deal with errors and non-deterministic measurement outcomes, we need to extend the path integral with \emph{flux} and \emph{charge defects}.
The configuration of 
flux defects 
is given by 
a triple of 2-chains $b\coloneqq (b_\ttr, b_\ttg, b_\ttb)\in \zz_2^{\col\times F}$, 
and the configuration of charge defects by 
a triple of 1-chains $c\coloneqq (c_\ttr,c_\ttg, c_\ttb)\in \zz_2^{\col\times E}$.
The extended 
path integral 
is
\begin{equation}
\label{eq:path_integral_defects}
Z[b,c] = \sum_{a\in \zz_2^{\col\times E}\colon da=b} e^{2\pi i\big(S[a]+ \frac12\langle c, a\rangle\big)}\;.
\end{equation}
Evaluating the path integral on a closed 3D manifold with periodic boundary condition yields a complex number $Z$.
We can also evaluate the path integral on a 3D lattice with \emph{state boundaries}:
When doing so, we keep the configuration $a$ on the boundary fixed, and only sum over different configurations of $a$ on the interior edges.
If $a_\partial$ denotes the configuration of $a$ on the boundary edges, then we denote this evaluation by $Z[b,c]_{a_\partial}$. $Z[b,c]_{a_\partial}$ can be interpreted as the amplitude of a (unnormalized) state vector, with three qubits at every boundary edge.
For a bi-partite boundary this state can be interpreted as the vectorized Krauss operator of an associated quantum channel, the TQD circuit.

\myparagraph{Error-correcting circuit revisited}
After introducing the TQD path integral, let us consider its relation with the TQD code and TQD circuit from Section~\ref{sec:TQDcircuits}.
Using cup products, we can concretely spell out the action of the stabilizers $A$ on a qubit configuration $a$:
\begin{equation}
\begin{gathered}
A_{\ttr v}\ket a = (-1)^{\langle\epsilon, v\cup a_\ttg \cup a_\ttb\rangle}\ket{a+d(\ttr,v)}\;,\\
A_{\ttg v}\ket a = (-1)^{\langle\epsilon,a_\ttr\cup v\cup a_\ttb\rangle}\ket{a+d(\ttg,v)}\;,\\
A_{\ttb v}\ket a = (-1)^{\langle\epsilon,a_\ttr\cup a_\ttg\cup v\rangle}\ket{a+d(\ttb,v)}\;.
\end{gathered}
\end{equation}
For example, the cup product occurring in $A_{\ttg v}$ is depicted in Figure~\ref{fig:cup_product} (b).
Here, we are slightly abusing notation and also use $(\lambda,v)$ to denote the element of $\zz_2^{\col\times V}$ with a single $1$ entry at $(\lambda,v)$.
Projectors corresponding to the two measurement outcomes $s=\pm 1$ of the two stabilizers are given by
\begin{align}
    \Gamma_{\lambda f}^s\coloneqq\frac12(1+ (-1)^s B_{\lambda f})
\end{align}
and
\begin{align}
    \Pi_{\lambda v}^s\coloneqq\frac12(1+ (-1)^s A_{\lambda v}).
\end{align}
We consider the TQD circuit consisting of $T$ rounds stabilizer measurements on an $L\times L$ square lattice with periodic boundary conditions.
To slightly simplify the derivation, we assume that Pauli-$X_\lambda$ errors happen in the circuit right after the $B_\lambda$ measurement in the periodic schedule $B_\ttb\rightarrow A_\ttb \rightarrow B_\ttg\rightarrow A_\ttg\rightarrow B_\ttr\rightarrow A_\ttr$, as opposed to right before like in our simulations.
Both are equivalent as the $B$ stabilizers are purely Pauli, and the Pauli-$X$ errors can be (anti-)commuted through, flipping the affected $B$ outcomes in post-processing.
As described in Section~\ref{sec:TQDcircuits}, the configurations of measurement outcomes and errors $c_m$ and $c_e$ can be identified with (triples of) 1-chains on an  $L\times L\times T$ cubic lattice whose first two dimensions correspond to the spatial $L\times L$ lattice and whose last dimension corresponds to time.
Similarly, $b_m$, $b_e$, and $b_c$ can be identified with 2-chains on that lattice.
To be precise, the $A$ measurement outcome at vertex $v\in V$ in time step $0\leq t<T$ corresponds to the cubic-lattice edge with spatial coordinate $v$ and time coordinate $[t,t+1]$, 
which we denote by ``$t^+v$''.
The $B$ measurement at face $f\in F$ in time step $0\leq t<T$ corresponds to the cubic-lattice face with spatial coordinate $f$ and time coordinate $t$ which we denote by ``$tf$''.
A Pauli-$X$ error at edge $e\in E$ in time step $0\leq t<T$ corresponds to the cubic-lattice face with spatial coordinate $e$ and time coordinate $[t,t+1]$ which we denote by ``$t^+e$''.
A Pauli-$Z$ error at edge $e\in E$ in time step $0\leq t<T$ corresponds to the cubic-lattice edge with spatial coordinate $e$ and time coordinate $t$ which we denote by ``$te$''.

With these correspondences, the linear operator implemented by the TQD circuit with  a fixed configuration of $\{b_m,b_e, b_c,c_m,c_e\}$ is given by
\begin{align}
\label{eq:circuit_explicit}
\begin{split}
C_{b,c} =& C_{b_m,b_e,b_c,c_m,c_e}\\
=& \prod_{0\leq t<T} \widetilde{\prod_{\lambda\in\col}}\Big(
\prod_{v\in V} \Pi_{\lambda v}^{c_{m\lambda t^+v}+c_{e \lambda t^+v}}\\
&\prod_{e\in E} X_{\lambda e}^{b_{e\lambda t^+e}+b_{c\lambda t^+e}}Z_{\lambda e}^{c_{e\lambda te}}
\prod_{ f\in F}  \Gamma_{\lambda f}^{b_{m\lambda tf}+b_{e\lambda tf}}
\Big)\\
=& \prod_{0\leq t<T} \widetilde{\prod_{\lambda\in \col}}\Big(
\prod_{v\in V} \Pi_{\lambda v}^{c_{\lambda t^+v}}\prod_{e\in E} X_{\lambda e}^{b_{\lambda t^+e}}Z_{\lambda e}^{c_{\lambda te}}
\prod_{f\in F} \Gamma_{\lambda f}^{b_{\lambda tf}}
\Big)\;.
\end{split}
\end{align}
Here, $\widetilde\prod$ indicates that the three factors of the product do not commute and we consider the fixed order $(\ldots)_{\lambda=\ttr}(\ldots)_{\lambda=\ttg}(\ldots)_{\lambda=\ttb}$.
We find that $C_{b_m,b_e,b_c,c_m,c_e}$ only depends on the sums of chains (or triples of chains), $b\coloneqq b_e+b_m+b_c$ and $c\coloneqq c_e+c_m$.

\myparagraph{Circuit = path integral} 
Consider the path integral on the $L\times L\times T$ cubic lattice with periodic boundary conditions in both length-$L$ dimensions and state boundaries in the length-$T$ direction.
The boundary configuration is given by $a_\partial=(a_0, a_T)$, where $a_0$ ($a_T$) is the restriction of $a$ to the $L\times L$ square lattice at time $0$ (at time $T$).
Then the circuit and the path integral are equal
\begin{equation}
\label{eq:circuit_pathintegral}
\bra{a_T} C_{b,c}\ket{a_0}
=Z[b,c]_{a_0,a_T}
\;.
\end{equation}
To derive this, let us first see how the operators in Eq.~\eqref{eq:circuit_explicit} act on qubit configurations $g$,
\begin{align}
\label{eq:circuit_pathintegral_operators}
\begin{split}
\Gamma_{\lambda f}^s \ket{g}
=& \delta_{(dg_\lambda)_f=s} \ket g\;,\\
\Pi_{\ttr v}^s \ket{g}
=& \frac12 \sum_x (-1)^{xs+\langle\epsilon, xv\cup g_\ttg\cup g_\ttb\rangle}  \ket{g+d(\ttr,xv)}\;,\\
\Pi_{\ttg v}^s \ket{g}
=& \frac12 \sum_x (-1)^{xs+\langle\epsilon,g_\ttr\cup xv\cup g_\ttb\rangle}  \ket{g+d(\ttg,xv)}\;,\\
\Pi_{\ttb v}^s \ket{g}
=& \frac12 \sum_x (-1)^{xs+\langle\epsilon,g_\ttr\cup g_\ttg\cup xv\rangle}  \ket{g+d(\ttb,xv)}\;,\\
Z_{e\lambda}^s\ket g =& (-1)^{sg_{\lambda e}}\ket g\,,\\
X_{e\lambda}^s\ket g =& \ket{g+s(\lambda,e)}\,,\\
\end{split}
\end{align}
Here, we use $v\in V$ to also denote the 0-chain with a single $1$ entry at $v$, and we use $xv$ to denote the product of the 0-chain $v$ with the scalar $x\in\zz_2$.
In particular, the product of all operators within one time slice acts as
\begin{align}
\begin{split}
\prod_{f\in F}\Gamma_{\lambda f}^{b_{\lambda tf}}\ket{g}
=& \delta_{dg_\lambda=b_{\lambda t}} \ket g\;,\\
\prod_{v\in V}\Pi_{\ttr v}^{c_{\lambda t^+v}} \ket g
=& \sum_{x_\ttr\in C^0} (-1)^{\langle c_{\ttr t^+},x_\ttr\rangle + \langle\epsilon, x_\ttr\cup g_\ttg\cup g_\ttb\rangle}\ket{g+(\ttr,dx_\ttr)}\;,\\
\prod_{e\in E} Z_{\lambda e}^{c_{\lambda te}}\ket g
=& (-1)^{\langle c_{\lambda t},g_\lambda\rangle}\ket g\;,\\
\prod_{e\in E} X_{\lambda e}^{b_{\lambda t^+e}}
=& \ket{g+b_{\lambda t^+}}\;,
\end{split}
\end{align}
and similar for $\Pi_\ttg$ and $\Pi_\ttb$.
With this, we can rewrite the circuit as
\begin{widetext}
\begin{align}
\label{eq:circuit_path_integral_derivation}
\begin{split}
\bra{a_T} C_{c, b}\ket{a_0}
\overset{\eqref{eq:circuit_explicit}, \eqref{eq:id_resolution}}{=} \bra{a_T}\prod_{0\leq t<T}\Big( \prod_{g\in \zz_2^{\col\times E_\twod}} \widetilde{\prod_{\lambda\in\col}} \Big( \prod_{v\in V}\Pi_{\lambda v}^{c_{\lambda t^+v}}\prod_{e\in E} X_{\lambda e}^{b_{\lambda t^+e}}Z_{\lambda e}^{c_{\lambda te}} \prod_{f\in F}\Gamma_{\lambda f}^{b_{\lambda tf}}\Big)\ket{g}\bra{g}\Big)\ket{a_0}\\
= \prod_{g\in \zz_2^{\col\times[0,T]\times E_\twod}}\braket{a_T}{g_T}\prod_{0\leq t<T}\bra{g_{t+1}} \widetilde{\prod_{\lambda\in\col}} \Big( \prod_{v\in V}\Pi_{\lambda v}^{c_{\lambda t^+v}}\prod_{e\in E} X_{\lambda e}^{b_{\lambda t^+e}}Z_{\lambda e}^{c_{\lambda te}} \prod_{f\in F}\Gamma_{\lambda f}^{b_{\lambda tf}}\Big)\ket{g_t}\braket{g_0}{a_0}\\
\overset{\eqref{eq:circuit_pathintegral_operators}}{=}
 \prod_{\substack{g\in \zz_2^{\col\times[0,T]\times E_\twod},\\x\in \zz_2^{\col\times[0,T]\times V_\twod}}} \delta_{g_0=a_0}\delta_{g_T=a_T}\prod_{0\leq t<T}  (-1)^{\langle c_t,g_t\rangle + \langle c_{t^+}, x_t\rangle} \delta_{dg_t=b_t} \delta_{g_{t+1}=g_t+b_{t^+}+dx_t}\\
 \times (-1)^{\langle\epsilon,x_{\ttr t}\cup g_{\ttg (t+1)}\cup g_{\ttb (t+1)} + g_{\ttr t}\cup x_{\ttg t}\cup g_{\ttb (t+1)} + g_{\ttr t}\cup g_{\ttg t}\cup x_{\ttb t}\rangle} \\
\overset{*, \eqref{eq:cup_product_circuit_mapping}}{=} \delta_{a_0=a_0}\delta_{a_T=a_T} \prod_{a_\lambda\in \zz_2^{E_\threed}}  (-1)^{\langle c,a\rangle + \langle\epsilon,a_\ttr\cup a_\ttg\cup a_\ttb\rangle} \delta_{da=b} = Z[b,c]_{a_0,a_T}\;.
\end{split}
\end{align}
\end{widetext}
For the first equality, we have inserted a resolution of the identity
\begin{equation}
\label{eq:id_resolution}
\one = \sum_{g} \ket{g}\bra{g}\;.
\end{equation}
For the equation labeled $*$, we have mapped pairs of $(g,x)$ to 1-chains $a$ on the 3D cubic lattice using the identification $a_{te}=g_{te}$ and $d_{t^+v}=x_{tv}$.
In other words, $a$ is given by $g$ on the time-perpendicular edges and by $x$ on the time-parallel edges.
It can be seen from Figure~\ref{fig:cup_product} that under this mapping, we have
\begin{equation}
\label{eq:cup_product_circuit_mapping}
\begin{gathered}
\begin{multlined}
\big(x_{\ttr t}\cup g_{\ttg (t+1)}\cup g_{\ttb (t+1)} + g_{\ttr t}\cup x_{\ttg t}\cup g_{\ttb (t+1)}\\ + g_{\ttr t}\cup g_{\ttg t}\cup x_{\ttb t}\big)_f
\end{multlined}\\
=\big(a_\ttr\cup a_\ttg\cup a_\ttb)_{t^+f}\;,
\end{gathered}
\end{equation}
where ``$t^+f$'' denotes the cubic-lattice volume whose spatial coordinate is the face $f$ and whose time coordinate $[t,t+1]$.
For example, the middle summand of the left-hand side corresponds to Figure~\ref{fig:cup_product} (b), and the right-hand side corresponds to Figure~\ref{fig:cup_product} (e).
Under the mapping $(g,x)\rightarrow a$ the two summands of Figure~\ref{fig:cup_product} (b) become equal to the last and first summand of Figure~\ref{fig:cup_product} (e).
Note that we can use the path-integral mapping in reverse to construct syndrome-extraction circuits for the TQD phase.
This was used, for example, in Refs.~\cite{Bauer2025allAbelian, davydova2025universal} to construct an interleaved circuit where one round of syndrome extraction consists of only 6 rounds of non-overlapping $CX$ and $CCZ$ gates.

\myparagraph{Deriving constraints and equivalences}
After equating the TQD syndrome-extraction circuit with the TQD path integral, we are ready to derive the constraints and equivalences in Eqs.~\eqref{eq:tqdb_conseq} and \eqref{eq:tqdc_conseq}.
Eq.~\eqref{eq:tqdb_constraint} is easy to see:
If $db\neq 0$ then the sum over $a$ in Eq.~\eqref{eq:path_integral_defects} is the empty sum, as $db=d(da)=0$ must hold for any $a$.
Eq.~\eqref{eq:tqdb_equivalence} follows directly from the topological order property of the path integral:
Namely, the boundary state on a region of the 3D cubic spacetime lattice of topology $S_2\times B_1$ (that is, a 3-ball with a smaller 3-ball removed from its interior) is a product state with respect to the bipartition into the inside boundary sphere and the outside boundary sphere.
For more details on the topological order property in path integrals, we refer the reader to Refs.~\cite{thesis, Bauer2025allAbelian, bauer2025planar}.

Next, we derive the constraints and equivalences in Eq.~\ref{eq:tqdc_conseq}.
We start by computing the \emph{gauge variance} of the action $S$:
\begin{equation}
\label{eq:gauge_variance}
\begin{gathered}
\Delta S[a, v] \coloneqq S[a+dv]-S[a]\\
\begin{multlined}
= \frac12 \Big(\langle\epsilon,(a^\ttr+dv^\ttr)\cup (a^\ttg+dv^\ttg) \cup (a^\ttb+dv^\ttb)\rangle\\
- \langle\epsilon,a^\ttr\cup a^\ttg\cup a^\ttb\rangle\Big)
\end{multlined}\\
=\frac12\Big\langle\epsilon,a^\ttr\cup b^\ttg\cup v^\ttb + b^\ttr\cup a^\ttg\cup v^\ttb + a^\ttr\cup v^\ttg\cup b^\ttb\\
+ b^\ttr\cup v^\ttg\cup a^\ttb + v^\ttr\cup a^\ttg\cup b^\ttb + v^\ttr\cup b^\ttg\cup a^\ttb\\
+ b^\ttr\cup v^\ttg\cup dv^\ttb + dv^\ttr\cup b^\ttg\cup v^\ttb + dv^\ttr\cup v^\ttg\cup b^\ttb\Big\rangle.\\
\end{gathered}
\end{equation}
Here, we have used the Leibnitz rule in Eq.~\eqref{eq:leibniz_rule}, the fact that $da=b$, as well as Eq.~\eqref{eq:fundamental_class_property}.
The gauge variance is zero if $da=b=0$, in other words, the action is \emph{gauge invariant} in the absence of flux ($b$) defects.
Next, we consider the variance of the gauge variance $\Delta S$ under $a\rightarrow a+dw$, or in other words, the \emph{second-order gauge variance}
\begin{equation}
\label{eq:gauge_double_variance}
\begin{gathered}
\Delta S[a+dw, v]-\Delta S[a,v]\\
=\frac12\Big\langle\epsilon,dw^\ttr\cup b^\ttg\cup v^\ttb + b^\ttr\cup dw^\ttg\cup v^\ttb + dw^\ttr\cup v^\ttg\cup b^\ttb \\
+ b^\ttr\cup v^\ttg\cup dw^\ttb + v^\ttr\cup dw^\ttg\cup b^\ttb + v^\ttr\cup b^\ttg\cup dw^\ttb\Big\rangle\\
\eqqcolon \frac12 \langle w,M^b v\rangle\;.
\end{gathered}
\end{equation}
Here, we have used that $\Delta S[a,v]$ in Eq.~\eqref{eq:gauge_variance} is affine linear in $a$, so its variance is obtained simply by replacing $a$ by $dw$.
We see that the result is a symmetric bilinear form $M^b$ in $v$ and $w$.

Now, the matrix $K^b$ and vector $\kappa^b$ in Eq.~\eqref{eq:tqdc_constraint} are obtained from $M^b$ as follows:
$(K^b)^T$ is any injective matrix whose image is equal to the kernel of $M^b$, in particular $M^b(K^b)^T=0$.
$\kappa^b$ is given by
\begin{equation}
\label{eq:kappa_definition}
\frac12 \kappa^b_x = \Delta S[a,(K^b)^Tx]\;,
\end{equation}
for some arbitrary choice of 1-chain $a$ such that $da=b$.
Note that the choice of $a$ does not matter if $b$ is 
contains no non-contractible loops since
\begin{equation}
\begin{gathered}
\Delta S[a+dx,(K^b)^Tx]\\
=\Delta S[a,(K^b)^Tx]+\frac12 \langle v,M^b (K^b)^T x\rangle=\Delta S[a,(K^b)^Tx]\;.
\end{gathered}
\end{equation}
Second, if $b$ contains no non-contractible loop then there are representatives of all 2-cohomology classes for $a$ that are are supported away from $b$, whereas $(K^b)^Tx$ is always supported in the vicinity of $b$.
Adding these representatives to $a$ does not change $\Delta S[a,(K^b)^Tx]$.
Note that if $b$ does contain non-contractible loops, then $\kappa^b$ is indeed ill-defined and $Z$ decoding is impossible.
To avoid this we need to decode the $X$ errors just-in-time.

Having defined $K^b$ and $\kappa^b$, let us now show that Eq.~\eqref{eq:tqdc_constraint} indeed holds,
\begin{align}
\begin{split}
Z[b,c]
=& \sum_{a: da=b} e^{2\pi i \big(S[a]+\frac12 \langle c, a\rangle\big)}\\
=& \sum_{a: da=b} \frac12 \big( e^{2\pi i \big(S[a]+\frac12 \langle c, a\rangle\big)}\\
&+ e^{2\pi i \big(S[a+d(K^b)^Tx]+\frac12 \langle c, a+d(K^b)^Tx\rangle\big)}\big)\\
=& \sum_a \frac12 e^{2\pi i \big(S[a]+\frac12 \langle c, a\rangle\big)} \big(1\\
&+ e^{2\pi i \big(\Delta S[a,(K^b)^Tx]+\frac12 \langle c, d(K^b)^T x\rangle\big)}\big)\\
\overset{\eqref{eq:kappa_definition}}{=}&\sum_a \frac12 e^{2\pi i \big(S[a]+\frac12 \langle c, a\rangle\big)} \big(1 + (-1)^{\kappa^b_x+ \langle c, d(K^b)^T x\rangle}\big)\\
=& Z[b,c]\cdot \delta_{\langle K^bd^Tc,x\rangle=\kappa^b_x}\;.
\end{split}
\end{align}
Using the equivalence between the circuit 
and path integral in Eq.~\eqref{eq:circuit_pathintegral}, Eq.~\eqref{eq:tqdc_constraint} follows.
The matrix $\widetilde M^b$ in Eq.~\eqref{eq:tqdc_equivalence} is any matrix 
that satisfies
\begin{equation}
d^T \widetilde M^b=M^b\;.
\end{equation}
More precisely, for any color-vertex pair $(\lambda,v)$, $M^b(\lambda,v)$ is supported only in the vicinity of $v$, and we can choose $\widetilde M^b(\lambda,v)$ to also be supported in this vicinity.
Concretely, it is very easy to read a choice for $\widetilde M^b$ off Eq.~\eqref{eq:gauge_double_variance} by simply replacing $dw$ by $x$,
\begin{align}
\label{eq:twisted_error_formula}
\begin{split}
\langle x, &\widetilde M^b v\rangle\\
=& x^\ttr\cup b^\ttg\cup v^\ttb + b^\ttr\cup x^\ttg\cup v^\ttb + x^\ttr\cup v^\ttg\cup b^\ttb \\
&+ b^\ttr\cup v^\ttg\cup x^\ttb + v^\ttr\cup x^\ttg\cup b^\ttb + v^\ttr\cup b^\ttg\cup x^\ttb\;.
\end{split}
\end{align}
We are now ready to show that Eq.~\eqref{eq:tqdc_equivalence} indeed holds,
following the steps
\begin{widetext}
\begin{align}
\label{eq:charge_equivalence_derivation}
\begin{split}
Z[b,c]
=& \sum_{a\in C^1: da=b} e^{2\pi i\big(S[a]+\frac12 \langle c,a\rangle\big)}
\overset{*}{=} \sum_{v\in C^0} e^{2\pi i\big(S[a_0+dv]+\frac12 \langle c,a_0+dv\rangle\big)}\\
=& e^{2\pi i\big(S[a_0]+\frac12 \langle c,a_0\rangle\big)} \sum_{v\in C^0} e^{2\pi i\big(\Delta S[a_0, v]+\frac12 \langle c,dv\rangle\big)}\\
\overset{**}{=}& e^{2\pi i\big(S[a_0+dx]+\frac12 \langle c,a_0+dx\rangle\big)} \sum_{v\in C^0} e^{2\pi i\big(\Delta S[a_0+dx, v]+\frac12 \langle c,dv\rangle\big)}\\
\overset{\eqref{eq:gauge_double_variance}}{=}& e^{2\pi i\big(\Delta S[a_0,x]+\frac12 \langle c,dx\rangle\big)}e^{2\pi i\big(S[a_0]+\frac12 \langle c,a_0\rangle\big)} \sum_{v\in C^0} e^{2\pi i\big(\Delta S[a_0, v]+ \frac12\langle M^bx, v\rangle +\frac12 \langle c,dv\rangle\big)}\\
=& e^{2\pi i\big(\Delta S[a_0,x]+\frac12 \langle c,dx\rangle +\frac12 \langle \widetilde M^bx,a_0\rangle\big)}e^{2\pi i\big(S[a_0]+\frac12 \langle c+\widetilde M^bx,a_0\rangle\big)} \sum_{v\in C^0} e^{2\pi i\big(\Delta S[a_0, v] +\frac12 \langle c+\widetilde M^bx,dv\rangle\big)}\\
=& e^{2\pi i\big(\Delta S[a_0,x]+\frac12 \langle c,dx\rangle +\frac12 \langle \widetilde M^bx,a_0\rangle\big)}
Z[b,c+\widetilde M^bx]
\propto Z[b,c+\widetilde M^bx]\;.
\end{split}
\end{align}
\end{widetext}
Note that the equation is required to hold on a patch of lattice with state boundary that is topologically a 3-ball $B_3$.
Since $B_3$ has trivial 1-cohomology, every $a$ with $da=b$ can be written as $a_0+dv$ for some fixed $a_0$ with $da_0=b$, which was used in the equation labelled $*$.
In the equation labelled $**$, we have performed a change of variables $a_0\rightarrow a_0+dx$

Finally, we should briefly point out how the formula for $\widetilde M^b$ shown in Eq.~\eqref{eq:twisted_error_formula} indeed coincides with what is shown in Figure~\ref{fig:zeff_illustration} (a).
The first, second, and third part of Figure~\ref{fig:zeff_illustration} (a) correspond to the last two, first two, and middle two terms of Eq.~\eqref{eq:twisted_error_formula}, respectively.
The cup product formula for the second term is shown in Figure~\ref{fig:cup_product} (c), and the cup product formula for the third term is shown in Figure~\ref{fig:cup_product} (d).
Figure~\ref{fig:zeff_illustration} (a) is obtained by taking the three terms in Figure~\ref{fig:cup_product} (c) or (d) and 
shifting them to different cubes such that the $v$ vertices coincide.

\section{Decoding with partial syndrome information}\label{app:generalJIT}
In this section we define a general JIT decoder for a spacetime decoding problem defined by a weighted hypergraph $G$, provided we have access to a (global) min-weight decoder $\sD$ for $G$.
We consider a decoding problem on a (hyper-)graph $G=(V,E)$ that is expressed as a linear problem over an Abelian group $A$.
Syndromes take values in $A^V$ and errors are associated with $A^E$.
We use the additive convention for the group operation in $A$ and denote the trivial group element with $0$.
For simplicity we omit the prefix ``hyper'' for the hyper-graph $G$ and its hyper-edges $E$.
For Pauli errors and measurement errors on binary measurements $A=\bZ_2$. Other groups appear naturally in the decoding of qudit and continuous-variable codes~\cite{Anwar2014quditcodes, Vuillot2019toricGKP}.

Let $\partial\colon E\to V$ be the boundary map that defines the adjacency between $V$ and $E$.
This induces a linear map $\partial_A\colon A^{E}\to A^{V}$, where $A^X = \spn_A(X)$ is the $A$-ring of dimension $\abs{X}$.
We also allow for additional weights $w\in \bR^E$.
\begin{defn}
    The min-weight decoding problem over $A$ defined by $G$ can be expressed as the problem of, given a \emph{syndrome $s\in A^V$} 
    of an unknown error $e\in A^E$, find
    \begin{align}\label{app:eq:minweight-correction}
        e_{\mathrm{cor}} = \argmin_{e\in A^E}(\norm{e}_w\;|\; \partial_A e = s),
    \end{align}
    where $\norm{e}_w=\sum_{i\in E} w_i \norm{e_i}_A$ with a norm $\norm{\bullet}_A$ on $A$.
\end{defn}

Let $\sD$ be a global decoder for $G$. That is, given a syndrome $s$ and weights $w$ it returns a correction $\sD(s,w)\in A^E$ that fits the syndrome. $\sD$ should be thought of as a (approximate) minimum-weight correction, as defined in Eq.~\eqref{app:eq:minweight-correction}.

A \emph{just-in-time} (JIT) decoder 
is necessary in QEC circuits with causality constraints that demand that at certain times a correction must be applied based on (partial) syndrome information available up to this timestep.
This happens, for example, when physical non-Clifford gates are applied, either for syndrome extraction of a non-Pauli code or as unitary logic gate.
In both cases a JIT decoder allows to ``decode through'' the non-Clifford gates.
We also demand that $G$ is of the form that any syndrome can be mapped to a future boundary, as necessary in online decoding~\cite{psidecoding, riverlanedecoding}.
For matchable decoding problems one solution to formalize this possibility is by introducing an additional future-lying vertex, as done in Section~\ref{sec:JIT-decoder} for a $\bZ_2$ matchable graph. Alternatively, half-edges can be added to the vertices on the future-lying boundary.

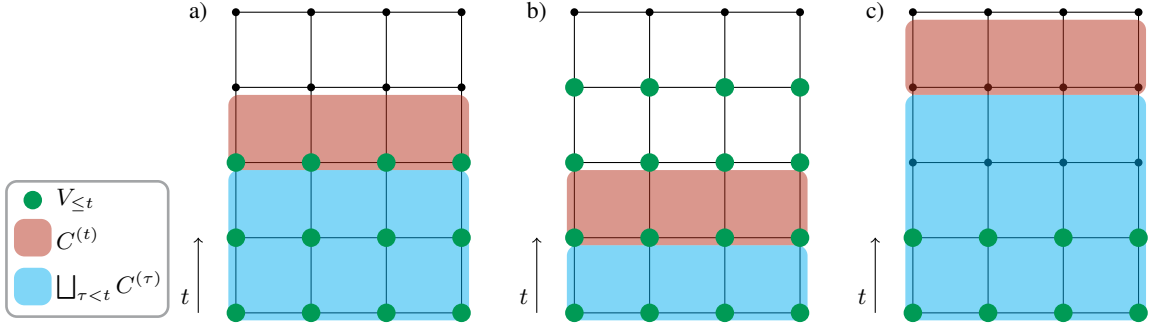
\begin{figure*}
    \centering
    \begin{tikzpicture}[
    vertex/.style={circle,fill=black,inner sep=0pt,minimum size=3pt},
    edge/.style={thin},
    circ/.style={circle,fill=ForestGreen,inner sep=0pt,minimum size=7pt},
    past/.style={opacity=0.5, Cyan, rounded corners},
    commit/.style={opacity=0.5, BrickRed, rounded corners}
]
\newcommand{\drawgrid}[2]{%
\def\rows{#1} 
\def\cols{#2} 
\def\spacing{1} 
\foreach \i in {0,...,\numexpr\rows-1\relax}{
    \foreach \j in {0,...,\numexpr\cols-1\relax}{
        \coordinate (v\i\j) at (\j*\spacing, \i*\spacing);
    }
}
\foreach \i in {0,...,\numexpr\rows-1\relax}{
    \foreach \j in {0,...,\numexpr\cols-2\relax}{ 
        \pgfmathtruncatemacro{\next}{\j+1}
        \draw[edge] (v\i\j) -- (v\i\next);
    }
}
\foreach \i in {0,...,\numexpr\rows-2\relax}{ 
    \foreach \j in {0,...,\numexpr\cols-1\relax}{
        \pgfmathtruncatemacro{\next}{\i+1}
        \draw[edge] (v\i\j) -- (v\next\j);
    }
}
\foreach \i in {0,...,\numexpr\rows-1\relax}{
    \foreach \j in {0,...,\numexpr\cols-1\relax}{
        \node[vertex] at (v\i\j) {};
    }
}
}

\begin{scope}[shift={(-7,-1.25)}]
    \draw[Gray, rectangle, fill=white, opacity=0.85, rounded corners, line width=1pt] (-0.05,1.2) rectangle (2.1,3);
    \node[circ] (A) at (0.3,2.75) {};
    \node[anchor=west] at (0.5, 2.75) {$V_{\leq t}$};
    \fill[past] (0.05,1.35) rectangle (0.55,1.85);
    \node[anchor=west] at (0.5,1.6) {$\bigsqcup_{\tau < t}C^{(\tau)}$};
    \fill[commit] (0.05,2) rectangle (0.55,2.5);
    \node[anchor=west] at (0.5,2.25) {$C^{(t)}$};
\end{scope}

\begin{scope}[shift={(-4,0)}]
\node at (-0.5,4) {a)};
 \drawgrid{5}{4}
 \draw[->] (-0.5,0) -- +(0,1);
 \node[anchor=south east] at (-0.5,0) {$t$};
 \fill[past] (-0.1,-0.1) rectangle (3.1,1.9);
 \fill[commit] (-0.1,1.9) rectangle (3.1,2.9);
 \foreach\i in {0,1,...,3} {
    \foreach\j in {0,1,2}{
    \node[circ] at (\i,\j) {};
 }}
\end{scope}

\begin{scope}[shift={(0.5,0)}]
\node at (-0.5,4) {b)};
 \drawgrid{5}{4}
 \draw[->] (-0.5,0) -- +(0,1);
 \node[anchor=south east] at (-0.5,0) {$t$};
 \fill[past] (-0.1,-0.1) rectangle (3.1,0.9);
 \fill[commit] (-0.1,0.9) rectangle (3.1,1.9);
 \foreach\i in {0,1,...,3} {
    \foreach\j in {0,1,...,3} {
    \node[circ] at (\i,\j) {};
 }}
\end{scope}

\begin{scope}[shift={(5,0)}]
\node at (-0.5,4) {c)};
 \drawgrid{5}{4}
 \draw[->] (-0.5,0) -- +(0,1);
 \node[anchor=south east] at (-0.5,0) {$t$};
 \fill[past] (-0.1,-0.1) rectangle (3.1,2.9);
 \fill[commit] (-0.1,2.9) rectangle (3.1,3.9);
 \foreach\i in {0,1,...,3} {
    \foreach\j in {0,1} {
    \node[circ] at (\i,\j) {};
 }}
\end{scope}

\end{tikzpicture}
    \caption{Sketch of different scenarios for JIT decoders, based on the temporal structure of the \emph{accessible syndromes at $t$} $V_{\leq t}$ (green vertices), \emph{commit region at $t$} $C^{(t)}$ (edges in red region) and  \emph{previous commits} $\bigsqcup_{\tau < t} C^{(\tau)}$ (edges in blue region). The figure shows an illustration of a part of a decoding graph, with time running upwards. 
    a) In the scenario that we consider in the main text the syndrome acquired at $t$ live on vertices that connect to all edges in the commit region.
    It applies to circuits that require immediate corrections in the whole spacetime region that is connected to the newly-acquired syndromes.
    b) In a different scenario the JIT decoder has access to future-lying syndrome vertices that are not directly connected to the commit region.
    This scenario emerges in a setting in which the error-correcting circuit measures through a higher-dimensional resource state, keeping multiple consecutive time-layers alive at the same time. In this case, corrections only need to be applied to a part of system, before it is measured out.
    c) In the third scenario the JIT decoder only has access to syndromes in the past of the commit region. This arises when dealing with a non-negligible reaction time of the JIT decoder as it must then decide on a correction in the future.
    Both settings b) and c) have natural parameters that could be tuned: The amount of future-lying accessible syndrome information and the delay time both of which have a clear cost assigned to them.
    The accessible syndrome in the future can be increased by increasing the spatial and connectivity overhead in a protocol and the delay time can be reduced at the cost of increased classical computing resources.
    The JIT decoder is expected to have a threshold in all of the above scenarios but likely with varying numerical values.
    We expect the threshold to increase in b) and decrease in c) when compared to a) but leave a quantitative comparison to future work.}
    \label{fig:jit-versions}
\end{figure*}

We first formalize the data that captures the causal constraints that the JIT decoder is based on.
\begin{defn}[Commit region]\label{def:JITgeneral-setup}
    At each time $t\in \bZ$, we define the \emph{commit region}
    \begin{align}
        C^{(t)}\subseteq E
    \end{align}
    such that $E=\bigsqcup_t C^{(t)}$ and $C^{(t)}\cap C^{(t')} = \emptyset$ for $t\neq t'$.
    We also require that each $C^{(t)}$ is simply connected in $G$ and that $C^{(t)}$ is only shares vertices with $C^{(t-1)}$ and $C^{(t+1)}$.
    We denote the set of vertices connected to at least one edge in $C^{(t)}$ with $\partial C^{(t)}\subseteq V$.
    At each time, we define the set of \emph{syndromes obtained at $t$}, $V_t \subseteq V$, such that $V_t\cap V_{t'} = \emptyset$ for $t\neq t'$.
    The set of \emph{accessible syndrome at $t$} is defined as
    \begin{align}
        V_{\leq t} = \bigsqcup_{\tau\leq t}V_\tau.
    \end{align}
\end{defn}
Note that the time variable in Def.~\ref{def:JITgeneral-setup} labels the points in time where the JIT decoder must provide a correction and not necessarily label time steps in the quantum circuit.

We illustrate different scenarios for $C^{(t)}$ and $V_{\leq t}$ in Figure~\ref{fig:jit-versions}. Note that a valid JIT decoder can be defined even when $V_t\cap \partial C^{(t)} = \emptyset$, either because it lies in the past or the future of $\partial C^{(t)}$.
The JIT decoder presented in the main text is a simplified version in which $V_t$ is defined to be the vertices shared between $C^{(t-1)}$ and $C^{(t)}$.

The JIT decoder consists of two subroutines.
Both steps can be expressed by solving a decoding problem on $G$ with modified weights.
\begin{defn}
    At time $t\in \bZ$ we define two modified weight vectors.
    \begin{itemize}
        \item The \emph{estimating weights} $w^{\mathrm{g}}_t$ is defined to equal $w$ on edges that connect to  at least one vertex in $V_{\leq t}$ and $0$ on all other (future-lying) edges.
        \item The \emph{commit weightss} $w^{\mathrm{c}}_t$ is defined to equal $w$ on $C^{(t)}$ and $\infty$ otherwise.
    \end{itemize}
\end{defn}
The estimating weights define a decoding problem in which a syndrome can be matched into the future for free and the commit weights are chosen to enforce that the correction is supported on $C^{(t)}$.

We are now in position to define $\sDJIT$ 
as a sequence of decoders for all times $t$ based on $\sD$.
\begin{defn}
    The JIT decoder $\sDJIT$ consists of a sequence of sub decoders $\{\sDJIT(t)\}_t$.
    Each $\sDJIT(t)$ takes as an input the currently accessible syndrome $s_{\leq t} = s|_{V_{\leq t}}$ and outputs a correction $e_{\mathrm{JIT}}(t)$ supported on $C^{(t)}$ and the output of $\sDJIT(t-1)$.
    The output of $\sDJIT$ is a global correction
    \begin{align}
        e_{\mathrm{JIT}} = \sum_t e_{\mathrm{JIT}}(t)\in A^E.
    \end{align}
    Each $\sDJIT(t)$ consists of two subroutines, which we define in Algorithm~\ref{alg:DJITt}.
\end{defn}

\begin{algorithm}
\caption{$\sDJIT(t)$}
\label{alg:DJITt}
\DontPrintSemicolon
\KwIn{$s_{\leq t}$, $e_{\mathrm{JIT}}(t-1)$}
\KwOut{$e_{\mathrm{JIT}}(t)\in A^{C^{(t)}}$}
\hrulefill\;
Run $\sD(s_{\leq t},w^{\mathrm{g}}_t)$, get $e_{1}(t)$ \tcp*{Step 1}
Compute $\tilde{s}_{t} = \partial_A e_1(t)|_{\partial C^{(t)}}$\;
Compute $\tilde{s}_{t-1} = \partial_Ae_{\mathrm{JIT}}(t-1)|_{\partial C^{(t)}}$\;
Run $\sD(\tilde{s}_{t} + \tilde{s}_{t-1}, w^{\mathrm{c}}_t)$, get $e_{\mathrm{JIT}}(t)$ \tcp*{Step 2}
\Return{$e_{\mathrm{JIT}}(t)$}\;
\end{algorithm}

The JIT decoder is defined to always provide a valid correction.

\begin{thm}[Validity of corrections]
    $\sDJIT$ outputs a valid correction $\varepsilon_{\mathrm{JIT}}\in A^E$,
    \begin{align}
        \partial_A (e_{\mathrm{JIT} } + e) = 0,
    \end{align}
    where $e\in A^E$ is the unknown error that caused syndrome $s=\partial_A e$ which has been processed by $\sDJIT$.
\end{thm}

\begin{proof}
The committed corrections at time $t$ can only have non-trivial boundary in $\partial C^{(t)}$.
Hence, we only need to check that the sum of commits from $t-1$ and $t$ have the correct boundary 
on $B = \partial C^{(t-1)}\cap \partial C^{(t)}$,
\begin{align}
    \eval{\partial_A (e_{\mathrm{JIT}}(t-1) + e_{\mathrm{JIT}}(t))}_{B} \stackrel{!}{=} \eval{s}_{B}.
\end{align}
We consider each vertex $v\in B$ and possible syndrome value $s_v$ individually.
Consider the case when $s_v = 0$. In this case $\tilde{s}_{t}|_v = 0$ and $\tilde{s}_{t-1}|_v = (\partial_Ae_{\mathrm{JIT}}(t))_v$ by definition of the second step of $\sDJIT(t)$.
It follows that $(\partial_A e_{\mathrm{JIT}})_v = 0 = s_v$.

Consider the case where $s_v=1$.
In this case, $\tilde{s}_{t}|_v = 1$.
The second matching step now commits $e_{\mathrm{JIT}}(t)$ with that has syndrome on $v$ $\tilde{s}_{t}|_v = 0$ if $\tilde{s}_{t-1}|_v = 1$ and $\tilde{s}_{t}|_v = 1$ if $\tilde{s}_{t-1}|_v = 0$.
In both cases the boundary of the combined correction, on $v$, is $\tilde{s}_{t-1}|_v + \tilde{s}_{t}|_v = 1 = s_v$ by linearity of $\partial_A$.

This proves that for all possible values of $s|_B$ $\sDJIT$ provides a valid correction after time $t$. Over the full protocol this combines to valid global correction.
\end{proof}

\section{Twisted errors in arbitrary CSS-type circuits}
In this appendix, we show how to compute the twisted errors in an arbitrary CSS-type error-correcting circuit consisting of $CX$ gates and (multi-qubit) $X$ or $Z$ measurements that is enriched with diagonal gates in the third level of the Clifford hierarchy, that is, $T$, $CS$, and $CCZ$ gates~\cite{Cui2017diagonal}.
Another example where twisted errors have been computed for error-correcting circuits with $T$ gates can be found in Ref.~\cite{bauer2025planar} -- here we give a general prescription.

Each such circuit can be mapped onto a cohomological path integral, defined on a spacetime chain complex.
The space of 1-chains of that complex are associated with locations of physical $Z$ errors and $X$-measurement outcomes.
This corresponds to the spacetime history of qubit configurations over the circuit and is similar to $a=(g,x)$ in Eq.~\eqref{eq:circuit_path_integral_derivation}.
After that identification, the circuit can be mapped onto a path integral of the form
\begin{equation}\label{app:eq:general-pathintegral}
Z[b,c]=\sum_{a\in \zz_2^X: d_1a=b} e^{2\pi i \big(S[a] + \frac12\langle c,a\rangle\big)}\;.
\end{equation}
The function $S$ encodes the diagonal non-Pauli gates in the circuit.
The sum runs over all 1-chains $a$ subject to a linear constraint, encoded in $d_1: \zz_2^X\rightarrow \zz_2^Y$, where $Y$ coincides with the space of all $X$-like errors.
We assume that the circuit implements a \emph{spacetime LDPC protocol} where $d_1$ is \emph{local} on a degree-bounded (hyper)graph, which we call the \emph{locality graph}.
That is, $d_1$ is only non-zero on pairs of locations $x\in X$ and $y\in Y$ whose combinatorial distance on the locality graph is smaller than some global constant.
This could for example be a 3D spacetime lattice for a 2+1D topological protocol, a 5D lattice for repeated syndrome extraction rounds of a bivariate bicycle code~\cite{bbcodes}, or an infinite-dimensional expander graph for some good qLDPC code family.
The charge configuration $c\in \zz_2^X$ consists of variables assigned to the same places as $a$.
The flux configuration $b\in \zz_2^Y$ has variables assigned to different places on the same locality graph.
We give a more instructive explanation of that mapping in the remark at the end of this section.

The action $S[a]$ is a function $S: \zz_2^X\rightarrow \mathbb R/\zz$.
For a fault-tolerant error-correcting circuit of the type above, the action needs to have three properties.
The first property is locality on the underlying locality graph, i.e., $S$ is of the form
\begin{equation}
S[a] = \sum_{z\in Z} S_z[a]\;,
\end{equation}
where $Z$ is a set of places on the locality graph, and $S_z$ only depends on the $a$ variables within a constant radius from $z$.
The second property of the action is \emph{gauge invariance}.
Gauge invariance ensures that $S$ only depends on the cohomology class of the spacetime qubit configuration $a$, into which the logical information is encoded.
To define gauge invariance, we first find a generating set of local $a\in \zz_2^X$, supported inside a constant-size region on the locality graph, such that $d_1a=0$.
We can combine all of these local generating $a$ to form the columns of a matrix $d_0$, such that $d_0$ is local and $d_1d_0=0$.
Then gauge invariance is given by
\begin{equation}
b=0\quad\Rightarrow\quad \Delta S[a,v]\coloneqq S[a+d_0v]-S[a] = 0\;.
\end{equation}
The third property is that $S_z$ needs to be a a \emph{third order function}~\cite{quadratic_tensors}.
This is not strictly necessary for fault tolerance, but makes the derivation of the twisted errors particularly nice algebraically.
To define a third order function, we first define the \emph{derivative} of an $n$-argument function
\begin{equation}
f: \zz_2^{X_0}\times \zz_2^{X_1} \times \zz_2^{X_2}\times \ldots \rightarrow \mathbb R/\zz
\end{equation}
as the $n+1$-argument function
\begin{equation}
\begin{gathered}
\partial f(x_0,x_1,x_2,\ldots)\\
\coloneqq f(x_0+x_1,x_2,\ldots)-f(x_0,x_2,\ldots)-f(x_1,x_2,\ldots)\;.
\end{gathered}
\end{equation}
That is, the derivative tests the linearity of $f$ in the first argument.
We denote the $i$th derivative of a 1-argument function $S$ as $S^{(i)}$.
$S$ is an $i$th order function if $S^{(i+1)}=0$, or equivalently $S^{(i)}$ is linear in the first component -- note that $S^{(i)}$ is symmetric in its arguments and is thus linear in every component.
So $S$ is a third order function if $S^{(3)}$ is trilinear.
The trilinearity of $S$ is in direct correspondence with the fact that the non-Clifford diagonal gates in the circuit are in the third level of the Clifford hierarchy.
Relatedly, it is known~\cite{quadratic_tensors} that any third order function consists of ``$T$-like'', ``$CS$-like'', and ``$CCZ$-like'' terms, i.e., it is of the form
\begin{align}
\begin{split}
S[a] =& \sum_{x_0\in X} \frac{c_{x_0}}{8} a_x + \sum_{x_0,x_1\in X} \frac{c_{x_0x_1}}{4} a_{x_0} a_{x_1}\\
&+ \sum_{x_0,x_1,x_2\in X} \frac{c_{x_0x_1x_2}}{2} a_{x_0} a_{x_1} a_{x_2}
\end{split}
\end{align}
for some coefficients $c$.
With this, the matrix $\widetilde M^b$ defining the twisted errors is given by
\begin{equation}
\langle e,\widetilde M^b v\rangle
= S^{(3)}(a_0,e,d_0v)\;,
\end{equation}
where $a_0$ is any element of $\zz_2^X$ such that $d_1a_0=b$.
The right-hand side is bilinear in $e$ and $v$ because $S^{(3)}$ is trilinear.
If $e\in X$ and $v\in Y$ are generators, it suffices to choose $a_0$ inside a small neighborhood of $e$ and $v$, due to the locality of $S$ and consequently of $S^{(3)}$.
Even though the twisted errors $\langle e,\widetilde M^b v\rangle$ depend on the choice of $a_0$, what we really care about is their syndrome 
\begin{equation}    
M^b\coloneqq d_0^T\widetilde M^b, 
\end{equation}
which is independent of $a_0$:
Since we only care about $a_0$ in the vicinity of $v$ and $e$, any other $a_0'$ with $d_1a_0'=b$ is of the form $a_0'=a_0+d_0\alpha$, and we find
\begin{align}
\begin{split}
\langle w, M^b v\rangle
=& S^{(3)}(a_0,d_0w,d_0v)\\
=& S^{(3)}(a_0,d_0w,d_0v) + S^{(3)}(d_0\alpha,d_0w,d_0v)\\
=& S^{(3)}(a_0+d\alpha,d_0w,d_0v)\\
=& S^{(3)}(a_0',d_0w,d_0v)\;,
\end{split}
\end{align}
where 
\begin{equation}
S^{(3)}(d_0\bullet,d_0\bullet,d_0\bullet)=0
\end{equation}
follows straight-forwardly from gauge invariance as $S[d_0\bullet]=0$.
The derivation that $\widetilde M^b$ gives rise to equivalences for the charge configuration $c$ is analogous to Eq.~\eqref{eq:charge_equivalence_derivation}.
In this context, we notice that the gauge variance corresponds to the second derivative of $S$,
\begin{align}
\begin{split}
\Delta S[a,v]
=& S[a+d_0v]-S[a]\\
=& S[a+d_0v]-S[a]-S[d_0v]\\
=& S^{(2)}(a,d_0v)\;.
\end{split}
\end{align}
In the same way, the second-order gauge variance corresponds to the third derivative:
\begin{align}
\begin{split}
\Delta^2 S[a,v,w]
=& \Delta S[a+d_0w,v]-\Delta S[a,v]\\
=& \Delta S[a+d_0w,v]-\Delta S[a,v]-\Delta S[d_0w,v]\\
=& S^{(2)}(a+d_0w,d_0v)-S^{(2)}(a,d_0v)\\
&-S^{(2)}(d_0w,d_0v)\\
=& S^{(3)}(a,d_0w,d_0v)\;.
\end{split}
\end{align}

\myparagraph{Chain complex from the circuit}
The identification of the circuit with a path integral can also be understood on the level of the $ZX$-like tensor-network representation of the circuit, see Refs.~\cite{Bombin2023, path_integral_qec, xyzrubycode, xy_floquet, kissinger2022phasefreezx}.
In particular, we consider a bipartite diagram built from two types of tensors, \emph{$X$ tensors},
\begin{equation}
\begin{tikzpicture}
\atoms{circ,small}{0/lab={t=$s$,p=45:0.25}}
\draw (0)edge[ind=$a$]++(0:0.5) (0)edge[ind=$b$]++(90:0.5) (0)edge[ind=$c$]++(180:0.5);
\node at (-90:0.5){$\ldots$};
\end{tikzpicture}
=
\delta_{a+b+c+\ldots=s}\;.
\end{equation}
and \emph{$Z$ tensors},
\begin{equation}
\label{eq:delta_definition}
\begin{tikzpicture}
\atoms{circ, small, all}{0/lab={t=$s$,p=45:0.25}}
\draw (0)edge[ind=$a$]++(0:0.5) (0)edge[ind=$b$]++(90:0.5) (0)edge[ind=$c$]++(180:0.5);
\node at (-90:0.5){$\ldots$};
\end{tikzpicture}
=
(-1)^s \delta_{a=b=c=\ldots}\;,
\end{equation}
where $s\in\{0,1\}=\zz_2$ is a \emph{sign} that can be added to each tensor.
Note that any diagram composed of $X$ and $Z$ tensors can be brought into bipartite form by fusing all adjacent tensors of the same type or adding two-legged sign-free tensors between tensors of the same type.
The circuits we consider are represented by a tensor network with $X$ and $Z$ tensors on the qubit worldlines and additional local third-order tensors~\cite{quadratic_tensors}, one for each non-Clifford gate in the circuit.
The map $d_1:\zz_2^X \to \zz_2^Y$ is then given by the adjacency matrix between $X$ and $Z$ tensors.
The map $d_0$ corresponds to the matrix whose column vectors are the local generators of the kernel of $d_1$.
It can be obtained from local relations among the $X$ stabilizers of the $Z$ tensors in the network when all $X$ tensors have trivial sign, similar to $X$-detectors in Clifford diagrams~\cite{xyzrubycode, rodatz2024floquetifying}.
Pauli errors, measurement errors, and $-1$ measurement outcomes in the circuit can then be mapped onto configurations of signs $s$ on the $X$ and $Z$ tensors.
These can be identified with 2-chains and 1-chains respectively, in a chain complex defined by $d_0$ and $d_1$.

\end{appendix}

\bibliography{refs}

@misc{bauer2025planar,
  title                 = {{Planar fault-tolerant circuits for non-Clifford gates on the 2D color code}},
  author                = {Bauer, Andreas and Magdalena de la Fuente, Julio C. },
  year                  = {2025},
  eprint                = {2505.05175},
  archiveprefix         = {arXiv},
  optoptoptprimaryclass = {quant-ph},
  url                   = {https://arxiv.org/abs/2505.05175}
}

@misc{davydova2025universal,
  title                 = {{Universal fault tolerant quantum computation in 2D without getting tied in knots}},
  author                = {Davydova, Margarita and Bauer, Andreas and Magdalena de la Fuente, Julio C. and Webster, Mark and Williamson, Dominic J. and Brown, Benjamin J.},
  year                  = {2025},
  eprint                = {2503.15751},
  archiveprefix         = {arXiv},
  optoptoptprimaryclass = {quant-ph},
  url                   = {https://arxiv.org/abs/2503.15751}
}

@article{Brown2020cczsurfacecode,
  title     = {{A fault-tolerant non-Clifford gate for the surface code in two dimensions}},
  volume    = {6},
  pages     = {eaay4929},
  doi       = {10.1126/sciadv.aay4929},
  journal   = {Science Adv.},
  publisher = {American Association for the Advancement of Science (AAAS)},
  author    = {Brown, Benjamin J.},
  year      = {2020}
}

@misc{MindTheGaps,
  title         = {Mind the gaps: The fraught road to quantum advantage},
  author        = {J. Eisert and J. Preskill},
  year          = {2025},
  eprint        = {2510.19928},
  archiveprefix = {arXiv},
  url           = {https://arxiv.org/abs/2510.19928}
}

@misc{bombin20182d3d,
  title                 = {{2D quantum computation with 3D topological codes}},
  author                = {Hector Bombin},
  year                  = {2018},
  eprint                = {1810.09571},
  archiveprefix         = {arXiv},
  optoptoptprimaryclass = {quant-ph},
  url                   = {https://arxiv.org/abs/1810.09571}
}

@article{Dennis2002,
  title     = {Topological quantum memory},
  volume    = {43},
  issn      = {1089-7658},
  url       = {http://dx.doi.org/10.1063/1.1499754},
  doi       = {10.1063/1.1499754},
  number    = {9},
  journal   = {J. Math. Phys.},
  publisher = {AIP Publishing},
  author    = {Dennis, Eric and Kitaev, Alexei and Landahl, Andrew and Preskill, John},
  year      = {2002},
  month     = sep,
  pages     = {4452–4505}
}

@misc{kobayashi2025cliffordcodes,
  title                 = {{Clifford hierarchy stabilizer codes: transversal non-Clifford gates and magic}},
  author                = {Ryohei Kobayashi and Guanyu Zhu and Po-Shen Hsin},
  year                  = {2025},
  eprint                = {2511.02900},
  archiveprefix         = {arXiv},
  optoptoptprimaryclass = {quant-ph},
  url                   = {https://arxiv.org/abs/2511.02900}
}

@misc{christos2026twisted-ldpc,
  title                 = {{Non-Abelian quantum low-density parity check codes and non-Clifford operations from gauging logical gates via measurements}},
  author                = {Maine Christos and Chiu Fan Bowen Lo and Vedika Khemani and Rahul Sahay},
  year                  = {2026},
  eprint                = {2602.12228},
  archiveprefix         = {arXiv},
  optoptoptprimaryclass = {quant-ph},
  url                   = {https://arxiv.org/abs/2602.12228}
}

@misc{williamson2026higher-form-gauging,
  title                 = {Fast magic state preparation by gauging higher-form transversal gates in parallel},
  author                = {Dominic J. Williamson},
  year                  = {2026},
  eprint                = {2601.22939},
  archiveprefix         = {arXiv},
  optoptoptprimaryclass = {quant-ph},
  url                   = {https://arxiv.org/abs/2601.22939}
}

@misc{zhu2026twisted-ldpc,
  title                 = {{Non-Abelian qLDPC: TQFT formalism, addressable gauging measurement and application to magic state fountain on 2D product codes}},
  author                = {Guanyu Zhu and Ryohei Kobayashi and Po-Shen Hsin},
  year                  = {2026},
  eprint                = {2601.06736},
  archiveprefix         = {arXiv},
  optoptoptprimaryclass = {quant-ph},
  url                   = {https://arxiv.org/abs/2601.06736}
}

@article{Bauer2025allAbelian,
  title     = {{Low-overhead non-Clifford fault-tolerant circuits for all non-chiral Abelian topological phases}},
  volume    = {9},
  doi       = {10.22331/q-2025-03-25-1673},
  journal   = {Quantum},
  publisher = {Verein zur Forderung des Open Access Publizierens in den Quantenwissenschaften},
  author    = {Bauer,  Andreas},
  year      = {2025},
  month     = mar,
  pages     = {1673}
}

@misc{psidecoding,
  title                 = {Modular decoding: parallelizable real-time decoding for quantum computers},
  author                = {Hector Bombin and Chris Dawson and Ye-Hua Liu and Naomi Nickerson and Fernando Pastawski and Sam Roberts},
  year                  = {2023},
  eprint                = {2303.04846},
  archiveprefix         = {arXiv},
  optoptoptprimaryclass = {quant-ph},
  url                   = {https://arxiv.org/abs/2303.04846}
}

@article{riverlanedecoding,
  title     = {Parallel window decoding enables scalable fault tolerant quantum computation},
  volume    = {14},
  doi       = {10.1038/s41467-023-42482-1},
  pages     = 7040,
  journal   = {Nature Comm.},
  publisher = {Springer Science and Business Media LLC},
  author    = {Skoric, Luka and Browne, Dan E. and Barnes, Kenton M. and Gillespie, Neil I. and Campbell, Earl T.},
  year      = {2023},
  month     = nov
}

@article{Anwar2014quditcodes,
  doi       = {10.1088/1367-2630/16/6/063038},
  url       = {https://doi.org/10.1088/1367-2630/16/6/063038},
  year      = {2014},
  month     = {jun},
  publisher = {IOP Publishing},
  volume    = {16},
  number    = {6},
  pages     = {063038},
  author    = {Anwar, Hussain and Brown, Benjamin J and Campbell, Earl T and Browne, Dan E},
  title     = {Fast decoders for qudit topological codes},
  journal   = {New J. Phys.}
}

@misc{delfosse2023spacetime,
  title                 = {{Spacetime codes of Clifford circuits}},
  author                = {Nicolas Delfosse and Adam Paetznick},
  year                  = {2023},
  eprint                = {2304.05943},
  archiveprefix         = {arXiv},
  optoptoptprimaryclass = {quant-ph},
  url                   = {https://arxiv.org/abs/2304.05943}
}

@article{Scruby2022JITnumerics,
  title     = {Numerical Implementation of Just-In-Time Decoding in Novel Lattice Slices Through the Three-Dimensional Surface Code},
  volume    = {6},
  issn      = {2521-327X},
  url       = {http://dx.doi.org/10.22331/q-2022-05-24-721},
  doi       = {10.22331/q-2022-05-24-721},
  journal   = {Quantum},
  publisher = {Verein zur Forderung des Open Access Publizierens in den Quantenwissenschaften},
  author    = {Scruby, T. R. and Browne, D. E. and Webster, P. and Vasmer, M.},
  year      = {2022},
  month     = may,
  pages     = {721}
}

@misc{scruby2025nodistillation,
  title                 = {{Fault-tolerant quantum computation without distillation on a 2D device}},
  author                = {Thomas R. Scruby and Kae Nemoto and Zhenyu Cai},
  year                  = {2025},
  eprint                = {2412.12529},
  archiveprefix         = {arXiv},
  optoptoptprimaryclass = {quant-ph},
  url                   = {https://arxiv.org/abs/2412.12529}
}

@article{fowler2012surface,
  title     = {Surface codes: Towards practical large-scale quantum computation},
  author    = {Fowler, Austin G. and Mariantoni, Matteo and Martinis, John M. and Cleland, Andrew N.},
  journal   = {Phys. Rev. A},
  volume    = {86},
  issue     = {3},
  pages     = {032324},
  numpages  = {48},
  year      = {2012},
  month     = {Sep},
  publisher = {American Physical Society},
  doi       = {10.1103/PhysRevA.86.032324},
  url       = {https://link.aps.org/doi/10.1103/PhysRevA.86.032324}
}

@article{Litinski2019game,
  title     = {A Game of Surface Codes: Large-Scale Quantum Computing with Lattice Surgery},
  volume    = {3},
  issn      = {2521-327X},
  url       = {http://dx.doi.org/10.22331/q-2019-03-05-128},
  doi       = {10.22331/q-2019-03-05-128},
  journal   = {Quantum},
  publisher = {Verein zur Forderung des Open Access Publizierens in den Quantenwissenschaften},
  author    = {Litinski, Daniel},
  year      = {2019},
  month     = mar,
  pages     = {128}
}

@article{Brown2017poking,
  title     = {{Poking holes and cutting corners to achieve Clifford gates with the surface code}},
  volume    = {7},
  pages     = 021029,
  doi       = {10.1103/physrevx.7.021029},
  number    = {2},
  journal   = {Phys. Rev. X},
  publisher = {American Physical Society (APS)},
  author    = {Brown, Benjamin J. and Laubscher, Katharina and Kesselring, Markus S. and Wootton, James R.},
  year      = {2017},
  month     = may
}

@article{Boundaries,
  title   = {The boundaries and twist defects of the color code and their applications to topological quantum computation},
  author  = {M. S. Kesselring and F. Pastawski and J. Eisert and B. J. Brown},
  journal = {Quantum},
  volume  = 2,
  pages   = 101,
  year    = 2018,
  doi     = {10.22331/q-2018-10-19-101}
}

@misc{bombin2018error-propagation,
  title                 = {{Transversal gates and error propagation in 3D topological codes}},
  author                = {H. Bombin},
  year                  = {2018},
  eprint                = {1810.09575},
  archiveprefix         = {arXiv},
  optoptoptprimaryclass = {quant-ph},
  url                   = {https://arxiv.org/abs/1810.09575}
}

@misc{surti2025logicalmagic,
  title                 = {Efficient simulation of logical magic state preparation protocols},
  author                = {Samyak Surti and Lucas Daguerre and Isaac H. Kim},
  year                  = {2025},
  eprint                = {2512.23799},
  archiveprefix         = {arXiv},
  optoptoptprimaryclass = {quant-ph},
  url                   = {https://arxiv.org/abs/2512.23799}
}

@article{Gidney2021stim,
  title     = {Stim: a fast stabilizer circuit simulator},
  volume    = {5},
  issn      = {2521-327X},
  url       = {http://dx.doi.org/10.22331/q-2021-07-06-497},
  doi       = {10.22331/q-2021-07-06-497},
  journal   = {Quantum},
  publisher = {Verein zur Forderung des Open Access Publizierens in den Quantenwissenschaften},
  author    = {Gidney, Craig},
  year      = {2021},
  month     = jul,
  pages     = {497}
}

@misc{huang2026hybridsurgery,
  title                 = {{Hybrid lattice surgery: Non-Clifford gates via non-Abelian surface codes}},
  author                = {Sheng-Jie Huang and Alison Warman and Sakura Schafer-Nameki and Yanzhu Chen},
  year                  = {2026},
  eprint                = {2510.20890},
  archiveprefix         = {arXiv},
  optoptoptprimaryclass = {quant-ph},
  url                   = {https://arxiv.org/abs/2510.20890}
}

@misc{warman2026clifford-hierarchy,
  title                 = {{Transversal Clifford-hierarchy gates via non-Abelian surface codes}},
  author                = {Alison Warman and Sakura Schafer-Nameki},
  year                  = {2026},
  eprint                = {2512.13777},
  archiveprefix         = {arXiv},
  optoptoptprimaryclass = {quant-ph},
  url                   = {https://arxiv.org/abs/2512.13777}
}

@misc{huang2025magicstates,
  title                 = {{Generating logical magic states with the aid of non-Abelian topological order}},
  author                = {Sheng-Jie Huang and Yanzhu Chen},
  year                  = {2025},
  eprint                = {2502.00998},
  archiveprefix         = {arXiv},
  optoptoptprimaryclass = {quant-ph},
  url                   = {https://arxiv.org/abs/2502.00998}
}

@article{pastawski2015nogologic,
  title     = {Fault-tolerant logical gates in quantum error-correcting codes},
  author    = {Pastawski, Fernando and Yoshida, Beni},
  journal   = {Phys. Rev. A},
  volume    = {91},
  issue     = {1},
  pages     = {012305},
  numpages  = {11},
  year      = {2015},
  month     = {Jan},
  publisher = {American Physical Society},
  doi       = {10.1103/PhysRevA.91.012305},
  url       = {https://link.aps.org/doi/10.1103/PhysRevA.91.012305}
}

@article{BravyiKoenig,
  title     = {Classification of Topologically Protected Gates for Local Stabilizer Codes},
  author    = {Bravyi, Sergey and K\"onig, Robert},
  journal   = {Phys. Rev. Lett.},
  volume    = {110},
  issue     = {17},
  pages     = {170503},
  numpages  = {5},
  year      = {2013},
  month     = {Apr},
  publisher = {American Physical Society},
  doi       = {10.1103/PhysRevLett.110.170503},
  url       = {https://link.aps.org/doi/10.1103/PhysRevLett.110.170503}
}

@misc{CaltechRSA,
  year          = {2026},
  eprint        = {2603.28627},
  archiveprefix = {arXiv},
  title         = {{Shor's algorithm is possible with as few as 10,000 reconfigurable atomic qubits}},
  author        = {Madelyn Cain and Qian Xu and Robbie King and Lewis R. B. Picard and Harry Levine and Manuel Endres and John Preskill and Hsin-Yuan Huang and Dolev Bluvstein}
}

@misc{GidneyRSA,
  year          = {2025},
  eprint        = {2505.15917},
  archiveprefix = {arXiv},
  title         = {{How to factor 2048 bit RSA integers with less than a million noisy qubits}},
  author        = {C. Gidney}
}

@misc{IcebergRSA,
  year          = {2026},
  eprint        = {2602.11457},
  archiveprefix = {arXiv},
  title         = {{The Pinnacle Architecture: Reducing the cost of breaking RSA-2048 to 100,000 physical qubits using quantum LDPC codes}},
  author        = {Paul Webster and Lucas Berent and Omprakash Chandra and Evan T. Hockings and Nouédyn Baspin and Felix Thomsen and Samuel C. Smith and Lawrence Z. Cohen}
}

@article{OConnor2018disjointness,
  title     = {Disjointness of Stabilizer Codes and Limitations on Fault-Tolerant Logical Gates},
  volume    = {8},
  pages     = 021047,
  doi       = {10.1103/physrevx.8.021047},
  journal   = {Phys. Rev. X},
  publisher = {American Physical Society (APS)},
  author    = {Jochym-O’Connor, Tomas and Kubica, Aleksander and Yoder, Theodore J.},
  year      = {2018},
  month     = may
}

@misc{bombin2016dimensionaljump,
  title                 = {Dimensional Jump in Quantum Error Correction},
  author                = {H. Bombin},
  year                  = {2016},
  eprint                = {1412.5079},
  archiveprefix         = {arXiv},
  optoptoptprimaryclass = {quant-ph},
  url                   = {https://arxiv.org/abs/1412.5079}
}

@article{beverland2021costuniversal,
  title     = {Cost of Universality: A Comparative Study of the Overhead of State Distillation and Code Switching with Color Codes},
  author    = {Beverland, Michael E. and Kubica, Aleksander and Svore, Krysta M.},
  journal   = {PRX Quantum},
  volume    = {2},
  issue     = {2},
  pages     = {020341},
  numpages  = {46},
  year      = {2021},
  month     = {Jun},
  publisher = {American Physical Society},
  doi       = {10.1103/PRXQuantum.2.020341},
  url       = {https://link.aps.org/doi/10.1103/PRXQuantum.2.020341}
}

@article{Aaronson2004simulation,
  title     = {Improved simulation of stabilizer circuits},
  author    = {Aaronson, Scott and Gottesman, Daniel},
  journal   = {Phys. Rev. A},
  volume    = {70},
  issue     = {5},
  pages     = {052328},
  numpages  = {14},
  year      = {2004},
  month     = {Nov},
  publisher = {American Physical Society},
  doi       = {10.1103/PhysRevA.70.052328},
  url       = {https://link.aps.org/doi/10.1103/PhysRevA.70.052328}
}

@article{Bravyi2012low-overheadMSD,
  title     = {Magic-state distillation with low overhead},
  volume    = {86},
  pages     = 052329,
  doi       = {10.1103/physreva.86.052329},
  number    = {5},
  journal   = {Phys. Rev. A},
  publisher = {American Physical Society (APS)},
  author    = {Bravyi, Sergey and Haah, Jeongwan},
  year      = {2012},
  month     = nov
}

@article{Litinski2019notascostly,
  title     = {Magic State Distillation: Not as Costly as You Think},
  volume    = {3},
  issn      = {2521-327X},
  url       = {http://dx.doi.org/10.22331/q-2019-12-02-205},
  doi       = {10.22331/q-2019-12-02-205},
  journal   = {Quantum},
  publisher = {Verein zur Forderung des Open Access Publizierens in den Quantenwissenschaften},
  author    = {Litinski, Daniel},
  year      = {2019},
  month     = dec,
  pages     = {205}
}

@article{Bravyi2005MSD,
  title     = {{Universal quantum computation with ideal Clifford gates and noisy ancillas}},
  author    = {Bravyi, Sergey and Kitaev, Alexei},
  journal   = {Phys. Rev. A},
  volume    = {71},
  issue     = {2},
  pages     = {022316},
  numpages  = {14},
  year      = {2005},
  month     = {Feb},
  publisher = {American Physical Society},
  doi       = {10.1103/PhysRevA.71.022316},
  url       = {https://link.aps.org/doi/10.1103/PhysRevA.71.022316}
}

@article{Lee2025MSDColor,
  title     = {Low-Overhead Magic State Distillation with Color Codes},
  volume    = {6},
  doi       = {10.1103/ch5r-cnfq},
  pages     = 030317,
  journal   = {PRX Quantum},
  publisher = {American Physical Society (APS)},
  author    = {Lee, Seok-Hyung and Thomsen, Felix and Fazio, Nicholas and Brown, Benjamin J. and Bartlett, Stephen D.},
  year      = {2025},
  month     = jul
}

@article{Gidney2019efficientmagicstate,
  doi       = {10.22331/q-2019-04-30-135},
  url       = {https://doi.org/10.22331/q-2019-04-30-135},
  title     = {Efficient magic state factories with a catalyzed {$|CCZ\rangle$} to {$2|T\rangle$} transformation},
  author    = {Gidney, Craig and Fowler, Austin G.},
  journal   = {{Quantum}},
  issn      = {2521-327X},
  publisher = {{Verein zur F{\"{o}}rderung des Open Access Publizierens in den Quantenwissenschaften}},
  volume    = {3},
  pages     = {135},
  month     = apr,
  year      = {2019}
}

@misc{gidney2019autoCCZ,
  title              = {{Flexible layout of surface code computations using AutoCCZ states}},
  author             = {Craig Gidney and Austin G. Fowler},
  year               = {2019},
  eprint             = {1905.08916},
  archiveprefix      = {arXiv},
  optoptprimaryclass = {quant-ph},
  url                = {https://arxiv.org/abs/1905.08916}
}

@misc{gidney2024cultivation,
  title                 = {{Magic state cultivation: growing $T$ states as cheap as CNOT gates}},
  author                = {Craig Gidney and Noah Shutty and Cody Jones},
  year                  = {2024},
  eprint                = {2409.17595},
  archiveprefix         = {arXiv},
  optoptoptprimaryclass = {quant-ph},
  url                   = {https://arxiv.org/abs/2409.17595}
}

@article{Psi2024postselection,
  title     = {Fault-Tolerant Postselection for Low-Overhead Magic State Preparation},
  author    = {Bombin, Hector and Pant, Mihir and Roberts, Sam and Seetharam, Karthik I.},
  journal   = {PRX Quantum},
  volume    = {5},
  issue     = {1},
  pages     = {010302},
  numpages  = {19},
  year      = {2024},
  month     = {Jan},
  publisher = {American Physical Society},
  doi       = {10.1103/PRXQuantum.5.010302},
  url       = {https://link.aps.org/doi/10.1103/PRXQuantum.5.010302}
}

@article{Scruby2022nonPauli,
  title     = {{Non-Pauli errors in the three-dimensional surface code}},
  doi       = {10.1103/physrevresearch.4.043052},
  pages     = 043052,
  volume    = {4},
  journal   = {Phys. Rev. Res.},
  publisher = {American Physical Society (APS)},
  author    = {Scruby, Thomas R. and Vasmer, Michael and Browne, Dan E.},
  year      = {2022}
}

@article{SteaneCode,
  author  = {A. M. Steane},
  journal = {Phys. Rev. Lett.},
  title   = {Error correcting codes in quantum theory},
  year    = {1996},
  pages   = {793},
  doi     = {10.1103/PhysRevLett.77.793},
  volume  = {77}
}

@article{shor1995QEC,
  title     = {Scheme for reducing decoherence in quantum computer memory},
  author    = {Shor, Peter W.},
  journal   = {Phys. Rev. A},
  volume    = {52},
  issue     = {4},
  pages     = {R2493--R2496},
  numpages  = {0},
  year      = {1995},
  month     = {Oct},
  publisher = {American Physical Society},
  doi       = {10.1103/PhysRevA.52.R2493},
  url       = {https://link.aps.org/doi/10.1103/PhysRevA.52.R2493}
}

@inproceedings{Shor1996,
  title     = {Fault-tolerant Quantum Computation},
  author    = {Shor, P. W.},
  booktitle = {Proceedings of the 37th Annual Symposium on Foundations of Computer Science},
  year      = {1996},
  address   = {Washington, DC, USA},
  pages     = {56--},
  publisher = {IEEE Computer Society},
  series    = {FOCS '96},
  acmid     = {875509},
  opturl    = {http://dl.acm.org/citation.cfm?id=874062.875509}
}

@article{Roads,
  title   = {Roads towards fault-tolerant universal quantum computation},
  author  = {E. T. Campbell and B. M. Terhal and C. Vuillot},
  doi     = {10.1038/nature23460},
  journal = {Nature},
  volume  = 549,
  pages   = {172-179},
  year    = 2017
}

@article{terhal2015review,
  title     = {Quantum error correction for quantum memories},
  author    = {Terhal, B. M.},
  journal   = {Rev. Mod.  Phys.},
  volume    = {87},
  pages     = {307--346},
  numpages  = {40},
  year      = {2015},
  publisher = {American Physical Society},
  doi       = {10.1103/RevModPhys.87.307}
}

@article{Lukin2024QEC,
  title   = {Logical quantum processor based on reconfigurable atom arrays},
  author  = {Dolev Bluvstein and Simon J. Evered and Alexandra A. Geim and Sophie H. Li and Hengyun Zhou and Tom Manovitz and Sepehr Ebadi and Madelyn Cain and Marcin Kalinowski and Dominik Hangleiter and J. Pablo {Bonilla Ataides} and Nishad Maskara and Iris Cong and Xun Gao and Pedro Sales Rodriguez and Thomas Karolyshyn and Giulia Semeghini and Michael J. Gullans and Markus Greiner and Vladan Vuletic and Mikhail D. Lukin},
  journal = {Nature},
  volume  = 626,
  pages   = {58-65},
  year    = 2024,
  doi     = {10.1038/s41586-023-06927-3}
}

@article{Google2025Threshold,
  title   = {Quantum error correction below the surface code threshold},
  author  = {{Google Quantum AI and collaborators}},
  doi     = {10.1038/s41586-024-08449-y},
  journal = {Nature},
  volume  = 638,
  pages   = {920-926},
  year    = 2025
}

@misc{Quantinuum2024demonstration,
  title                 = {Demonstration of logical qubits and repeated error correction with better-than-physical error rates},
  author                = {A. Paetznick and M. P. da Silva and C. Ryan-Anderson and J. M. Bello-Rivas and J. P. Campora III and A. Chernoguzov and J. M. Dreiling and C. Foltz and F. Frachon and J. P. Gaebler and T. M. Gatterman and L. Grans-Samuelsson and D. Gresh and D. Hayes and N. Hewitt and C. Holliman and C. V. Horst and J. Johansen and D. Lucchetti and Y. Matsuoka and M. Mills and S. A. Moses and B. Neyenhuis and A. Paz and J. Pino and P. Siegfried and A. Sundaram and D. Tom and S. J. Wernli and M. Zanner and R. P. Stutz and K. M. Svore},
  year                  = {2024},
  eprint                = {2404.02280},
  archiveprefix         = {arXiv},
  optoptoptprimaryclass = {quant-ph},
  url                   = {https://arxiv.org/abs/2404.02280}
}

@article{Daguerre2025Demonstration,
  title     = {Experimental Demonstration of High-Fidelity Logical Magic States from Code Switching},
  volume    = {15},
  doi       = {10.1103/dck4-x9c2},
  journal   = {Phys. Rev. X},
  publisher = {American Physical Society (APS)},
  pages     = 041008,
  author    = {Daguerre, Lucas and Blume-Kohout, Robin and Brown, Natalie C. and Hayes, David and Kim, Isaac H.},
  year      = {2025}
}

@misc{rosenfeld2025cultivation,
  title                 = {Magic state cultivation on a superconducting quantum processor},
  author                = {{Google Quantum AI and collaborators}},
  year                  = {2025},
  eprint                = {2512.13908},
  archiveprefix         = {arXiv},
  optoptoptprimaryclass = {quant-ph},
  url                   = {https://arxiv.org/abs/2512.13908}
}

@article{Bluvstein2025architecture,
  title     = {A fault-tolerant neutral-atom architecture for universal quantum computation},
  volume    = {649},
  issn      = {1476-4687},
  url       = {http://dx.doi.org/10.1038/s41586-025-09848-5},
  doi       = {10.1038/s41586-025-09848-5},
  number    = {8095},
  journal   = {Nature},
  publisher = {Springer Science and Business Media LLC},
  author    = {Bluvstein, Dolev and Geim, Alexandra A. and Li, Sophie H. and Evered, Simon J. and Bonilla Ataides, J. Pablo and Baranes, Gefen and Gu, Andi and Manovitz, Tom and Xu, Muqing and Kalinowski, Marcin and Majidy, Shayan and Kokail, Christian and Maskara, Nishad and Trapp, Elias C. and Stewart, Luke M. and Hollerith, Simon and Zhou, Hengyun and Gullans, Michael J. and Yelin, Susanne F. and Greiner, Markus and Vuletić, Vladan and Cain, Madelyn and Lukin, Mikhail D.},
  year      = {2025},
  month     = nov,
  pages     = {39–46}
}

@article{Vuillot2019toricGKP,
  title     = {{Quantum error correction with the toric Gottesman-Kitaev-Preskill code}},
  volume    = {99},
  doi       = {10.1103/physreva.99.032344},
  pages     = 032344,
  journal   = {Phys. Rev. A},
  publisher = {American Physical Society (APS)},
  author    = {Vuillot, Christophe and Asasi, Hamed and Wang, Yang and Pryadko, Leonid P. and Terhal, Barbara M.},
  year      = {2019}
}

@article{Wang2003statmech,
  title     = {{Confinement-Higgs transition in a disordered gauge theory and the accuracy threshold for quantum memory}},
  volume    = {303},
  doi       = {10.1016/s0003-4916(02)00019-2},
  number    = {1},
  journal   = {Ann. Phys.},
  publisher = {Elsevier BV},
  author    = {Wang, Chenyang and Harrington, Jim and Preskill, John},
  year      = {2003},
  pages     = {31–58}
}

@article{Kubica2023colorcodedecoders,
  title     = {Efficient color code decoders in $d$ dimensions from toric code decoders},
  volume    = {7},
  issn      = {2521-327X},
  url       = {http://dx.doi.org/10.22331/q-2023-02-21-929},
  doi       = {10.22331/q-2023-02-21-929},
  journal   = {Quantum},
  publisher = {Verein zur Forderung des Open Access Publizierens in den Quantenwissenschaften},
  author    = {Kubica, Aleksander and Delfosse, Nicolas},
  year      = {2023},
  month     = feb,
  pages     = {929}
}

@article{xy_floquet,
  author  = {Andreas Bauer},
  title   = {{The $x+y$ Floquet code: A simple example for topological quantum computation in the path integral approach}},
  year    = {2025},
  journal = {Phys. Rev. A},
  volume  = {111},
  doi     = {10.1103/PhysRevA.111.032413},
  pages   = {032413}
}

@article{Dijkgraaf1990,
  author  = {Dijkgraaf, R. and Witten, E.},
  title   = {Topological Gauge Theories and Group Cohomology},
  year    = {1990},
  journal = {Commun. Math. Phys.},
  volume  = {129},
  pages   = {393-429},
  doi     = {10.1007/BF02096988}
}

@phdthesis{Propitius1995,
  author        = {Propitius, Mark de Wild},
  title         = {Topological interactions in broken gauge theories},
  archiveprefix = {arXiv},
  eprint        = {hep-th/9511195},
  school        = {University of Amsterdam},
  url           = {http://arxiv.org/abs/hep-th/9511195},
  year          = {1995}
}

@misc{thesis,
  author = {A. Bauer},
  title  = {Topological Phases, Fixed-Point Models, Extended TQFT, and Fault-tolerant quantum computation-tensors in spacetime},
  note   = {{PhD thesis, Freie Universit\"at Berlin}},
  year   = {2023},
  url    = {https://refubium.fu-berlin.de/bitstream/handle/fub188/45937/Dissertation_Andreas_Bauer.pdf?sequence=4}
}

@article{Ni2015XS,
  title     = {A non-commuting stabilizer formalism},
  volume    = {56},
  pages     = 052201,
  doi       = {10.1063/1.4920923},
  journal   = {J. Math. Phys.},
  publisher = {AIP Publishing},
  author    = {Ni,  Xiaotong and Buerschaper,  Oliver and Van den Nest,  Maarten},
  year      = {2015}
}

@misc{PsiMacromux,
  title              = {Macromux: scalable postselection for high-threshold fault-tolerant quantum computation},
  author             = {Patrick Birchall and Jacob Bridgeman and Christopher Dawson and Terry Farrelly and Yehua Liu and Naomi Nickerson and Mihir Pant and Sam Roberts and Karthik Seetharam and David Tuckett},
  year               = {2026},
  eprint             = {2603.04875},
  archiveprefix      = {arXiv},
  optoptprimaryclass = {quant-ph},
  url                = {https://arxiv.org/abs/2603.04875}
}

@misc{manjunath2026GSC,
  title              = {Universal quantum computation with group surface codes},
  author             = {Naren Manjunath and Vieri Mattei and Apoorv Tiwari and Tyler D. Ellison},
  year               = {2026},
  eprint             = {2603.05502},
  archiveprefix      = {arXiv},
  optoptprimaryclass = {quant-ph},
  url                = {https://arxiv.org/abs/2603.05502}
}

@misc{sorge_pyfssa_2015,
  author    = {Andreas Sorge},
  title     = {pyfssa 0.7.6},
  month     = dec,
  year      = 2015,
  publisher = {Zenodo},
  doi       = {10.5281/zenodo.35173}
}

@misc{melchert_autoscale_2009,
  title           = {autoScale.py - A program for automatic finite-size scaling analyses},
  author          = {Oliver Melchert},
  year            = {2009},
  eprint          = {0910.5403},
  archiveprefix   = {arXiv},
  optprimaryclass = {physics.comp-ph}
}

@article{path_integral_qec,
  author  = {Andreas Bauer},
  title   = {Topological error correcting processes from fixed-point path integrals},
  year    = {2024},
  journal = {Quantum},
  volume  = 8,
  pages   = {1288},
  doi     = {10.22331/q-2024-03-20-1288}
}

@article{Steenrod1947,
  issn      = {0003486X},
  turl      = {http://www.jstor.org/stable/1969172},
  author    = {N. E. Steenrod},
  journal   = {Ann. Math.},
  number    = {2},
  pages     = {290--320},
  publisher = {Ann. of Math.},
  title     = {Products of Cocycles and Extensions of Mappings},
  volume    = {48},
  year      = {1947},
  doi       = {10.2307/1969172}
}

@article{Chen2021,
  author        = {Yu-An Chen and Sri Tata},
  title         = {Higher cup products on hypercubic lattices: application to lattice models of topological phases},
  year          = {2023},
  archiveprefix = {arXiv},
  eprint        = {2106.05274},
  journal       = {J. Math. Phys.},
  volume        = {64},
  pages         = {091902},
  doi           = {10.1063/5.0095189}
}

@misc{quadratic_tensors,
  author        = {Andreas Bauer and Seth Lloyd},
  title         = {{Quadratic tensors as a unification of Clifford, Gaussian, and free-fermion physics}},
  year          = {2026},
  archiveprefix = {arXiv},
  eprint        = {2601.15396}
}

@misc{hetenyi2026cultivationgauging,
  title              = {Constant depth magic state cultivation with Clifford measurements by gauging},
  author             = {Bence Hetényi and Benjamin J. Brown and Dominic J. Williamson},
  year               = {2026},
  eprint             = {2603.05429},
  archiveprefix      = {arXiv},
  optoptprimaryclass = {quant-ph},
  url                = {https://arxiv.org/abs/2603.05429}
}

@misc{lake2025localQEC,
  title              = {Local active error correction from simulated confinement},
  author             = {Ethan Lake},
  year               = {2025},
  eprint             = {2510.08056},
  archiveprefix      = {arXiv},
  optoptprimaryclass = {quant-ph},
  url                = {https://arxiv.org/abs/2510.08056}
}

@misc{paletta2026localdecodertoriccode,
  title              = {Local decoder for the toric code with a high pseudo-threshold},
  author             = {Louis Paletta},
  year               = {2026},
  eprint             = {2603.02328},
  archiveprefix      = {arXiv},
  optoptprimaryclass = {quant-ph},
  url                = {https://arxiv.org/abs/2603.02328}
}

@article{Cui2017diagonal,
  title   = {{Diagonal gates in the Clifford hierarchy}},
  volume  = {95},
  doi     = {10.1103/physreva.95.012329},
  pages   = 012329,
  journal = {Phys. Rev. A},
  author  = {Cui, Shawn X. and Gottesman, Daniel and Krishna, Anirudh},
  year    = {2017}
}

@misc{breuckmann2025cupsgatesicohomology,
  title           = {Cups and Gates I: Cohomology invariants and logical quantum operations},
  author          = {Nikolas P. Breuckmann and Margarita Davydova and Jens N. Eberhardt and Nathanan Tantivasadakarn},
  year            = {2025},
  eprint          = {2410.16250},
  archiveprefix   = {arXiv},
  optprimaryclass = {quant-ph},
  url             = {https://arxiv.org/abs/2410.16250}
}

@article{xyzrubycode,
  title     = {$\mathrm{XYZ}$ Ruby Code: Making a Case for a Three-Colored Graphical Calculus for Quantum Error Correction in Spacetime},
  author    = {Magdalena de la Fuente, Julio C. and Old, Josias and Townsend-Teague, Alex and Rispler, Manuel and Eisert, Jens and M\"uller, Markus},
  journal   = {PRX Quantum},
  volume    = {6},
  issue     = {1},
  pages     = {010360},
  numpages  = {56},
  year      = {2025},
  month     = {Mar},
  publisher = {American Physical Society},
  doi       = {10.1103/PRXQuantum.6.010360},
  url       = {https://link.aps.org/doi/10.1103/PRXQuantum.6.010360}
}

@article{bbcodes,
  author        = {Sergey Bravyi and Andrew W. Cross and Jay M. Gambetta and Dmitri Maslov and Patrick Rall and Theodore J. Yoder},
  title         = {High-threshold and low-overhead fault-tolerant quantum memory},
  year          = {2024},
  archiveprefix = {arXiv},
  eprint        = {2308.07915},
  journal       = {Nature},
  volume        = {627},
  pages         = {778-782},
  doi           = {10.1038/s41586-024-07107-7}
}

@article{Bombin2023,
  author  = {Hector Bombin and Daniel Litinski and Naomi Nickerson and Fernando Pastawski and Sam Roberts},
  title   = {{Unifying flavors of fault tolerance with the ZX calculus}},
  year    = {2024},
  journal = {Quantum},
  doi     = {10.22331/q-2024-06-18-1379},
  volume  = {8},
  pages   = {1379}
}

@misc{rodatz2024floquetifying,
  title           = {Floquetifying stabiliser codes with distance-preserving rewrites},
  author          = {Benjamin Rodatz and Boldizsár Poór and Aleks Kissinger},
  year            = {2024},
  eprint          = {2410.17240},
  archiveprefix   = {arXiv},
  optprimaryclass = {quant-ph},
  url             = {https://arxiv.org/abs/2410.17240}
}

@misc{kissinger2022phasefreezx,
  title           = {{Phase-free ZX diagrams are CSS codes (...or how to graphically grok the surface code)}},
  author          = {Aleks Kissinger},
  year            = {2022},
  eprint          = {2204.14038},
  archiveprefix   = {arXiv},
  optprimaryclass = {quant-ph},
  url             = {https://arxiv.org/abs/2204.14038}
}

@article{Delfosse2021unionfind,
  title     = {Almost-linear time decoding algorithm for topological codes},
  volume    = {5},
  issn      = {2521-327X},
  url       = {http://dx.doi.org/10.22331/q-2021-12-02-595},
  doi       = {10.22331/q-2021-12-02-595},
  journal   = {Quantum},
  publisher = {Verein zur Forderung des Open Access Publizierens in den Quantenwissenschaften},
  author    = {Delfosse, Nicolas and Nickerson, Naomi H.},
  year      = {2021},
  month     = dec,
  pages     = {595}
}

@article{Higgott2023BPmatching,
  title     = {Improved Decoding of Circuit Noise and Fragile Boundaries of Tailored Surface Codes},
  volume    = {13},
  pages     = 031007,
  doi       = {10.1103/PhysRevX.13.031007},
  journal   = {Phys. Rev. X},
  publisher = {American Physical Society (APS)},
  author    = {Higgott, Oscar and Bohdanowicz, Thomas C. and Kubica, Aleksander and Flammia, Steven T. and Campbell, Earl T.},
  year      = {2023}
}

@article{Jing2026heraldingD4,
  title     = {{Intrinsic Heralding and optimal decoders for non-Abelian topological order}},
  volume    = {136},
  doi       = {10.1103/ccj7-ctd8},
  pages     = 120405,
  journal   = {Phys. Rev. Lett.},
  publisher = {American Physical Society (APS)},
  author    = {Jing, Dian and Sala, Pablo and Jiang, Liang and Verresen, Ruben},
  year      = {2026}
}

@article{Hsin2025,
  title     = {{Non-Abelian self-correcting quantum memory and transversal non-Clifford gate beyond the $n/3$
               distance barrier}},
  volume    = {6},
  doi       = {10.1103/hkgn-hhqx},
  pages     = {040360},
  journal   = {PRX Quantum},
  publisher = {American Physical Society (APS)},
  author    = {Hsin, Po-Shen and Kobayashi, Ryohei and Zhu, Guanyu},
  year      = {2025}
}

\end{document}